\documentclass[11pt]{article}
\usepackage[utf8]{inputenc}
\usepackage[T1]{fontenc}
\usepackage{times}
\usepackage{fixltx2e}
\usepackage{tensor, multicol, multirow, comment}
\usepackage{stmaryrd}
\usepackage{graphicx, float, subcaption}
\usepackage{caption}
\usepackage{longtable}
\usepackage{setspace}
\usepackage{wrapfig}
\usepackage{soul}
\usepackage{pdfpages}
\usepackage{amsmath,amssymb,mathtools}
\usepackage{textcomp}
\usepackage{marvosym}
\usepackage{wasysym}
\usepackage{latexsym}
\usepackage[colorlinks=true, linkcolor=blue, citecolor=blue, urlcolor=blue,pdfborder={0 0 0 0}]{hyperref}
\usepackage{algorithm}
\usepackage{algpseudocode}
\usepackage{xcolor}
\usepackage[round, colon]{natbib}
\usepackage[margin=1in]{geometry}
\usepackage{enumerate,enumitem}
\usepackage{bigstrut}

\usepackage{lineno}
\usepackage[affil-it]{authblk}

\def\T{\top}

\input{./Definitions}
\begin{document} 
\title{CoMET: A Compressed Bayesian Mixed-Effects Model for High-Dimensional Tensors}
\author{Sreya Sarkar \footnote{PhD Candidate, Department of Statistics and Actuarial Science, The University of Iowa}, Kshitij Khare \footnote{Professor, Department of Statistics, University of Florida}, Sanvesh Srivastava \footnote{Associate Professor, Department of Statistics and Actuarial Science, The University of Iowa}}
\date{}

\maketitle

\begin{abstract}
Mixed-effects models are fundamental tools for analyzing clustered and repeated-measures data, but existing high-dimensional methods largely focus on penalized estimation with vector-valued covariates. 
Bayesian alternatives in this regime are limited, with no sampling-based mixed-effects framework that supports tensor-valued fixed- and random-effects covariates while remaining computationally tractable.
We propose the Compressed Mixed-Effects Tensor (CoMET) model for high-dimensional repeated-measures data with scalar responses and tensor-valued covariates. CoMET performs structured, mode-wise random projection of the random-effects covariance, yielding a low-dimensional covariance parameter that admits simple Gaussian prior specification and enables efficient imputation of compressed random-effects. For the mean structure, CoMET leverages a low-rank tensor decomposition and margin-structured Horseshoe priors to enable fixed-effects selection. These design choices lead to an efficient collapsed Gibbs sampler whose computational complexity grows approximately linearly with the tensor covariate dimensions. We establish high-dimensional theoretical guarantees by identifying regularity conditions under which CoMET’s posterior predictive risk decays to zero. Empirically, CoMET outperforms penalized competitors across a range of simulation studies and two benchmark applications involving facial-expression prediction and music emotion modeling.

\end{abstract}

Keywords: Gibbs sampling; mixed-effects model; random projection; tensor decomposition; uncertainty quantification

\section{Introduction}\label{Section:Intro}

Mixed-effects models are flexible extensions of regression models that include random-effects to capture within-cluster dependence. Recent high-dimensional literature has largely focused on penalization-based methods for vector-valued covariates, and, to our knowledge, there is no sampling-based Bayesian mixed-effects framework that accommodates tensor-valued fixed- and random-effect covariates in high dimensions. To fill this gap, we propose a Compressed Mixed-Effects Tensor (CoMET) model for sampling-based Bayesian inference in high-dimensional repeated measures data with a scalar response and tensor-valued covariates. CoMET uses random projection-based compression of the random-effects covariance structure induced by tensor covariates to reduce the effective dimension of the covariance parameter, bypassing the primary computational bottleneck in posterior computation. Together with a low-rank shrinkage prior on the fixed-effects coefficients, this approach yields an efficient collapsed Gibbs sampler for fixed-effects selection and prediction that outperforms penalized competitors across a range of simulated and real-data analyses. We also provide theoretical support for empirical results by establishing regularity conditions under which CoMET's posterior predictive risk decays to zero.


Modern vision and audio applications generate tensor-valued covariates with within-cluster dependence. For example, a celebrity in the Labeled Faces in the Wild (LFW) database has multiple facial images and associated expression intensity (e.g., smiling score) \citep{2016_LFWSurvey}. Images of the same celebrity exhibit similar celebrity relationships between facial features and expression intensity, motivating regression models with random-slope effects for tensor covariates. Similarly, in the MediaEval Database for Emotional Analysis in Music (DEAM), the valence or arousal of a song at a time point depends on the corresponding matrix-valued spectral acoustic representation (e.g., acoustic feature summaries and their first-order derivatives) \citep{2017_DEAM}. The valence or arousal scores within a song are correlated across time points, and the covariate-response relationship can vary across songs due to differences in spectral features. Tensor regression methods assume independence across observations, ignoring within-celebrity and within-song dependence. Mixed-effects methods fail to exploit tensor structure because they vectorize the covariates. CoMET fills this gap by extending random-slope mixed-effects modeling to tensor covariates while enabling sampling-based Bayesian inference.

There is an extensive literature on tensor regression for independent observations, which focuses on regressing a scalar response on a high-dimensional tensor-valued predictor. The main idea is to impose a low-rank structure on the regression coefficient tensor. Common notions of tensor rank include spectral rank, Tucker rank, and related multilinear ranks \citep{2009_Kolda}. Non-Bayesian methods estimate the regression coefficient tensor via penalized (quasi) likelihood  formulations, encouraging low-rank structure through spectral penalties, nuclear-norm relaxations, and other structured regularizers \citep{2013_Zhou_etal,2014_ZhouLi,2015_SlawskiRegularization,2017_ZhaoTraceReg,2018_Lock,2019_Raskutti,2019_ChenRaskutti,2025_PengTensorTrain}. Analogous Bayesian methods replace explicit penalties with structured priors on the coefficient tensor and perform posterior inference via Monte Carlo algorithms. These include multiway shrinkage priors and Gaussian priors on low-rank factorizations of the tensor coefficient \citep{2017_Guh_etal,2024_WangXu}. These approaches do not directly extend to mixed-effects models with tensor covariates, where repeated measurements induce within-subject dependence and inference requires estimating high-dimensional mean and covariance parameters.

The literature on repeated-measures data with multi-dimensional features is limited, especially in modeling within-subject dependence. \cite{2019_Zhang_etal} developed a scalar-on-tensor regression framework based on generalized estimating equations (GEE), imposing low-rank structure on the tensor regression coefficient. The GEE-based approach has two limitations.  First, it relies on a pre-specified working correlation structure to account for within-subject correlation. Second, it is not well suited to high-dimensional regimes because stable estimation requires a large number of observations per cluster. Mixed-effects models can also be applied to tensor-valued clustered data by vectorizing the tensor covariates. However, this strategy fails to exploit the multilinear structure and cross-mode dependence of a tensor and becomes computationally prohibitive as the tensor dimensions grow.

These limitations have motivated recent extensions of mixed-effects models that directly accommodate tensor covariates. \citet{2020_Yue_etal} proposed a tensor-on-tensor mixed model, but their framework has two limitations. First, it does not include tensor-valued random-effects covariates, so it cannot model random-slope effects. Second, it requires the mode-specific random-effects covariance matrices to be positive definite, which can be restrictive in high-dimensional settings. \cite{2020_Spencer_GuhPrado} regress a tensor-valued response on a scalar stimulus and a scalar subject-region-specific random intercept. This approach does not accommodate tensor-valued random slopes, so it fails to model within-subject dependence across tensor modes. Addressing these issues, \citet{2025_HulSri} introduced a mixed-effects model with tensor covariates that jointly estimates mean and covariance components under sparsity and low-rank assumptions. A Bayesian extension of this framework is impractical because it requires priors for high-dimensional mean and covariance parameters, leading to intractable posterior computation even at moderately large dimensions.

CoMET addresses these gaps through structured dimension reduction of the random-effects covariance. We employ mode-wise random projections to represent random-effects slopes in a low-dimensional space. Because the resulting compressed covariance parameter is low-dimensional, we place a default Gaussian prior on it, which simplifies posterior updates and enables efficient imputation of the compressed random-effects. Similar to the tensor regression literature, we assume a low-rank decomposition for the regression coefficient tensor, yielding linear scaling of the mean-parameter dimension. We then assign Horseshoe priors structured according to the low-r margins of the mean parameter. Together, these design choices yield an efficient collapsed Gibbs sampler, with computational complexity that scales approximately linearly in the tensor covariate dimensions. Finally, CoMET's setup enables us to establish that the its posterior predictive risk decays to zero in the high-dimensional regime.

Recently, compression of the random-effects covariance has been used to enable efficient Bayesian inference in high-dimensional linear mixed-effects models \citep{2025_SarKhSri}. Applying this approach to tensor covariates via vectorization has two limitations. First, vectorization necessitates placing shrinkage priors on a coefficient vector whose dimension grows polynomially with the tensor dimensions. CoMET avoids this by imposing a low-rank decomposition on the coefficient tensor and assigning shrinkage priors to its low-dimensional margins, yielding linear scaling of the mean-parameter dimension. Second, vectorization-based covariance compression fails to exploit the mode-specific structure of tensor-valued random-effects. CoMET instead imposes a separable covariance structure on the random-effects, enabling mode-wise compression via random projection matrices.

In summary, our contributions are threefold. First, we develop CoMET, a compressed mixed-effects model for clustered data with scalar responses and tensor-valued fixed- and random-effects covariates (Section \ref{Section:Methods}). Leveraging the tensor regression literature, we impose sparsity in the mean structure via a low-rank factorization and perform fixed-effects selection through Horseshoe priors structured according to the low-rank margin factors. Second, we compress the random-effects covariance by imposing a separable covariance structure and performing mode-wise compression via random projections, which yields a tensor mixed-effects model with random slopes for tensor-valued covariates. This substantially reduces the effective dimensions of the mean and covariance parameters and expands modeling flexibility relative to existing penalized and Bayesian alternatives for tensor regression, while keeping posterior computation tractable. We further provide theoretical support for these empirical findings by establishing regularity conditions under which CoMET's posterior predictive risk decays to zero (Section \ref{Section:tensorTheory}). Finally, we demonstrate across a range of simulation studies (Section \ref{Section:Simulations}) and two benchmark applications (a music database and an image database) that CoMET outperforms penalized competitors on multiple performance metrics (Section \ref{Section:RealDataAnalysis}).

\section{Methodology}\label{Section:Methods}
\subsection{Background}\label{Section:Background}

Consider the setup of a mixed model with scalar response and tensor-variate fixed- and random-effects covariates. Let there be $n$ subjects with $m_i$ observations recorded for the $i$-th subject and $N = \sum_{i=1}^n m_i$ be the total sample size. For $i=1,\dots,n$ and $ j=1,\dots,m_i$, suppose $\mathcal{X}_{ij} \in \mathbb{R}^{p_1 \times \dots \times p_D}$ and $\mathcal{Z}_{ij} \in \mathbb{R}^{q_1 \times \dots \times q_D}$ are $D$-th order tensors, respectively,  denoting the fixed- and random-effect covariates, and $y_{ij} \in \RR$ is the response for the $i$-th subject on the $j$-th occasion. A scalar-on-tensor mixed-effects model for relating $y_{ij}$ to $(\mathcal{X}_{ij}, \mathcal{Z}_{ij})$ is 
\begin{equation}\label{eq:Origmodel}
    y_{ij} = \langle \mathcal{X}_{ij}, \mathcal{B} \rangle + \langle \mathcal{Z}_{ij}, \mathcal{A}_{i} \rangle + \epsilon_{ij}, \quad \Acal_{ij} \perp \epsilon_{ij},\quad i = 1, \dots, n, \quad j = 1, \dots, m_i,
\end{equation}
where $\mathcal{B} \in \mathbb{R}^{p_1 \times \dots \times p_D}$ is a $D$-th order tensor representing the population-specific mean parameter, and $\langle \cdot, \cdot \rangle$ denotes the inner product between two tensors of the same order \citep{2025_HulSri}. The idiosyncratic  errors $\epsilon_{ij}$s are independently and identically distributed (iid) as $\mathcal{N}(0, \tau^2)$. The random-effect for subject $i$ is $\mathcal{A}_i \in \mathbb{R}^{q_1 \times \dots \times q_D}$, $\Acal_i$s and $\epsilon_{ij}$s are mutually independent, and $\Acal_i$s are iid as $\mathcal{N}_{q_1, \dots, q_D}(0, \tau^2 \Sigma_1, \dots, \Sigma_D)$, where $\mathcal{N}_{q_1, \dots, q_D}$ denotes an order-$D$ tensor normal distribution and $\Sigma_d$ is the symmetric positive semi-definite covariance matrix corresponding to mode-$d$ \citep{2011_Hoff}.

Tensor regression is recovered as a special case of \eqref{eq:Origmodel} by setting $\mathcal{A}_i = 0$ and $m_i = 1$ for all $i$. Existing regularization-based methods estimate $\mathcal{B}$ using penalties that promote low-rank structures, most commonly by expressing $\Bcal$ as a small number of rank-one tensors \citep{2013_Zhou_etal,2014_ZhouLi,2015_SlawskiRegularization,2017_ZhaoTraceReg,2018_Lock,2019_Raskutti,2019_ChenRaskutti,2025_PengTensorTrain}. 
Their Bayesian counterparts place shrinkage priors on the rank-one tensor components, typically shrinking a subset of these components. For example, multiway Dirichlet generalized double Pareto priors induce adaptive lasso-type shrinkage on the entries of $\mathcal{B}$ to promote rank reduction by simultaneously shrinking many rank-one components towards zero \citep{2017_Guh_etal, 2025_Spencer_Biostats}. These methods quantify uncertainty in $\mathcal{B}$ by drawing samples from its posterior distribution. 

After vectorization, the model in \eqref{eq:Origmodel} can be written as a linear mixed-effects (LME) model with a Kronecker-structured random-effects covariance. In particular, $D=1$ corresponds to the standard LME model, and its recent quasi-likelihood extensions replace $\Sigma_1$ with a suitable proxy matrix to estimate the fixed-effect parameter $\mathcal{B} \in \mathbb{R}^{p_1}$ using sparsity-inducing penalties \citep{FanLi12,Lietal21}. For a general $D$, let $x_{ij} = \text{vec}(\mathcal{X}_{ij}) \in \mathbb{R}^{p^*}$ and $z_{ij} = \text{vec}(\mathcal{Z}_{ij}) \in \mathbb{R}^{q^*}$ denote the vectorized covariates, where $p^* = \prod_{d=1}^D p_d$, $q^* = \prod_{d=1}^D q_d$, and $\text{vec}(\cdot)$ vectorizes a tensor by stacking its entries so that the first-mode index varies fastest, followed by the second, and so on. Using this representation, \eqref{eq:Origmodel} becomes an LME model with fixed- and random-effects covariates $(x_{ij}, z_{ij})$, fixed-effects parameter $\text{vec}(\Bcal) \in \RR^{p^*}$, and random-effect $\text{vec}(\Acal_i) \in \RR^{q^*}$ with covariance $\tau^2 (\Sigma_D \otimes \cdots \otimes \Sigma_1)$, where $\otimes$ denotes the Kronecker product \citep{2011_Hoff}. This implies that the mean and covariance parameter dimensions are $O(p^*)$ and $O(\sum_{d=1}^D q_d^2)$. While existing penalized and quasi-likelihood approaches for high-dimensional LME models can be applied to this vectorized form, they treat $(x_{ij}, z_{ij})$'s as unstructured vectors and fail to exploit the underlying tensor structure, leading to unnecessary computational burden and potentially suboptimal inference.

GEE-based models keep the mean structure in \eqref{eq:Origmodel} but model within-subject dependence through a marginal error term with covariance replaced by a scaled working correlation matrix. The marginal error plays the role of the random-effects and idiosyncratic errors in \eqref{eq:Origmodel}. Similar to tensor regression, GEE-based methods use penalties that promote low-rank structures for estimating $\mathcal{B}$ \citep{2019_Zhang_etal}.
The main limitation of these methods is that the working correlation structure is restricted to a small set of parametric forms (e.g., equicorrelation). Moreover, they cannot adapt to the dependence structure in the data because they do not include an equivalent of the random-effects term such as $\langle \mathcal{Z}_{ij}, \mathcal{A}_{i} \rangle$ in \eqref{eq:Origmodel}. They are also not well-suited to high-dimensional regimes with $N \ll \min\{p^*, q^*\}$.

Mixed model trace regression methods use \eqref{eq:Origmodel} to obtain sparse estimates of the mean parameter and low-rank estimates of the covariance components in high-dimensional settings \citep{2025_HulSri}. Conditional on $\Sigma_1, \ldots, \Sigma_D$, the regularized estimation of $\Bcal$ 
follows from standard tensor regression methods that impose low-rank or sparse structures on $\Bcal$.
The key additional step in the mixed-effects setting is to treat the random-effects $\Acal_i$s as latent variables and to model their separable covariance through a low-rank Tucker representation \citep{2009_Kolda}:
\begin{equation}\label{eq:Ai_dist2}
    \mathcal{A}_i = \mathcal{D}_i \times_1 L_1 \times_2 \cdots \times_D L_D, \quad \mathcal{D}_i \sim \mathcal{N}_{q_1, \dots, q_D}(0, \tau^2 I_{q_1}, \dots, I_{q_D}), \quad d = 1, \dots, D,
\end{equation}
where $L_d$ is a (possibly low-rank) square-root factor of the mode-$d$ covariance matrix $\Sigma_d$, satisfying $\Sigma_d = L_d L_d^{\T}$. The covariance components $\Sigma_d$'s are estimated via their factors $L_d$'s using an EM-type algorithm: the E-step uses the Gaussian conditional distribution of $\Acal_i$ and the M-step sequentially updates $L_1, \ldots, L_D$ under a group-lasso penalty on their columns. Despite their flexibility, this approach becomes computationally prohibitive in high-dimensional applications and does not provide uncertainty estimates.

A Bayesian extension of \eqref{eq:Origmodel} is computationally prohibitive in high dimensions. The main 
challenge is specifying priors on $\Sigma_1, \ldots, \Sigma_D$ that still yield tractable posterior computation. Standard high-dimensional priors for covariance matrices, including the shrinkage inverse Wishart prior \citep{2020_BergerSIW}, require $\Sigma_d$ to be positive definite and are not suited for estimating $O(q_d^2)$ covariance parameters when the number of observations per subject is limited. In the next section, we introduce the CoMET model, which addresses these limitations by using a quasi-likelihood based on mode-wise covariance compression to bypass high-dimensional covariance estimation. Together with a low-rank shrinkage prior on $\Bcal$, this formulation enables sampling-based posterior inference with uncertainty quantification. 

\subsection{CoMET: A Compressed Mixed-Effects Tensor Model}\label{Section:OurModel}

The key idea behind CoMET's mode-wise covariance compression is to replace each high-dimensional random-effects covariance component
with a low-dimensional proxy obtained by applying random projections to a covariance factor. Let $R_d, S_d \in \mathbb{R}^{k_d \times q_d}$ be mode-$d$-specific random matrices with independent entries distributed as $\mathcal{N}(0, 1/k_d)$, where $k_d \ll q_d$. For each mode $d = 1, \dots, D$, the CoMET model compresses the square-root covariance factor $L_d$ in \eqref{eq:Ai_dist2} by replacing it with $\Gamma_d = S_d L_d R_d^{\T} \in \mathbb{R}^{k_d\times k_d}$.

The CoMET model reformulates the tensor mixed-effects model in \eqref{eq:Origmodel} using the compressed factors $(\Gamma_1, \ldots, \Gamma_D)$. Following \eqref{eq:Ai_dist2}, the compressed random-slope tensor for subject $i$ is defined as
\begin{equation}\label{eq:compAi_dist}
    \tilde{\mathcal{A}}_i
    = \tilde{\mathcal{D}}_i \times_1 \Gamma_1 \times_2 \cdots \times_D \Gamma_D,
    \qquad
    \tilde{\mathcal{D}}_i \sim \mathcal{N}_{k_1,\ldots,k_D}\big(0,\ \tau^2 R_1R_1^{\T},\ldots, R_DR_D^{\T}\big),
\end{equation}
where $\tilde{\mathcal{D}}_i = \mathcal{D}_i \times_1 R_1 \times_2 \cdots \times_D R_D$ is the compressed core tensor. To match the dimension of $\tilde {\mathcal{A}}_i$ and its random-effects covariates, we define $\tilde{\mathcal{Z}}_{ij}
=\mathcal{Z}_{ij}\times_1 S_1 \times_2 \cdots \times_D S_D$.  While the fixed-effects covariates $\mathcal{X}_{ij}$'s remain unchanged, we impose a low-rank structure on their coefficient tensor $\mathcal{B}$ to reduce the effective mean parameter dimension. The resulting CoMET model is 
\begin{equation}\label{eq:cme1}
    y_{ij}
    = \langle \mathcal{X}_{ij}, \mathcal{B} \rangle
    + \langle \tilde{\mathcal{Z}}_{ij}, \tilde{\mathcal{A}}_i \rangle
    + \epsilon_{ij}, \quad \epsilon_{ij} \sim \Ncal(0, \tau^2),
    \quad j=1,\ldots,m_i,\quad i=1,\ldots,n ,
\end{equation}
where $\tau^2$ is the idiosyncratic error variance.  The compressed inner product $\langle \tilde{\mathcal{Z}}_{ij}, \tilde {\mathcal{A}}_i \rangle$ in \eqref{eq:cme1} is a low-dimensional proxy for $\langle \mathcal{Z}_{ij}, {\mathcal{A}}_i \rangle$ in \eqref{eq:Origmodel}. This approximation is achieved by mode-wise multiplication of $(S_1, \ldots, S_D)$ to the covariates $Z_{ij}$ in \eqref{eq:Origmodel} and $(R_1, \ldots, R_D)$ to the core random-effects tensor $\mathcal{D}_i$ in \eqref{eq:Ai_dist2}. 

Compared to \eqref{eq:Origmodel}, the CoMET model in \eqref{eq:cme1} reduces the mean and covariance parameter dimensions.  First, the compressed covariance factors $(\Gamma_1, \dots, \Gamma_D)$ reduce the covariance parameter dimension from $O(q_1^2 + \cdots + q_D^2)$ to $O( k_1^2 + \cdots + k_D^2)$, where $k_d \ll q_d$ for $d = 1, \dots, D$. Second, we impose a rank-$K$ CP structure on $\Bcal$  as
\begin{equation}\label{eq:cp}
    \mathcal{B} = \sum_{g=1}^{K} \bm \beta^{(g)}_{1} \circ \cdots \circ \bm \beta^{(g)}_{D}, \quad \bm \beta^{(g)}_{d} \in \mathbb{R}^{p_d}, \quad g = 1, \dots, K, \quad d = 1, \dots, D,
\end{equation}
where $K \ll \min\{ p_1, \ldots, p_D\}$ and $\circ$ denotes the vector outer product. Specifically, \eqref{eq:cp} implies that the $(j_1, \dots, j_D)$-th entry of $\mathcal{B}$ satisfies $\mathcal{B}_{j_1 \ldots j_D} = \sum_{g=1}^K \beta^{(g)}_{1 j_1}\cdots\beta^{(g)}_{D j_D}$ \citep{2009_Kolda}. For $d=1, \ldots, D$, let $B_d = \big[ \bm \beta^{(1)}_{d}, \cdots, \bm \beta^{(K)}_{d} \big] \in \mathbb{R}^{p_d \times K}$ be the mode-$d$ mean factor matrix. This parametrization replaces the $O(\prod_{d=1}^D p_d)$ free parameters in an unstructured $\Bcal$ with $O(K \sum_{d=1}^D p_d)$ parameters in $(B_1, \ldots, B_D)$. Overall, CoMET reduces the parameter dimension from $O(\prod_{j=1}^D p_j + \max_d q_d^2)$ in \eqref{eq:Origmodel} to 
$O(K \sum_{j=1}^D p_j + \max_d k_d^2)$.

We complete the CoMET model specification by assigning priors to $\{B_1, \ldots, B_D,  \Gamma_1, \dots, \Gamma_D, \tau^2\}$. Because $B_1, \ldots, B_D$ determine $\Bcal$ through the CP representation in \eqref{eq:cp}, we place Horseshoe priors on the entries of the factor matrices to promote shrinkage of the $\Bcal$ entries and fixed-effects selection \citep{2010_CarPolSco,2017_VanderPas_SzaboVaart,2022_BhatKharePal,2023_SonLia}. For $g = 1, \ldots, K$, $d = 1, \dots, D$, and $j = 1, \dots, p_d$, we specify prior on $(\Bcal, \tau^2)$ using the following hierarchy
\begin{align}\label{eq:hsprior_2}
    \beta^{(g)}_{dj} \mid \lambda_{gdj}^2, \delta_{g}^2, \tau^2 \overset{\text{ind}}{\sim} \mathcal{N}(0, \lambda_{gdj}^2 \delta_{g}^2 \tau^2 ),\quad \lambda_{gdj} \overset{\text{ind}}{\sim} \mathcal{C}^{+}(0, 1), \quad 
    \delta_{g} \overset{\text{ind}}{\sim} \mathcal{C}^{+}(0, 1), \quad 
    \tau^2 \sim \mathcal{IG}(a_0, b_0),
\end{align}
where $\mathcal{C}^{+}(0, 1)$ denotes the standard half-Cauchy distribution
and $\mathcal{IG}(a_0, b_0)$ denotes the inverse-gamma distribution with the shape and scale parameters $a_0$ and $b_0$, respectively. The global shrinkage parameter $\delta_{g}$ controls overall shrinkage across all modes within the $g$-th rank-one component in \eqref{eq:cp}, encouraging rank reduction by shrinking the entire component toward zero. The local shrinkage parameter $\lambda_{gdj}$ controls entry-specific shrinkage of $\beta_{dj}^{(g)}$, which impacts the shrinkage of the tensor entries $\mathcal{B}_{j_1 \ldots j_D}$ through the  multiplicative structure in \eqref{eq:cp} and promotes variable selection. Finally, independent of $(\mathcal{B}, \tau^2)$, we place Gaussian priors,  $\mathcal{N}(0, \sigma_{d}^2 I_{k_d^2})$, on $\gamma_d$ independently across $d$ modes, where $\gamma_d = \text{vec}(\Gamma_d)$ and $\sigma_{d}^2$ is a hyperparameter. 

\subsection{Fixed-Effects Selection and Prediction using CoMET}
\label{Section:Sampler}

We develop a collapsed Gibbs sampler for fixed-effects selection and prediction using the CoMET model. For simplicity, we present the sampler for three-dimensional tensors (i.e., $D = 3$), but the extension to a general order-$D$ tensor is straightforward. The algorithm cycles between two blocks. In the first bock, we  condition on $(\mathcal{B}, \tau^2)$ and update the compressed random-effects cores 
$\tilde {\mathcal{D}}_i$s in \eqref{eq:cme1}
given $(\Gamma_1, \Gamma_2, \Gamma_3)$  followed by updating $(\Gamma_1, \Gamma_2, \Gamma_3)$
given $\tilde {\mathcal{D}}_i$s. In the second block, we update $(\Bcal, \tau^2)$ conditional on 
$(\Gamma_1, \Gamma_2, \Gamma_3)$ after marginalizing out $\tilde {\mathcal{D}}_i$s in \eqref{eq:cme1}.

We derive the joint distribution of the responses and random-effects under \eqref{eq:cme1}, which is key to obtaining analytically tractable full conditional distributions for the Gibbs sampler. Let $\Gamma^* = \Gamma_3 \otimes \Gamma_2 \otimes \Gamma_1, \, R^* = R_3 \otimes R_2 \otimes R_1$, and $S^* = S_3 \otimes S_2 \otimes S_1$. Then, we can re-express the CoMET model in \eqref{eq:cme1} as
\begin{equation}\label{eq:cme2}
    y_{ij} = x_{ij}^{\T} \bm{\beta} + \tilde z_{ij}^{\T}\Gamma^* \tilde d_i + \epsilon_{ij}, \quad \tilde d_i \sim \mathcal{N}(0, \tau^2 R^*R^{*\T}),
\end{equation}
where $x_{ij} = \text{vec}(\mathcal{X}_{ij}), \, \bm{\beta} = \text{vec}(\mathcal{B}), \, \tilde z_{ij} = \text{vec}(\mathcal{\tilde Z}_{ij}) = S^* \text{vec}(\mathcal{Z}_{ij})$, and $\tilde d_i = \text{vec}(\tilde {\mathcal{D}}_i)$. For the $i$-th subject, we then stack $y_{ij}$s, $x_{ij}^{\T}$s, and $\tilde z_{ij}^{\T}$s row-wise to construct $y_i \in \mathbb{R}^{m_i}$, $ X_i \in \mathbb{R}^{m_i \times p_1 p_2 p_3}$, and $\tilde Z_i \in \mathbb{R}^{m_i \times k_1 k_2 k_3}$, respectively. The joint distribution of $(y_i, \tilde d_i)$ is Gaussian:
\begin{align}\label{eq:joint_yitildedi}
    &\mathcal{N}_{m_i + k_1 k_2 k_3}\left( \begin{pmatrix} X_i \bm\beta \\ 0 \end{pmatrix}, \tau^2 \begin{pmatrix}V_{y_i, y_i} & V_{y_i, \tilde d_i} \\ V_{y_i, \tilde d_i}^{\T} & V_{\tilde d_i, \tilde d_i}\end{pmatrix} \right), 
\end{align}
where $V_{\tilde d_i, \tilde d_i} = R^* R^{*\T}$, $V_{y_i, \tilde d_i} = \tilde Z_i \Gamma^* R^* R^{*\T}$, and  $V_{y_i, y_i} = \tilde Z_i \Gamma^* R^* R^{*\T}\Gamma^{*\T}\tilde Z_i^{\T} + I_{m_i}$.

For the first block, we condition on $(\mathcal{B},\tau^2)$ and derive the full conditional distributions for the compressed core random-effects $\tilde{\mathcal{D}}_i$s and the compressed covariance factors $(\Gamma_1,\Gamma_2,\Gamma_3)$. For $i=1, \ldots, n$, the joint distribution in \eqref{eq:joint_yitildedi} implies that the full conditional distribution of $\tilde d_i$ given $y_i$ and $(\Gamma_1,\Gamma_2,\Gamma_3)$  is
\begin{equation}\label{eq:fullcond_tildedi}
    \mathcal{N}_{k_1 k_2 k_3}(\mu_{\tilde d_i \mid .}, V_{\tilde d_i \mid .}), \ \mu_{\tilde d_i \mid .}  = V_{y_i, \tilde d_i}^{\T} V_{y_i,  y_i}^{-1} (y_i - X_i \bm\beta), \ V_{\tilde d_i \mid .} = V_{\tilde d_i, \tilde d_i} - V_{y_i, \tilde d_i}^{\T} V_{y_i, y_i}^{-1} V_{y_i, \tilde d_i}.
\end{equation}
Conditional on $(\Bcal, \tau^2)$, $\tilde d_i$s, and $y_i$s, we sample $\Gamma_1, \Gamma_2$, and $\Gamma_3$ sequentially. For any $d=1, 2, 3$, we have
\begin{align}\label{innerprod_gammad}
   \langle \mathcal{\tilde Z}_{ij}, \mathcal{\tilde A}_{i} \rangle = \text{tr}(\mathcal{\tilde Z}_{ij(d)}^{\T} \mathcal{\tilde A}_{i(d)}) & = \{\operatorname{vec}(\mathcal{\tilde Z}_{ij(d)} (\Gamma_3 \otimes \cdots \otimes \Gamma_{d+1} \otimes \Gamma_{d-1} \otimes \cdots \otimes \Gamma_1)\tilde {\mathcal{D}}_{i(d)}^{\T}) \}^{\T} \gamma_d,
\end{align}
where $\text{tr}(\cdot)$ is the trace operator and $\mathcal{\tilde Z}_{ij(d)}$ and $\tilde {\mathcal{D}}_{i(d)}$ are mode-$d$ matricization \citep{2009_Kolda} of the tensors $\mathcal{\tilde Z}_{ij}$ and $\tilde {\mathcal{D}}_{i}$, respectively. Substituting \eqref{innerprod_gammad} in \eqref{eq:cme1} yields the following linear regression model 
\begin{equation}\label{eq:regmodel_gammad}
    \check y = \check Z_{\gamma_d} \gamma_d + \epsilon, \quad \epsilon \sim \mathcal{N}(0, \tau^2 I_N), \quad \gamma_d \in \mathbb{R}^{k_d^2}, \quad \pi (\gamma_d)\sim \mathcal{N}(0, \sigma_{d}^2 I_{k_d^2}),
\end{equation}
where $\check y$ is obtained by stacking $y_{ij} - \langle \mathcal{X}_{ij}, \mathcal{B}\rangle$s row-wise and $\check Z_{\gamma_d} \in \mathbb{R}^{N \times k_d^2}$ is obtained by stacking $\{\text{vec}(\mathcal{\tilde Z}_{ij(d)} (\Gamma_3 \otimes \cdots \otimes \Gamma_{d+1} \otimes \Gamma_{d-1} \otimes \cdots \otimes \Gamma_1)\tilde {\mathcal{D}}_{i(d)}^{\T}) \}^{\T}$s  row-wise for $j = 1, \dots, m_i$ and $i = 1, \dots, n$. The Gaussian likelihood and prior on $\gamma_d$ in \eqref{eq:regmodel_gammad} implies that the full conditional distribution of $\gamma_d$ is
\begin{equation}\label{eq:post_gammad}
   \mathcal{N}(\mu_{\gamma_d}, \Sigma_{\gamma_d}), \quad \Sigma_{\gamma_d} = \left(\tau^{-2} \check{Z}_{\gamma_d}^{\T} \check{Z}_{\gamma_d} + \sigma_{d}^{-2} I_{k_d^2} \right)^{-1}, \quad \mu_{\gamma_d} = \tau^{-2}\Sigma_{\gamma_d} \check{Z}_{\gamma_d}^{\T} \check y, \quad  d = 1, \ldots, 3. 
\end{equation}
Using \eqref{eq:post_gammad}, we update $\Gamma_1$, $\Gamma_2$, and $\Gamma_3$ 
sequentially from their full conditional distributions, each conditioned on all remaining parameters.

Before deriving the full conditionals in the second block, we restructure the model in \eqref{eq:cme1} using \eqref{eq:cp} and identities based on the Khatri-Rao product \citep{2009_Kolda}. Let $\mathcal{X}_{ij(d)}$ and $\mathcal{B}_{(d)}$ respectively denote the mode-$d$ matricization of the tensors $\mathcal{X}_{ij}$ and $\mathcal{B}$. Then, leveraging the CP decomposition \eqref{eq:cp}, we can express the mode-wise matricizations of $\mathcal{B}$ in terms of the mode-specific factor matrices as follows
\begin{align}\label{eq:property_modeB}
    \mathcal{B}_{(1)} = B_1 (B_3 \odot B_2)^{\T} , \quad \mathcal{B}_{(2)} = B_2 (B_3 \odot B_1)^{\T}, \quad \mathcal{B}_{(3)} = B_3 (B_2 \odot B_1)^{\T} ,
\end{align}
where $\mathcal{B}_{(1)} \in \RR^{p_1 \times p_2 p_3}$, $\mathcal{B}_{(2)} \in \RR^{p_2 \times p_1 p_3}$, and $\mathcal{B}_{(3)}  \in \RR^{p_3 \times p_1 p_2}$ and the operator $\odot$ denotes the Khatri-Rao product. For instance,  
$B_3 \odot B_2 = [\bm \beta^{(1)}_{3} \otimes \bm \beta^{(1)}_{2} \quad \cdots \quad \bm \beta^{(K)}_{3} \otimes \bm \beta^{(K)}_{2}] \in \RR^{p_3 p_2 \times K}$, where $\bm \beta^{(1)}_{3}, \bm \beta^{(1)}_{2}, \ldots, \bm \beta^{(K)}_{3}, \bm \beta^{(K)}_{2}$ are defined in \eqref{eq:cp}. Using \eqref{eq:property_modeB}, the marginal mean component in \eqref{eq:cme1} satisfies
\begin{equation}\label{eq:property_innerprod_fixef}
    \langle \mathcal{X}_{ij}, \mathcal{B} \rangle = \text{tr}(\mathcal{X}_{ij(d)}^{\T} \mathcal{B}_{(d)}) = \{\text{vec}(\mathcal{X}_{ij(d)} B_{-d})\}^{\T} \tilde{\bm\beta}_d, \quad \tilde{\bm\beta}_d= \text{vec}(B_d), \quad  d = 1, 2, 3,
\end{equation}
where  $B_{-1} = B_3 \odot B_2, \, B_{-2} = B_3 \odot B_1$ and $B_{-3} = B_2 \odot B_1$. We re-express \eqref{eq:cme1} using \eqref{eq:property_innerprod_fixef} as
\begin{align}\label{eq:modelBd}
    y_{ij} & = \{\text{vec}(\mathcal{X}_{ij(d)}B_{-d})\}^{\T}\tilde{\bm\beta}_d + \tilde z_{ij}^{\T}\Gamma^* \tilde d_i + \epsilon_{ij}, \quad i=1, \ldots, n, \quad   j=1, \ldots, m_i.
\end{align}

We marginalize over $\tilde d_i$s in \eqref{eq:modelBd} and derive the full conditional distribution of $(B_1, B_2, B_3, \tau^2)$ given $(\Gamma_1, \Gamma_2, \Gamma_3)$ in the second block. After marginalization, \eqref{eq:modelBd} reduces to the following weighted regression model
\begin{equation}\label{eq:modelBd_marg}
y_i \sim \mathcal{N}_{m_i}(\tilde X_i \tilde{\bm\beta}_d, \tau^2 C_i), \quad C_i = V_{y_i, y_i} = \tilde Z_i \Gamma^* R^* R^{*\T}\Gamma^{*\T}\tilde Z_i^{\T} + I_{m_i}, \quad i = 1, \ldots, n,
\end{equation}
where $\tilde X_i \in \RR^{m_i \times K p_d}$ and $\{\text{vec}(\mathcal{X}_{ij(d)}B_{-d})\}^{\T}$ is the $j$-th row of $\tilde X_i$. We scale $y_i$ and $\tilde X_i$ to obtain $y_i^* = C_i^{-1/2} y_i$ and $\tilde X_i^* = C_i^{-1/2} \tilde X_i$, where $C_i^{-1/2}$ is the square root of $C_i^{-1}$. Finally, we stack $y_i^*$s and $\tilde X_i^*$s row-wise to obtain $y^* \in \mathbb{R}^N$ and $X_{B_d}^* \in \mathbb{R}^{N \times K p_d}$,  respectively, and obtain the following regression model 
\begin{equation}\label{eq:modelBd_cycle2}
    y^* = X_{B_d}^* \tilde{\bm\beta}_d + \epsilon^*, \quad \epsilon^* \sim \mathcal{N}(0, \tau^2 I_N).
\end{equation}
The mode-structured Horseshoe prior in \eqref{eq:hsprior_2} and the Gaussian likelihood in \eqref{eq:modelBd_cycle2} yield the following full conditional distribution of the mode-$d$ margins of $\Bcal$:
\begin{align}\label{eq:fullcondBd}
    \tilde{\bm\beta}_d \sim \mathcal{N}(\Sigma_{B_d}X_{B_d}^{*\T}y^{*}, &\tau^{2} \Sigma_{B_d}), \quad 
    \Sigma_{B_d} = \left[ X_{B_d}^{*\T}X_{B_d}^{*} + \Bigl\{\text{diag}(\delta_{1}^{2} \Lambda_{1d}, \dots, \delta_{K}^{2} \Lambda_{Kd}) \Bigr\}^{-1} \right]^{-1},
\end{align}
where $\Lambda_{gd} = \text{diag}(\lambda^{2}_{gd1}, \dots, \lambda^{2}_{gdp_d}) $ for $d=1, 2, 3$ and $g=1, \ldots, K$; therefore, after marginalization with respect to the compressed core random-effects, we use \eqref{eq:fullcondBd} to sequentially update the vectorized mean factor matrices $\tilde{\bm\beta}_1$, $\tilde{\bm\beta}_2$ and $\tilde{\bm\beta}_3$, conditional on the observed data and all the remaining parameters. 

The second block updates are finished by updating $\tau^2$ and  the local and global shrinkage parameters, $\lambda^2_{gdj}$s and $\delta_g$s,  given the observed data, factor margins $(B_1, B_2, B_3)$, and compressed factors $(\Gamma_1, \Gamma_2, \Gamma_3)$. The vectorized factor margins $\tilde{\bm\beta}_1$, $\tilde{\bm\beta}_2$ and $\tilde{\bm\beta}_3$ yield $\mathcal{B}$ using the CP structure in \eqref{eq:cp}. Given $\mathcal{B}$ and $(\Gamma_1, \Gamma_2, \Gamma_3)$, the Gaussian likelihood in \eqref{eq:cme2} implies that the full conditional distribution of $\tau^2$ is inverse-gamma.
Exploiting a parameter-expanded representation of the half-Cauchy prior in \eqref{eq:hsprior_2}, we obtain closed-form full conditional distributions for these shrinkage parameters, along with the associated auxiliary variables introduced through the reparameterization.  Detailed derivations of these conditional updates are provided in Section~\ref{Section:FullCondDerivations} of the supplementary material. Algorithm \ref{algoCoMET:1} summarizes the overall sampling scheme including the analytic forms of the full conditional distributions.

Given the posterior samples of the CoMET model's parameters, posterior prediction for new subjects proceeds as follows. Let $(\mathcal{B}^{(t)}, \tau^{2(t)}, \Gamma_1^{(t)}, \Gamma_2^{(t)}, \Gamma_3^{(t)})$ denote the marginal parameter chain at the $t$-th iteration of Algorithm \ref{algoCoMET:1}. For a test subject with $m^*$ observations and covariates $(X_{\text{new}}, Z_{\text{new}})$, the posterior predictive distribution at iteration $t$ is
\begin{equation}\label{eq:predict}
    y_{\text{new}}^{(t)} \sim \mathcal{N}(X_{\text{new}}\bm\beta^{(t)}, \tau^{2(t)}(Z_{\text{new}} S^{*\T} \Gamma^{* (t)} R^{*} R^{*\T}\Gamma^{* (t) \T}S^{*}Z_{\text{new}}^{\T} + I_{m^*})),
\end{equation}
where $\bm\beta^{(t)} = \text{vec}(\mathcal{B}^{(t)})$ and $\Gamma^{* (t)} = \Gamma_3^{(t)} \otimes \Gamma_2^{(t)} \otimes \Gamma_1^{(t)}$, and $R^*$ and $S^*$ consist of the same random projection matrices used during posterior sampling of the parameters. Finally, we average over these marginal posterior predictive draws to obtain point estimates of the test response.

For fixed-effects selection, we use the marginal chain $(\mathcal{B}_{j_1 j_2 j_3})_{t=1}^{\infty}$ to construct cell-wise credible intervals. These intervals are used for  uncertainty quantification at a desired nominal level. Additionally, we use the sequential 2-means algorithm \citep{LiPati17} to perform fixed-effects selection. Applying this procedure on the marginal chain $(\mathcal{B}_{j_1 j_2 j_3})_{t=1}^{\infty}$ partitions the fixed-effect coefficient tensor cells into two clusters corresponding to zero (inactive) and non-zero (active) cells. The selected fixed-effects corresponds to the indices of the non-zero cells.

\begin{algorithm}[h!]
\caption{\textbf{Collapsed Gibbs sampler for CoMET}}
\label{algoCoMET:1}
\begin{enumerate}\setlength{\itemsep}{0pt}
    \item Set $k_d = O(\log(\max_{d}(q_d)) )$, $ K = O( \log(\max_{d}(p_d)) )$, $a_0 = b_0 = 0.01$, $\sigma_d^2 = 1$ for $d = 1, 2, 3$.
    \item Draw $R_d = ((r_{ij})), \ r_{ij} \overset{\text{iid}}{\sim} \mathcal{N}(0, 1/k_d)$ and $S_d = ((s_{ij})), \ s_{ij} \overset{\text{iid}}{\sim} \mathcal{N}(0, 1/k_d)$, independently for $d = 1, 2, 3$.
    \item Initialize $(\tilde{\bm\beta}_1^{(0)}, \tilde{\bm\beta}_2^{(0)}, \tilde{\bm\beta}_3^{(0)}, \tau^{2(0)}, \gamma_1^{(0)}, \gamma_2^{(0)}, \gamma_3^{(0)})$. For superscript $t = 1, 2, \ldots, \infty$, the $t$-th iteration of the Gibbs sampler cycles through the following steps.
    \begin{enumerate}
        \item Given $(\tilde{\bm\beta}_1^{(t-1)}, \tilde{\bm\beta}_2^{(t-1)}, \tilde{\bm\beta}_3^{(t-1)}, \tau^{2(t-1)}, \gamma_1^{(t-1)}, \gamma_2^{(t-1)}, \gamma_3^{(t-1)})$, draw $d_i^{(t)}$'s using \eqref{eq:fullcond_tildedi}, independently for $i = 1, \dots, n$.
        \item Given $(d_1^{(t)}, \dots, d_n^{(t)})$ and $(\tilde{\bm\beta}_1^{(t-1)}, \tilde{\bm\beta}_2^{(t-1)}, \tilde{\bm\beta}_3^{(t-1)}, \tau^{2(t-1)})$, use \eqref{eq:post_gammad} to sample $\gamma_d^{(t)}$ given $(\gamma_1^{(t)}, \dots, \gamma_{d-1}^{(t)}, \gamma_{d+1}^{(t-1)}, \dots, \gamma_3^{(t-1)})$, sequentially for $d = 1, 2, 3$.
        \item Given $(\tilde{\bm\beta}_1^{(t)}, \dots, \tilde{\bm\beta}_{d-1}^{(t)}, \tilde{\bm\beta}_{d+1}^{(t-1)}, \dots, \tilde{\bm\beta}_3^{(t-1)}, \tau^{2(t-1)}, \gamma_1^{(t)}, \gamma_2^{(t)}, \gamma_3^{(t)})$, use \eqref{eq:modelBd}-\eqref{eq:fullcondBd} sequentially for $d = 1, 2, 3$ to draw
        \begin{enumerate}\setlength{\itemsep}{0pt}
        \item $\tilde{\bm\beta}_d^{(t)} \sim \mathcal{N}(\Sigma_{B_d}X_{B_d}^{*\T}y^{*}, \tau^{2(t-1)} \Sigma_{B_d}),$
        \item $\lambda_{gdj}^{2(t)} \sim \mathcal{IG} \left( 1, \frac{1}{\nu_{gdj}^{(t-1)}} + \frac{\beta_{dj}^{2(g)(t)}}{2 \tau^{2(t-1)}\delta_{g}^{2(t-1)}} \right)$, and
        $\nu_{gdj}^{(t)} \sim \mathcal{IG} \left(1, 1 + \frac{1}{\lambda_{gdj}^{2(t)}} \right), \ g = 1, \dots, K, \ j = 1, \dots, p_d.$
        \end{enumerate}
    
        \item Given $(\tilde{\bm\beta}_1^{(t)}, \tilde{\bm\beta}_2^{(t)}, \tilde{\bm\beta}_3^{(t)}, \gamma_1^{(t)}, \gamma_2^{(t)}, \gamma_3^{(t)})$, for $g = 1 \dots, K$, draw $$\delta_g^{2(t)} \sim \mathcal{IG}\left(\frac{1 + \sum_{d=1}^{3}p_d}{2}, \frac{1}{\xi_g^{(t-1)}} + \frac{1}{2\tau^{2(t-1)}}\sum_{d=1}^3 \sum_{j = 1}^{p_d}\frac{\beta_{dj}^{2(g)(t)}}{\lambda_{gdj}^{2(t)}}\right), \text{ and } \xi_{g}^{(t)} \sim \mathcal{IG} \left( 1, 1 + \frac{1}{\delta_{g}^{2(t)}} \right)$$
        \item Given $(\tilde{\bm\beta}_1^{(t)}, \tilde{\bm\beta}_2^{(t)}, \tilde{\bm\beta}_3^{(t)}, \gamma_1^{(t)}, \gamma_2^{(t)}, \gamma_3^{(t)})$, compute $\bm \beta^{(t)} = \text{vec}(\mathcal{B}^{(t)})$ using \eqref{eq:cp} and sample $$\tau^{2(t)} \sim \mathcal{IG}\left( a_0 + \frac{N + K \sum_{d=1}^3 p_d}{2}, b_0 + \frac{1}{2}\left\{(y^* - X^*\bm \beta^{(t)})^{\T}(y^* - X^*\bm \beta^{(t)}) + \sum_{g=1}^K \sum_{d=1}^3 \frac{\beta_{dj}^{2(g)(t)}}{\delta_g^{2(t)} \lambda_{gdj}^{2(t)}}\right\} \right).$$
    \end{enumerate}
\end{enumerate}
\end{algorithm}

\section{Theoretical Properties}\label{Section:tensorTheory}
In this section, we examine the asymptotic behavior of mean square prediction risk of the CoMET model. Our approach is grounded in the (matrix‑variate) methodology introduced by \citet{2025_SarKhSri}. Here, we generalize those ideas by formulating and analyzing their tensor‑variate counterparts. As emphasized in \cite{2025_SarKhSri}, deriving high‑dimensional asymptotic properties for posterior distributions in Bayesian mixed‑effects models is a very difficult problem that has received limited attention in the literature. The tensor‑variate setting considered here introduces an additional layer of complexity. Nonetheless, under a set of reasonable simplifying assumptions, we establish asymptotic results that offer both theoretical support and insight into the proposed methodology. 

We focus on third-order tensor covariates without any loss of generality, noting that our results directly extend to tensor covariates of any order. For this evaluation, we consider a setting where the (marginalized) true data generating model for the $i$-th subject is given by 
\begin{equation}\label{eq:populationmodel}
    y_i = X_i \bm\beta^0 + \epsilon_{0i}, \ \epsilon_{0i} \sim \mathcal{N}_{m_i}(0, \tau_0^2 V_{0i}), \ V_{0i} = Z_i (\Sigma_{3}^{0} \otimes \Sigma_{2}^{0} \otimes \Sigma_{1}^{0}) Z_i^{\T} + I_{m_i}, \ i = 1, \dots, n,
\end{equation}

\noindent
where $\mathcal{B}^0, \tau^2_0$ and $\Sigma^0 = \Sigma_{3}^{0} \otimes \Sigma_{2}^{0} \otimes \Sigma_{1}^{0}$ denote the "true" parameter values, $\bm\beta^0 = \text{vec}(\mathcal{B}^0)$, and 
\begin{equation}
y_i = \begin{pmatrix} y_{i1} \\ \vdots \\ y_{i m_{i}} \end{pmatrix} \in \RR^{m_i}, \ X_i = \begin{pmatrix} \text{vec}(\mathcal{X}_{i1})^{\T} \\ \vdots \\ \text{vec}(\mathcal{X}_{i m_{i}})^{\T} \end{pmatrix} \in \RR^{m_i \times p^*}, \ Z_i = \begin{pmatrix} \text{vec}(\mathcal{Z}_{i1})^{\T} \\ \vdots \\ \text{vec}(\mathcal{Z}_{i m_{i}})^{\T} \end{pmatrix} \in \RR^{m_i \times q^*}, 
\end{equation}

\noindent
with $p^* = p_1 p_2 p_3$, $q^* = q_1 q_2 q_3$, $ k^* = k_1 k_2 k_3$. On the other hand, the (marginalized) working compressed model for the $i$-th subject 
is given by 
\begin{equation}\label{eq:workingmodel}
        y_i = X_i \bm\beta + \epsilon'_{i}, \ \epsilon'_{i} \sim \mathcal{N}_{m_i}(0, \tau_0^2 C_{i}), \ C_{i} = Z_i S^{*\T} \Gamma^* R^*R^{* \T} \Gamma^{*\T} S^* Z_i^{\T} + I_{m_i}, \ i = 1, \dots, n,
\end{equation}

\noindent
where 
\begin{align}\label{eq:notations:parameters}
    S^* = S_3 \otimes S_2 \otimes S_1 &\in \RR^{k^* \times q^*}, \quad R^* = R_3 \otimes R_2 \otimes R_1 \in \RR^{k^* \times q^*}, \nonumber \\ \Gamma^* = \Gamma_3 \otimes \Gamma_2 \otimes \Gamma_1 &\in \RR^{k^* \times k^*}, \quad \bm\beta = \text{vec}(\mathcal{B}) \in \RR^{p^*}.
\end{align} 

\noindent
We establish the asymptotic decay of CoMET's prediction risk under compressed covariance structure with an unstructured fixed-effect coefficient tensor. In particular, we avoid the CP structure for $\mathcal{B}$ in (\ref{eq:cp}) and assume that $\mathcal{B}$ is full-rank for the purposes of this theoretical evaluation. Consequently, the working model assigns element-wise Horseshoe priors for $\text{vec}(\mathcal{B})$, given by
\begin{align}\label{eq:hsprior_vecB}
   \mathcal{B}_{j_1 j_2 j_3} \mid \lambda_{j_1 j_2 j_3}^2, \delta^2, \tau^2 & \sim \mathcal{N}(0, \lambda_{j_1 j_2 j_3}^2 \delta^2 \tau^2 ), \quad \lambda_{j_1 j_2 j_3} \sim \mathcal{C}^{+}(0, 1), \quad \delta \sim \mathcal{C}^{+}(0, 1), 
\end{align}
independently for $j_d = 1, \ldots, p_d$ and $ d = 1,2,3$. The squared global shrinkage parameter $\delta^2$ promotes the overall sparsity in $\mathcal{B}$ and $\lambda_{j_1 j_2 j_3}^2$ controls cell-specific shrinkage. The priors for the other parameters in the working model are identical to those prescribed in Section \ref{Section:OurModel}. 

Define $C = \diag(C_1, \ldots, C_n)$, $y \in \RR^N$ with $i$-th row block $y_i$, and $X \in \RR^{N \times p^*}$ with $i$-th row block $X_i$. Then, the compressed conditional posterior density of $\bm\beta$ given $\Lambda, \delta^2, C$ (under the working model) is given by 
\begin{align}
  \label{eq:cond_vecB}
  \bm\beta \mid \Lambda, \delta^2, C, y, X  \sim \mathcal{N} \left( \{ X^{\T} C^{-1} X + (\delta^{2} \Lambda)^{-1}\}^{-1} X^{\T} C^{-1} y, \tau_0^{2} \{ X^{\T} C^{-1} X + (\delta^{2} \Lambda)^{-1}\}^{-1}  \right),
\end{align}
where $\Lambda = \diag(\lambda_{111}^2, \lambda_{211}^2, \dots,\lambda_{(p_1-1)p_2p_3}^2, \lambda_{p_1p_2p_3}^2)$. The prediction risk of the compressed posterior based on the prior in \eqref{eq:hsprior_vecB} is
\begin{align}\label{eq:vecB_predrisk}
    \frac{1}{N}\EE \|X\bm{\beta}^0 - X \bar{\bm{\beta}}(\phi)\|_2^2 &= \frac{1}{N}\EE_{R,S,X,Z} \EE_y \{ \EE_{\phi \mid y} \|X\bm{\beta}^0 - X \bar{\bm{\beta}}(\phi)\|_2^2 \},
\end{align}
where the conditional posterior mean of $\bm\beta$ in \eqref{eq:cond_vecB} is denoted as $\bar{\bm\beta}(\phi)$ with $\phi = \{\Gamma^*, \Lambda, \delta^2\}$, $\EE_{R,S, X, Z}$, $\EE_{y}$, and $\EE_{\phi \mid y}$ respectively denote the expectations with respect to the distributions of $(R_1, R_2, R_3, S_1, S_2, S_3, X, Z_1, \dots, Z_n)$, $y$, and the conditional distribution of $\phi$ given $y$, respectively.

We now describe the regularity assumptions needed for deriving the theoretical guarantees in the current framework where the covariance parameter is compressed, but $\mathcal{B}$ is assumed to be full-rank.
\begin{enumerate}[label=(A\arabic*)]
\item The error variance $\tau^2$ is known and equals the true value $\tau_0^2$. The squared local shrinkage parameter $\lambda_{j_1 j_2 j_3}^2$ is supported on a compact domain $[a, 1/a]$ for a universal constant $a$.

\item The support for the Gaussian prior on $\gamma_d$ is restricted to the set $\{\Gamma_d \in \RR^{k_d \times k_d}: \| \Gamma_d\| \leq b_d\}$, where $b_d$ is a universal constant for $d = 1,2,3$, and $\| \cdot \|$ is the operator norm of a matrix. Consequently, using the fact that $\|A \otimes B\| = \|A\| \|B\|$ for any two matrices $A$ and $B$,  $\Gamma^*$ belongs to a class of matrices $\mathcal{G}$ with bounded operator norm, i.e., $\mathcal{G} = \{\Gamma^* \in \RR^{k^* \times k^*} : \|\Gamma^*\| \leq b^*\}$, where $b^* = b_1 b_2 b_3$.
\item The entries of the tensor covariates $\mathcal{X}_{ij}$ and $\mathcal{Z}_{ij}$ are mutually independent Gaussian random variables with mean zero and variance $\sigma_X^2$ and $\sigma_Z^2$, respectively, where $\sigma_X^2$ is a constant, and $\sigma_Z^2$ satisfies $n\sigma_Z^4 \left(\prod_{d=1}^3q_d^6/k_d^4+ \prod_{d=1}^3 m_{\max}^2/k_d^4 \right) = O(1)$, with $m_{\max} = \max(m_1, \dots, m_n)$.
\end{enumerate}

Assumption (A1) simplifies the setup by considering the error variance to be known and the local shrinkage parameters to belong to a compact domain. Assumption (A2) models the mode-specific covariance structures as low-rank matrices through controlling the "sparsity" of the compressed covariance parameters $\Gamma_d$s. Finally, Assumption (A3) specifies the distribution of the fixed- and random-effect covariates and facilitates the derivation of the mean square prediction risk, averaged over the distributions of $\mathcal{X}_{ij}$s and $\mathcal{Z}_{ij}$s. Based on these assumptions, Theorem \ref{vecB_TheoremPredRisk} quantifies the asymptotic behavior of the prediction risk under covariance compression.

\begin{theorem}
\label{vecB_TheoremPredRisk}
If the assumptions (A1)-(A3) hold, $p^* = o(N)$, $k^{*2}\log k^* = o(p^*)$, $k^{*2}\log \log N = o(p^*)$ and $\|\bm{\beta}^0 \|^2 = o(N)$, then the posterior predictive risk satisfies
\begin{align*}
   \frac{1}{N}\EE \| X \bm{\beta}^0 - X \bar{\bm{\beta}} (\phi) \|_2^2 &\leq \EE_{X}\left\{ \kappa^4(X)\right\} \Bigg[ \frac{2\| \bm{\beta}^0\|_2^{2}}{N}\left\{O\left(nb^{*8}\sigma_Z^8 \left(\prod_{d=1}^3 \frac{q_d^{12}}{k_d^8} + \prod_{d=1}^3\frac{m_{\max}^4}{k_d^8} \right) + nb^{*8}\sigma_Z^8 \right)\right\} \\
    &\qquad \qquad \qquad \qquad + \frac{4\tau_{0}^{2}}{a^{4}} \Bigg\{ \left(\frac{N-p^*}{N\log N}\right)O\left( nb^{*5}\sigma_Z^6 \left(\prod_{d=1}^3\frac{q_d^9}{k_d^6} + \prod_{d=1}^3 \frac{m_{\max}^3}{k_d^6}\right)\right) \\
        &\qquad \qquad \qquad \qquad + \left(\frac{p^* + 4e^{-p^*/8}}{N}\right)O\left(nb^{*3}\sigma_Z^4\left(\prod_{d=1}^3\frac{q_d^6}{k_d^4}+ \prod_{d=1}^3 \frac{m_{\max}^2}{k_d^4} \right)\right)\Bigg\}\Bigg] \\
    &= o(1),
\end{align*}
where $\kappa(X)$ denotes the condition number of $X$.
\end{theorem}

Theorem~\ref{vecB_TheoremPredRisk} justifies the practical utility of the CoMET model. It shows that the mode-specific compression dimensions $(k_1,k_2,k_3)$ can be chosen such that $k^* \ll p^*$ and $k^* \ll N$ while still yielding asymptotically negligible predictive risk. These dimensions control the complexity of the covariance compression, and computational tractability requires $k_d \ll q_d$ for $d=1,2,3$. In the random-slope setting we consider, $q_d \le p_d$ for each mode $d$ and $q_d = p_d$ when $\mathcal{Z}_{ij}$ uses the same tensor features as $\mathcal{X}_{ij}$; therefore, choosing $k_d \ll q_d$ implies that $k_d \ll p_d$, which satisfies the assumptions of the theorem. The proof of Theorem~\ref{vecB_TheoremPredRisk} is given in Section~\ref{Theory_vecB} of the supplementary material.

We establish Theorem~\ref{vecB_TheoremPredRisk} under the simplified assumption that the fixed-effect coefficient tensor $\mathcal{B}$ is unstructured. A natural extension would quantify predictive risk in terms of the CP-rank $K$ of $\mathcal{B}$, in addition to $N$, $n$, and $p_d, q_d, k_d$ for $d=1, 2, 3$. Proving such a result would require non-trivial extensions of the arguments used for Theorem~\ref{vecB_TheoremPredRisk}, and we leave this analysis to future work. We, however, conjecture that an analogous result holds because our CoMET implementation combines covariance compression with shrinkage on a low-rank CP decomposition of $\mathcal{B}$ (see Algorithm~\ref{algoCoMET:1}), and the empirical results in the next section show strong performance in estimation accuracy, prediction, and uncertainty calibration.

\section{Simulated Data Analyses}\label{Section:Simulations}
\subsection{Setup}\label{Section:Sim_Setup}

We evaluated the CoMET model's empirical performance in simulations with two-way tensors ($D=2$). We set $p_d = q_d = 32$ for $d = 1, 2$  and generated $\Bcal \in \RR^{32 \times 32}$ from the CP representation \eqref{eq:cp} with rank $K=4$. The factor matrices $B_d \in \RR^{32 \times 4}$ were generated by sampling $\lceil 0.25 p_1 K\rceil$ entries of $B_1$ and $\lceil 0.25 p_2 K\rceil$ entries of $B_2$ from $\{-2, -1, 1, 2\}$ with replacement. All the remaining $B_1$ and $B_2$ entries are set to zero. For the random-effects, we used a separable covariance with an equicorrelation structure in each mode. Specifically, for $d = 1, 2$, we set $(\Sigma_{d})_{jj^{\prime}} = 0.5$ for any $j, j^{\prime} \in \{1, \dots, q_d\} $ if $ j \neq j^{\prime}$ and $(\Sigma_{d})_{jj} = 1$ 
for any $j\in \{1, \dots, q_d\}$.  This full-rank mode-specific covariance structure allowed us to assess the performance of quasi-likelihood methods, including CoMET, that used low-rank covariance approximations. Finally, the idiosyncratic error variance was $\tau^2 = 0.1$.

Using these parameters, we simulated the clustered data using \eqref{eq:Origmodel} for $n = 100$ subjects with balanced cluster sizes $m_1 = \dots = m_n = m \in \{3, 6, 9, 12\}$. The entries of $\mathcal{X}_{ij}$s and $\mathcal{Z}_{ij}$s were independently generated from $\mathcal{N}(0, 1)$ distribution. Each of these settings with four different sample sizes were replicated 25 times. An independent validation set of $n^* = 50$ new subjects was generated corresponding to each of these training datasets to assess the predictive performance.

We applied the CoMET model to the simulated data using Algorithm \ref{algoCoMET:1} and examined its sensitivity to different choices of fixed-effect ranks and covariance compression dimensions. For simplicity, we imposed the same compression dimensions across all modes, by setting $k_1 = k_2 = k$ and varied $k \in \{\log(32) \approx 3, 6, 9\}$ and CP-rank $K \in \{1, 2, 4, 6, 8\}$. These design choices aligned with low-rank structural assumptions of CoMET. For a weakly informative prior specification, the hyperparameters in \eqref{eq:hsprior_2} were set as $a_0 = b_0 = 0.01$ and $\sigma_d^2 = 1$, for all $d$. We ran the Gibbs sampler described in Algorithm \ref{algoCoMET:1} for 11,000 iterations and discarded the first 1,000 iterations as burn-in.

We compared the CoMET model against an oracle variant of CoMET (oracle), and penalized competitors based on quasi-likelihood (PQL) \citep{FanLi12, Lietal21}, and generalized estimating equations (GEE) \citep{2019_Zhang_etal}. The oracle fixes the covariance components in \eqref{eq:Origmodel} at their true values and varies the CP-rank $K$ of $\Bcal$. For each choice of $K$, the oracle draws $(\mathcal{B}, \tau^2)$ using Algorithm \ref{algoCoMET:1} while treating the covariance components as known. The oracle's performance serves as a benchmark for evaluating the impact of covariance misspecification in the CoMET model across different choices of $K$.

The PQL methods have two variants, PQL-1 and PQL-2. We applied them after vectorizing the fixed- and random-effects covariates in \eqref{eq:Origmodel}.  PQL-1 approximates the $q^* \times q^*$ random-effects covariance matrix $\Sigma = \Sigma_2 \otimes \Sigma_1$ using $\log(N)I_{nq^*}$ as its proxy \citep{FanLi12}. In contrast, PQL-2 selects the scaling factor by cross-validation and incorporates a debiasing step for estimating $\mathcal{B}$ \citep{Lietal21}. Following \cite{Lietal21}, we 
set the penalty parameter for PQL-2 as $\lambda = \hat{\tau} \{\log(p^*)/N\}^{1/2}$, where $\hat{\tau}$ is estimated using the scaled lasso \citep{SunZha12}. PQL-1 has no debiasing step, so we follow the PQL-2 method for debiasing $\Bcal$ estimates and constructing confidence intervals using the penalty parameter that minimizes the cross-validation error (i.e., $\lambda = \lambda_{\min}$). Finally, PQL methods are unaffected by the choices of $k$ and $K$ because they vectorize the covariates and do not exploit the low-rank tensor structure.

We implemented the GEE approach using the \texttt{Sparsereg} MATLAB package \citep{sparsereg_lib} with an equicorrelation matrix as the working correlation structure for the marginal error term, together with a lasso penalty for estimating $\mathcal{B}$. The tuning parameter involved in penalization was chosen through cross-validation. Because this penalized GEE approach does not incorporate a debiasing procedure, we excluded it from the confidence interval comparisons. Similar to the PQL methods, GEE's performance is unaffected by the choices of $k$ and $K$.

We compared the five methods using metrics for fixed-effects estimation and inference and out-of-sample prediction. Fixed-effects estimation accuracy was assessed via the root mean square error (RMSE) of $\mathcal{B}$, and inference was assessed using the empirical coverage and mean width of 95\% credible or confidence intervals. The out-of-sample predictive accuracy was assessed using the root mean squared prediction error (RMSPE). Specifically,
\begin{equation}\label{eq:rmse_rmspe}
\begin{aligned}
\mathrm{RMSE}(\hat{\mathcal{B}})
&= \left\{\frac{\|\mathcal{B}-\hat{\mathcal{B}}\|_F^2}{p^*}\right\}^{1/2},\quad
\mathrm{RMSPE}(\hat y)
&= \left[\frac{1}{n^*}\sum_{i=1}^{n^*}\left\{\frac{1}{m_i}\sum_{j=1}^{m_i}(y^*_{ij}-\hat y_{ij})^2\right\}\right]^{1/2},
\end{aligned}
\end{equation}
where $\|\cdot\|_F$ denotes the Frobenius norm, $y^*_{ij}$ and $\hat y_{ij}$ are the true and predicted $j$-th responses for test subject $i$, and $n^*$ is the number of test subjects. A point estimate of $\mathcal{B}$ is denoted as $\hat{\mathcal{B}}$, which equals the posterior median for oracle and CoMET, and the regularized estimator for the PQL-1, POL-2, and GEE methods. Finally, we assessed posterior predictive uncertainty quantification by comparing CoMET and oracle in terms of the empirical coverage and mean width of their 95\% prediction intervals.

\subsection{Empirical Results}

The CoMET model outperformed its penalized competitors in prediction and estimation of $\Bcal$ for every choice of CP-rank $K$ and sample size when the compressed covariance dimension was fixed at $k=3$ (Figure~\ref{fig:rmsermspe_equicorr_k3}). In the small-sample regime ($m=3,6$), CoMET achieved lower RMSE and RMSPE than  PQL-1, PQL-2, and GEE methods, and closely matched the oracle benchmark across all choices of $K$. For larger sample sizes ($m=9,12$), PQL-1 outperformed CoMET and the oracle when $K=1$, but the performance of CoMET and the oracle improved as $K$ increased. CoMET's relatively poor performance at $K=1$ is consistent with bias due to underspecification of the CP-rank of $\Bcal$. These findings are robust to the choice of compression dimension, with similar conclusions holding for $k=6$ and $k=9$; see Section~\ref{Section:supp_numresults} of the supplementary material. CoMET's strong performance despite the covariance misspecification highlights the effectiveness of its compression strategy for both fixed-effect estimation and prediction.

The credible intervals for $\Bcal$ in the oracle and CoMET methods attain nominal coverage when the CP-rank $K$ is sufficiently large (Figure~\ref{fig:covgwidthCI_equicorr_k3}). When $K$ is underspecified (i.e., $K=1,2$), the credible intervals for both methods exhibit undercoverage. In contrast, for moderate CP-ranks (i.e., $K=4,6,8$), CoMET achieves near-nominal coverage and yields substantially narrower intervals than the debiased PQL-1 and PQL-2 estimators, while retaining a similar performance to the oracle. These patterns hold for higher sample sizes and covariance compression dimensions; see Section~\ref{Section:supp_numresults} of the supplementary material for results with $k=6$ and $k=9$.

The oracle and CoMET have comparable posterior predictive uncertainty quantification. Across all combinations of $k$, $K$, and $m$, CoMET’s 95\% prediction intervals attain near-nominal empirical coverage, closely matching the oracle's performance  (Table~\ref{tab:PredUQ_covgPI_equicorr}). For small covariance compression dimensions (i.e., $k=3$), the widths of prediction intervals obtained using the oracle and CoMET are nearly identical (Figure~\ref{fig:PredUQ_widthPI_equicorr}). For larger compression dimensions (i.e., $k=6,9$), CoMET’s prediction intervals are slightly wider than oracle’s, but the gap decreases as the CP-rank $K$ of the fixed-effects coefficients increases. These results demonstrate that CoMET provides calibrated predictive uncertainty and that moderate choices of $K$ strike a favorable balance between prediction-interval width and nominal coverage across a broad range of sample sizes and compression dimensions.

In summary, with appropriate choices of the covariance compression dimension $k$ and the fixed-effects coefficient's CP-rank $K$, CoMET outperforms existing penalized methods in both prediction and fixed-effects inference. For these  metrics, CoMET closely matches the oracle's performance, even when the compression dimensions are small relative to the true dimension of the covariance components. The stability of these results across different values of $m$ highlights CoMET's practical utility for fixed-effects inference and prediction when the random-effects covariance structure is unknown and high-dimensional.

\begin{figure}[htbp]
    \centering
    \includegraphics[width=6.5in,height=6in]{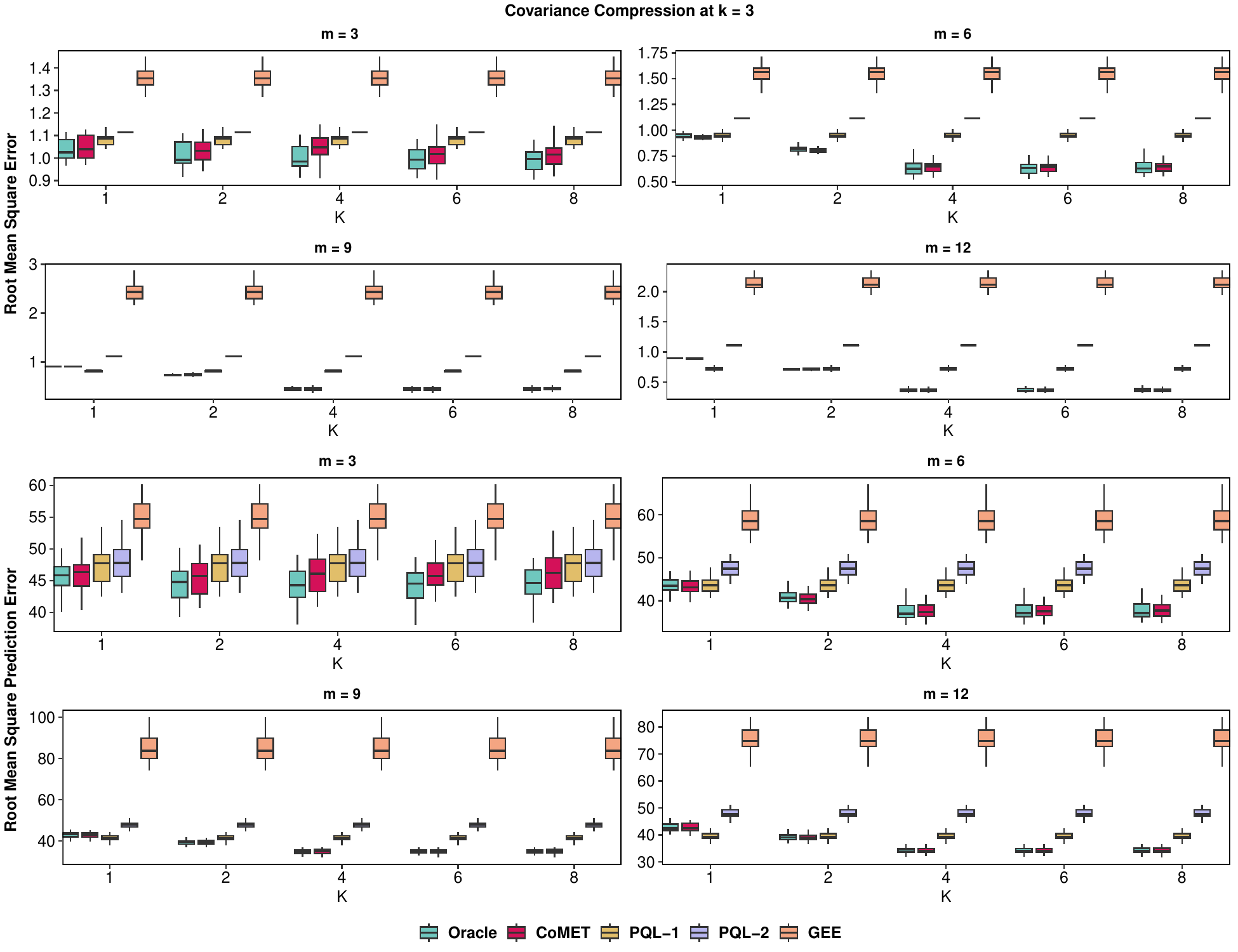}
    \caption{Comparison of RMSE and RMSPE, defined in \eqref{eq:rmse_rmspe}. CoMET substantially outperforms existing penalized quasi-likelihood methods in estimation of $\mathcal{B}$ (RMSE) and out-of-sample prediction (RMSPE) across various choices of fixed-effect rank $K$ and cluster sizes $m \in \{3, 6, 9, 12\}$. Results are presented for compressed covariance dimension $k = 3$, summarized over 25 replications. CoMET, the compressed mixed-effects tensor model; oracle, the oracle benchmark; PQL-1, the penalized quasi-likelihood approach of \citet{FanLi12}; PQL-2, the penalized quasi-likelihood approach of \citet{Lietal21}; GEE, the penalized generalized estimating equations method of \citet{2019_Zhang_etal}.}
    \label{fig:rmsermspe_equicorr_k3}
\end{figure}

\begin{figure}[htbp]
    \centering
    \includegraphics[width=6.5in,height=6in]{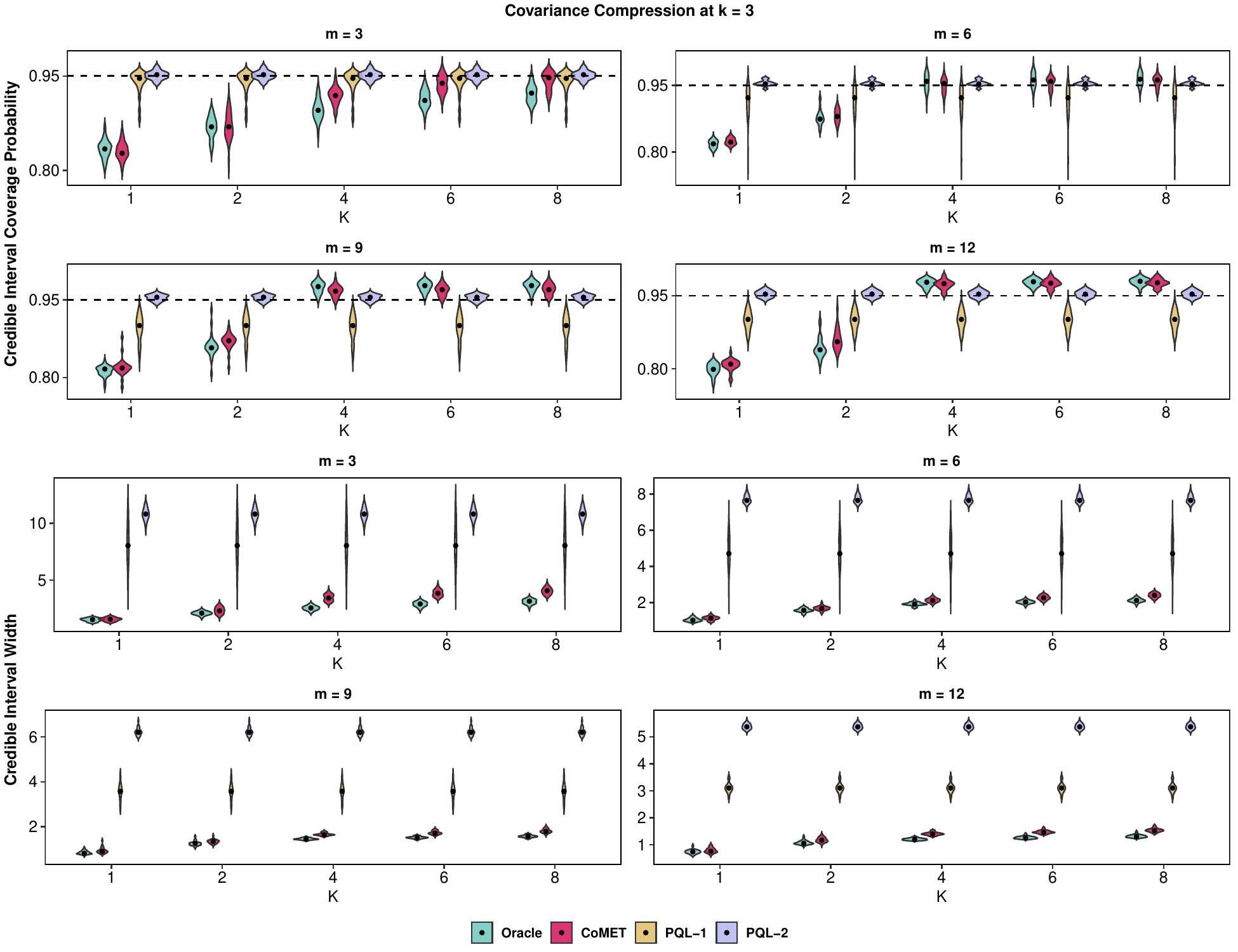}
    \caption{Fixed-effects inference comparisons. The CoMET model produces narrower credible intervals for $\mathcal{B}$ entries than the PQL methods when $k = 3$ and $K \in \{4, 6, 8\}$, along with achieving near-nominal coverage across all cluster sizes $m \in \{3, 6, 9, 12\}$. All metrics are summarized over 25 replications. CoMET, the compressed mixed-effects tensor model; oracle, the oracle variant of CoMET; PQL-1, the penalized quasi-likelihood approach of \citet{FanLi12}; PQL-2, the penalized quasi-likelihood approach of \citet{Lietal21}. GEE, the penalized generalized estimating equations approach of \citet{2019_Zhang_etal}, is excluded because of absence of debiasing technique for confidence intervals construction.}
    \label{fig:covgwidthCI_equicorr_k3}
\end{figure}

\begin{table}[htbp]
\centering
\caption{Comparison of the empirical coverage of CoMET's and oracle's posterior predictive intervals. CoMET yields 95\% prediction intervals with near-nominal empirical coverage across all choices of compressed covariance dimension ($k$) and fixed-effect coefficient's rank ($K$) for various cluster sizes ($m$), and closely matches the oracle's performance. Reported values are the median coverage probability over 25 replications, with subscripts indicating the corresponding quartile deviation.}
\scalebox{0.78}{
\begin{tabular}{|c|c|cc|cc|cc|cc|cc|}
\hline
\multirow{3}{*}{$k$} & \multirow{3}{*}{Cluster Sizes} & \multicolumn{10}{c|}{$K$}\\ \cline{3-12}
& & \multicolumn{2}{c|}{1} & \multicolumn{2}{c|}{2} & \multicolumn{2}{c|}{4} & \multicolumn{2}{c|}{6} & \multicolumn{2}{c|}{8} \\ \cline{3-12}
& & Oracle & CoMET  & Oracle & CoMET & Oracle & CoMET & Oracle & CoMET & Oracle & CoMET\\ \hline
\multirow{6}{*}{$3$} & & & & & & & & & & & \\
& $m = 3$ & $0.96_{0.01}$ & $\boldsymbol{0.95_{0.01}}$ & $0.95_{0.01}$ & $\boldsymbol{0.96_{0.01}}$ & $0.95_{0.02}$ & $\boldsymbol{0.95_{0.02}}$ & $0.94_{0.01}$ & $\boldsymbol{0.94_{0.01}}$ & $0.94_{0.01}$ & $\boldsymbol{0.93_{0.02}}$\\
& $m = 6$ & $0.96_{0.01}$ & $\boldsymbol{0.95_{0.01}}$ & $0.96_{0.01}$ & $\boldsymbol{0.96_{0.01}}$ & $0.95_{0.01}$ & $\boldsymbol{0.96_{0.02}}$ & $0.95_{0.02}$ & $\boldsymbol{0.95_{0.02}}$ & $0.95_{0.01}$ & $\boldsymbol{0.95_{0.02}}$\\
& $m = 9$ & $0.96_{0.01}$ & $\boldsymbol{0.95_{0.01}}$ & $0.96_{0.01}$ & $\boldsymbol{0.95_{0.01}}$ & $0.95_{0.01}$ & $\boldsymbol{0.94_{0.01}}$ & $0.95_{0.01}$ & $\boldsymbol{0.95_{0.01}}$ & $0.95_{0.01}$ & $\boldsymbol{0.95_{0.01}}$\\
& $m = 12$ & $0.96_{0.00}$ & $0.95_{0.00}$ & $0.96_{0.01}$ & $0.95_{0.01}$ & $0.95_{0.01}$ & $0.95_{0.01}$ & $0.95_{0.01}$ & $0.95_{0.01}$ & $0.95_{0.01}$ & $0.95_{0.01}$\\
& & & & & & & & & & & \\ \hline
\multirow{6}{*}{$6$} & & & & & & & & & & & \\
& $m = 3$ & $0.96_{0.01}$ & $\boldsymbol{0.95_{0.01}}$ & $0.95_{0.01}$ & $\boldsymbol{0.95_{0.01}}$ & $0.95_{0.02}$ & $\boldsymbol{0.95_{0.02}}$ & $0.94_{0.01}$ & $\boldsymbol{0.95_{0.02}}$ & $0.94_{0.01}$ & $\boldsymbol{0.95_{0.02}}$\\
& $m = 6$ & $0.96_{0.01}$ & $0.97_{0.01}$ & $0.96_{0.01}$ & $0.96_{0.01}$ & $0.95_{0.01}$ & $0.96_{0.01}$ & $0.95_{0.02}$ & $0.95_{0.01}$ & $0.95_{0.01}$ & $0.95_{0.01}$\\
& $m = 9$ & $0.96_{0.01}$ & $0.96_{0.01}$ & $0.96_{0.01}$ & $0.96_{0.01}$ & $0.95_{0.01}$ & $0.96_{0.01}$ & $0.95_{0.01}$ & $0.96_{0.01}$ & $0.95_{0.01}$ & $0.96_{0.01}$\\
& $m = 12$ & $0.96_{0.00}$ & $0.95_{0.00}$ & $0.96_{0.01}$ & $0.96_{0.01}$ & $0.95_{0.01}$ & $0.96_{0.00}$ & $0.95_{0.01}$ & $0.95_{0.01}$ & $0.95_{0.01}$ & $0.96_{0.01}$\\
& & & & & & & & & & & \\ \hline
\multirow{6}{*}{$9$} & & & & & & & & & & & \\
& $m = 3$ & $0.96_{0.01}$ & $\boldsymbol{0.95_{0.02}}$ & $0.95_{0.01}$ & $\boldsymbol{0.95_{0.02}}$ & $0.95_{0.02}$ & $\boldsymbol{0.95_{0.02}}$ & $0.94_{0.01}$ & $\boldsymbol{0.95_{0.01}}$ & $0.94_{0.01}$ & $\boldsymbol{0.95_{0.02}}$\\
& $m = 6$ & $0.96_{0.01}$ & $0.96_{0.01}$ & $0.96_{0.01}$ & $0.96_{0.01}$ & $0.95_{0.01}$ & $0.95_{0.01}$ & $0.95_{0.02}$ & $0.95_{0.01}$ & $0.95_{0.01}$ & $0.95_{0.01}$\\
& $m = 9$ & $0.96_{0.01}$ & $0.96_{0.01}$ & $0.96_{0.01}$ & $0.96_{0.01}$ & $0.95_{0.01}$ & $0.95_{0.01}$ & $0.95_{0.01}$ & $0.95_{0.01}$ & $0.95_{0.01}$ & $0.95_{0.01}$\\
& $m = 12$ & $0.96_{0.00}$ & $0.99_{0.00}$ & $0.96_{0.01}$ & $0.98_{0.00}$ & $0.95_{0.01}$ & $0.98_{0.00}$ & $0.95_{0.01}$ & $0.98_{0.01}$ & $0.95_{0.01}$ & $0.97_{0.01}$\\
& & & & & & & & & & & \\ \hline
\end{tabular}
}
\label{tab:PredUQ_covgPI_equicorr}
\end{table}

\begin{figure}[htbp]
    \centering
    \includegraphics[width=0.8\linewidth,height=3.2in]{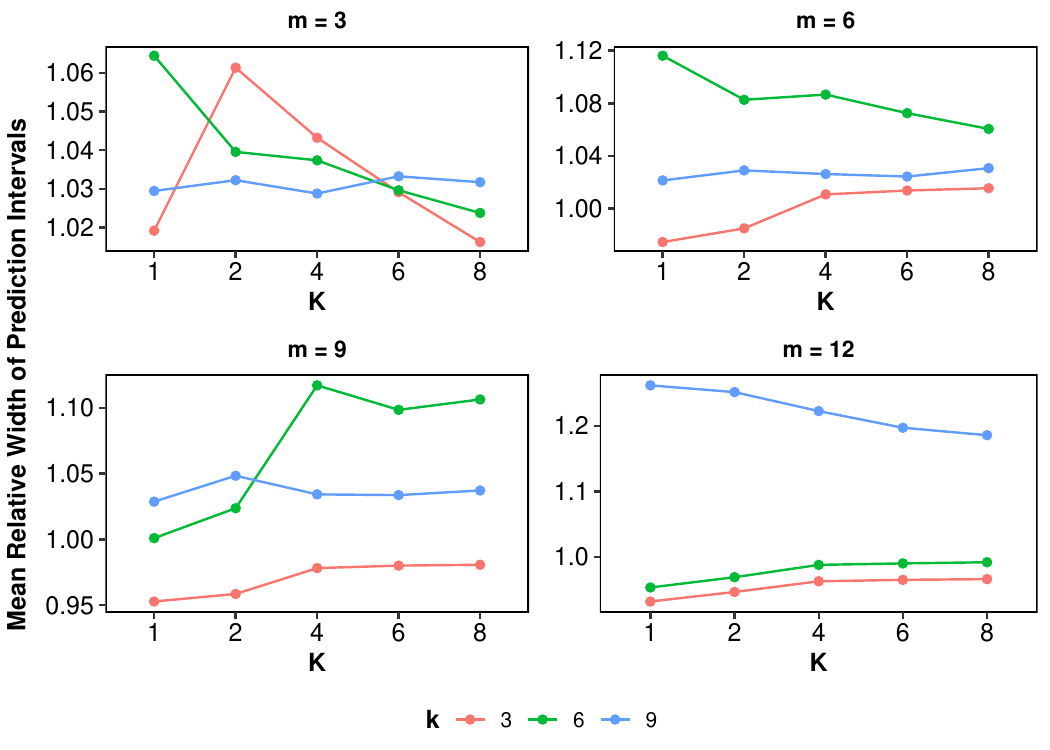}
    \caption{Posterior predictive interval widths of CoMET relative to those of the oracle. The 95\% prediction intervals produced by CoMET are comparable in width to those by the oracle benchmark for all choices of cluster size ($m$), covariance compression dimension ($k$), and fixed-effect rank ($K$). For each replication, relative width is computed by dividing the width of a prediction interval of CoMET by that of oracle. Across 25 replicates, the relative widths show minimal variability with standard deviations on the order of $10^{-2}$.}
  \label{fig:PredUQ_widthPI_equicorr}
\end{figure}

\section{Real Data Analyses}\label{Section:RealDataAnalysis}

We applied CoMET to two publicly available datasets, DEAM and LFW, introduced earlier. In the DEAM analysis, we predicted a song's emotion scores using matrix-valued acoustic features observed repeatedly over time. In the LFW analysis, we predicted a celebrity's facial traits from multiple face images represented as third-order tensor covariates. We compared CoMET's predictive performance with PQL-1, PQL-2 and GEE methods, with no fixed-effects estimation comparisons due to the absence of the ground truth. We did not include the oracle because the true covariance structure was unknown in these data.

\subsection{Music Emotion Recognition using DEAM Data}\label{Section:DEAMData}

Understanding human emotion perception through musical acoustics is a key problem in music emotion recognition. The DEAM data is one of the largest publicly available data in this domain, comprising emotion annotations for 1802 songs spanning a wide range of genres \citep{2017_DEAM}. For benchmarking the methods under consideration, we used the first 200 songs from DEAM's 2014 evaluation set.

Each selected song is accompanied by time-varying annotations of two kinds of emotions: valence and arousal. Valence reflects the degree of positive or negative emotion expressed in the song, and arousal captures the perceived energy of the song. These averaged dynamic annotations, generated at a 2Hz sampling rate, are measured on a continuous scale of $[-1, 1]$. The DEAM data also includes 260 features, comprising mean, standard deviation, and first-derivative summaries of acoustic characteristics (e.g., MFCCs and spectral and timbre features), computed over the interval from $15$ to $43.5$ seconds at $0.5$ second increments.

We used the first principal component of valence and arousal as the response, which captures the dominant emotion dimension. We focused on the 12 spectral descriptors in the data because spectral dynamics are primary drivers of perceived emotion \citep{2023_PandaSurvey}. We represented them as a two-way tensor covariates . The first mode indexed the spectral descriptors (e.g., centroid, flux, harmonicity), and the second mode indexed four statistics: mean, standard deviation, mean of first-order derivatives, and standard deviation of first-order derivatives. This resulted in matrix-valued fixed- and random-effects covariates with dimensions $p_1=q_1=12$ and $p_2=q_2=4$. We standardized each covariate entry across all observations to have mean zero and variance one. For each song, we averaged the response and covariates over pairs of consecutive $0.5$ second timestamps, resulting in $m=29$ repeated measurements per song. We randomly assigned 70\% of the songs ($n=140$) to the training set and held out the remaining 30\% as the test set ($n^*=60$). We repeated this train-test split 25 times to assess the stability of the results.

We used Algorithm \ref{algoCoMET:1} to fit the CoMET model and examined the sensitivity of its performance across various combinations of $(k, K)$. Because $\Bcal$ is a $12 \times 4$ matrix, its maximum CP-rank is 4; therefore, we set $K \in \{\log(12) \approx 2, 4\}$. Covariance compression was applied symmetrically across the two modes with $k_d = k \in \{2, 4\}$ for $d = 1, 2$. Following the simulation studies, we set the hyperparameters as $\sigma_d^2 = 1$ for $d = 1, 2$, and $a_0 = b_0 = 0.01$, and ran the Gibbs sampler for 11,000 iterations, discarding the first 1,000 iterations as burn-in. The setup for the penalized competitors was kept identical to that described in the simulations; see Section \ref{Section:Simulations}.

CoMET achieved the lowest RMSPE relative to PQL-1 and PQL-2 across all choices of $(k,K)$ (Figure~\ref{fig:rmspe_RealData}), demonstrating strong out-of-sample predictive performance and robustness to the choice of compression dimension and fixed-effects coefficients rank. Compared to GEE, CoMET yielded competitive predictive accuracy for all $(k, K)$, while exhibiting greater stability across data splits. To assess predictive uncertainty, we examined CoMET's 95\% posterior predictive intervals across $(k,K)$ and found that $(k,K)=(4,4)$ yields the narrowest intervals while maintaining a median empirical coverage of 0.96 (Table~\ref{tab:UQ_cometRealData}). At $(k,K)=(4,4)$, the CoMET estimate of $\mathcal{B}$ suggests that the average spectral centroid is most predictive of the dominant perceived emotion (Figure~\ref{fig:B_estimate_DEAM}). PQL-1 and GEE estimated $\Bcal$ with similar sparsity patterns. Compared to CoMET, the non-zero entries estimated by PQL-1 exhibit stronger shrinkage resulting in smaller magnitudes, whereas those from GEE have larger magnitudes. On the other hand, PQL-2 shrinks nearly every entry of $\mathcal{B}$ to zero. In summary, DEAM data analysis results indicate that CoMET captures the repeated measure structure of the acoustic features to improve emotion prediction, provides well-calibrated predictive uncertainty, and identifies the key acoustic features predictive of the dominant emotion.

\subsection{Facial Trait Recognition using LFW Data}\label{Section:LFWData}

The LFW data is a benchmark repeated-measures dataset in face recognition applications \citep{2007_LFWData,2016_LFWSurvey}. The database comprises over 13,000 publicly available images and 73 real-valued facial attributes of 5,749 individuals \citep{2009_Kumar}. Each image of the face of an individual is accompanied by several attributes capturing a wide range of characteristics measured on a continuous scale, including facial expressions (e.g., smiling, frowning), facial features (e.g., nose and mouth shape, hair color), imaging conditions (e.g., lighting type and blurriness), among others. Our objective is to predict these facial characteristics based on one or more images available per individual.

We pre-processed the dataset in three steps. First, we extracted individuals with at least 4 and at most 10 images. This filtering ensures that the data is computationally manageable. The resultant dataset consists of $N = 2421$ images for $n = 447$ subjects with $m_i \in \{ 4, \dots, 10\}$. Second, we computed the first principal component of the 73 facial attributes to construct a scalar response as an interpretable summary while retaining the largest proportion of variability among these related attributes. Finally, we resized each $250 \times 250 \times 3$ image into a $32 \times 32 \times 3$ tensor using the R package \emph{imager}, where the first two modes represent the spatial pixel locations and the third mode indexes the red-green-blue (RGB) color channels. This yielded three-way tensor-valued fixed- and random-effect covariates with dimensions $p_1 = p_2 = q_1 = q_2 = 32$ and $p_3 = q_3 = 3$. We standardized each covariate entry across all observations to have mean zero and variance one.

To implement the CoMET model, we allowed the covariance compression to vary across the modes, reflecting the fact that not all the modes are high-dimensional in this scenario. Specifically, for the first two modes corresponding to pixels, we set the compression dimensions as $k_1 = k_2 = k \in \{\log(32) \approx 3, 6, 9\}$, while retaining the original dimension for the third mode with $k_3 = q_3 = 3$. For simplicity, we varied the fixed-effect CP-rank over the same grid, $K \in \{3, 6, 9\}$. Similar to Section \ref{Section:Simulations}, we ran the Gibbs sampler for 11,000 iterations, discarding the first 1,000 iterations as burn-in with the same choices of hyperparameters for CoMET. The tuning parameters for the penalized methods were also kept unchanged. The equicorrelation structure failed to produce a valid covariance matrix for the LFW data, so we used the independent working correlation structure because GEE remains consistent under working correlation misspecification \citep{2019_Zhang_etal}. We compared the out-of-sample predictive performance of these methods by randomly choosing 70\% of the subjects $(n = 313)$ for model training and using the remaining 30\% individuals $(n^* = 134)$ for validation. We repeated this test-train split 25 times.

Across combinations of the covariance compression dimension $k$ and fixed-effects  coefficient's CP-rank $K$, CoMET achieved predictive accuracy that was comparable or better than its penalized competitors (Figure~\ref{fig:rmspe_RealData}). CoMET's 95\% posterior predictive intervals also attained near-nominal empirical coverage across all $(k,K)$ (Table~\ref{tab:UQ_cometRealData}). Among all the choices of $(k, K)$, $(k,K)=(9,6)$ provided the most favorable tradeoff between the predictive interval width and coverage, yielding the narrowest prediction intervals with empirical coverage closest to 0.95. CoMET with $(k,K)=(9,6)$ also outperforms its competitors in predictive accuracy. To interpret the $\Bcal$ estimates, we examined the frontal slices of $\mathcal{B}$ corresponding to the RGB channels. For each channel, the CoMET estimate for $(k,K)=(9,6)$ assigns the largest signal to the central facial region, indicating that this region is most informative for predicting the facial attribute (Figure~\ref{fig:B_estimate_LFW3D}). In contrast, the PQL methods produce substantially sparser coefficient estimates, possibly due to over-penalization. GEE, on the other hand, yields comparatively denser estimates of the frontal slices of $\mathcal{B}$ than CoMET, while producing the largest signals in the central region, consistent with the pattern observed for CoMET. In summary, these results suggest that when $(k,K)$ are low-to-moderate relative to the original tensor dimensions, CoMET achieves accurate prediction, provides reliable predictive uncertainty quantification, and identifies image regions predictive of facial traits.

\begin{figure}[h!]
    \centering
    \includegraphics[width=0.85\linewidth,height=3in]{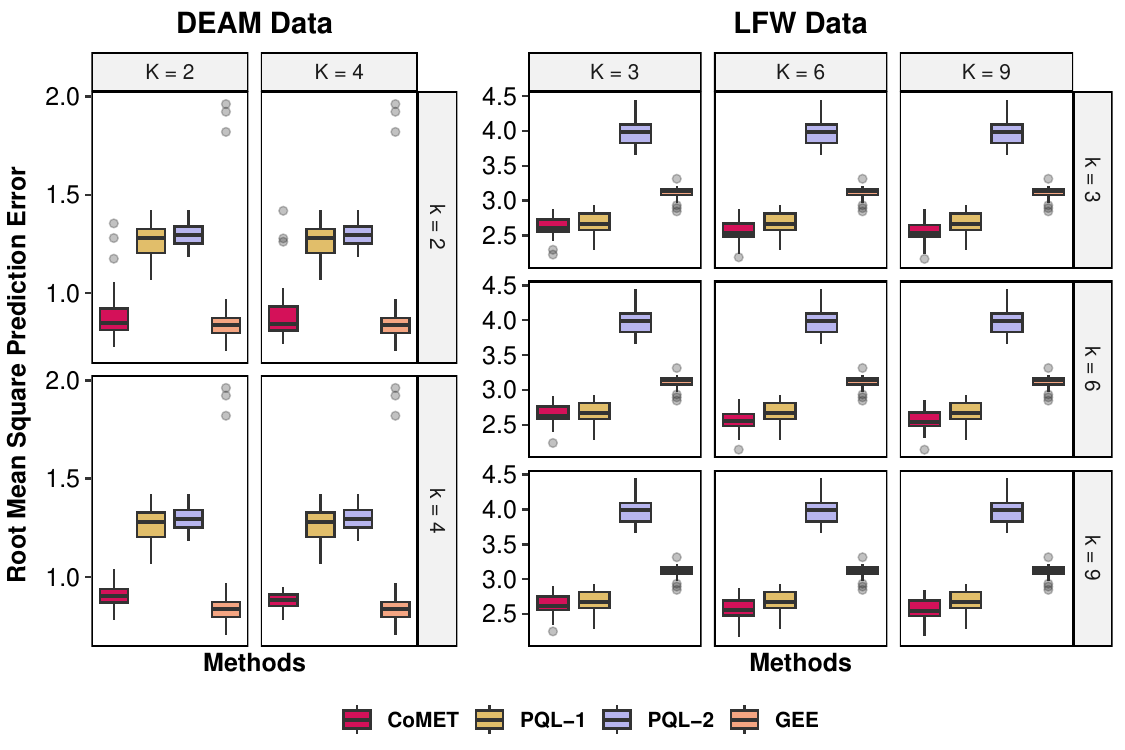}
    \caption{Comparison of RMSPE, defined in \eqref{eq:rmse_rmspe}, in real-data applications (\textbf{Left:} DEAM data and \textbf{Right:} LFW data). The CoMET model demonstrates either competitive or superior predictive accuracy across all covariance compression dimensions ($k$) and fixed-effects rank ($K$) compared to the penalized methods, PQL-1 \citep{FanLi12}, PQL-2 \citep{Lietal21} and GEE \citep{2019_Zhang_etal}. Results are summarized over 25 random splits of each of the datasets.}
  \label{fig:rmspe_RealData}
\end{figure}

\begin{table}[htbp]
\centering
\caption{Results for empirical coverage probability and width of 95\% prediction intervals produced by CoMET in real-data applications (\textbf{Left:} DEAM data and \textbf{Right:} LFW data). CoMET yields well-calibrated posterior predictive intervals for varying covariance compression dimensions ($k$) and fixed-effect ranks ($K$). Reported values are median coverage probability summarized over 25 random splits of each of the real datasets, with quartile deviation reported in subscript to assess the variability across splits.}

\makebox[0.8\textwidth][c]{
\begin{subtable}[t]{0.32\textwidth}
\centering
\caption{DEAM data}

\scalebox{0.85}{%
\begin{tabular}{|c|c|cc|}
\hline
\multirow{2}{*}{} & \multirow{2}{*}{$k$} & \multicolumn{2}{c|}{$K$} \\ \cline{3-4}
& & 2 & 4 \\ \hline
\multirow{2}{*}{Coverage} & 2 & $0.96_{0.01}$ & $0.96_{0.01}$ \\
& 4 & $0.95_{0.01}$ & $0.95_{0.01}$ \\
\hline
\multirow{2}{*}{Interval Width} & 2 & $1.74_{0.02}$ & $1.73_{0.02}$ \\
& 4 & $1.71_{0.03}$ & $1.69_{0.03}$ \\ \hline
\end{tabular}%
}
\end{subtable}
\hspace{0.05\textwidth}
\begin{subtable}[t]{0.43\textwidth}
\centering
\caption{LFW data}

\scalebox{0.65}{
\begin{tabular}{|c|c|ccc|}
\hline
\multirow{2}{*}{} & \multirow{2}{*}{$k$} & \multicolumn{3}{c|}{$K$} \\ \cline{3-5}
& & 3 & 6 & 9 \\ \hline
\multirow{3}{*}{Coverage} & 3 & $0.93_{0.01}$ & $0.93_{0.01}$ & $0.93_{0.01}$ \\
& 6 & $0.94_{0.01}$ & $0.95_{0.01}$ & $0.95_{0.02}$ \\
& 9 & $0.94_{0.01}$ & $0.95_{0.01}$ & $0.95_{0.01}$ \\ \hline
\multirow{3}{*}{Interval Width} & 3 & $9.53_{0.28}$ & $9.32_{0.26}$ & $9.27_{0.26}$ \\
& 6 & $10.98_{0.19}$ & $10.46_{0.42}$ & $10.76_{0.32}$ \\
& 9 & $10.63_{0.40}$ & $10.18_{0.24}$ & $10.20_{0.42}$ \\ \hline
\end{tabular}%
}
\end{subtable}
}
\label{tab:UQ_cometRealData}
\end{table}

\begin{figure}[h!]
    \centering
    \includegraphics[width=0.85\linewidth,height=2.5in]{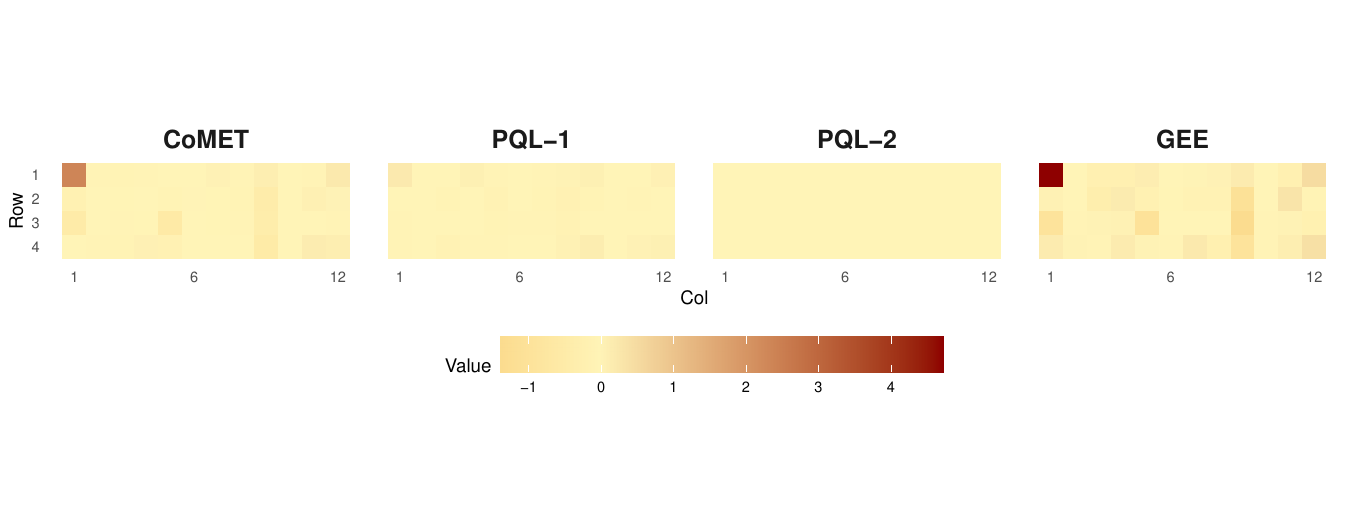}
    \caption{Estimated $12 \times 4$ matrix-valued fixed-effect coefficient $\mathcal{B}$ (transposed for visualization), averaged over 25 random training splits of the DEAM dataset. CoMET, the compressed mixed-effects model for tensors; PQL-1, the penalized quasi-likelihood method of \citet{FanLi12}; and PQL-2, the penalized quasi-likelihood approach of \citet{Lietal21}; GEE, the penalized generalized estimating equations approach of \citet{2019_Zhang_etal}.}
    \label{fig:B_estimate_DEAM}
\end{figure}

\begin{figure}[h!]
    \centering
    \includegraphics[width=0.85\linewidth,height=6in]{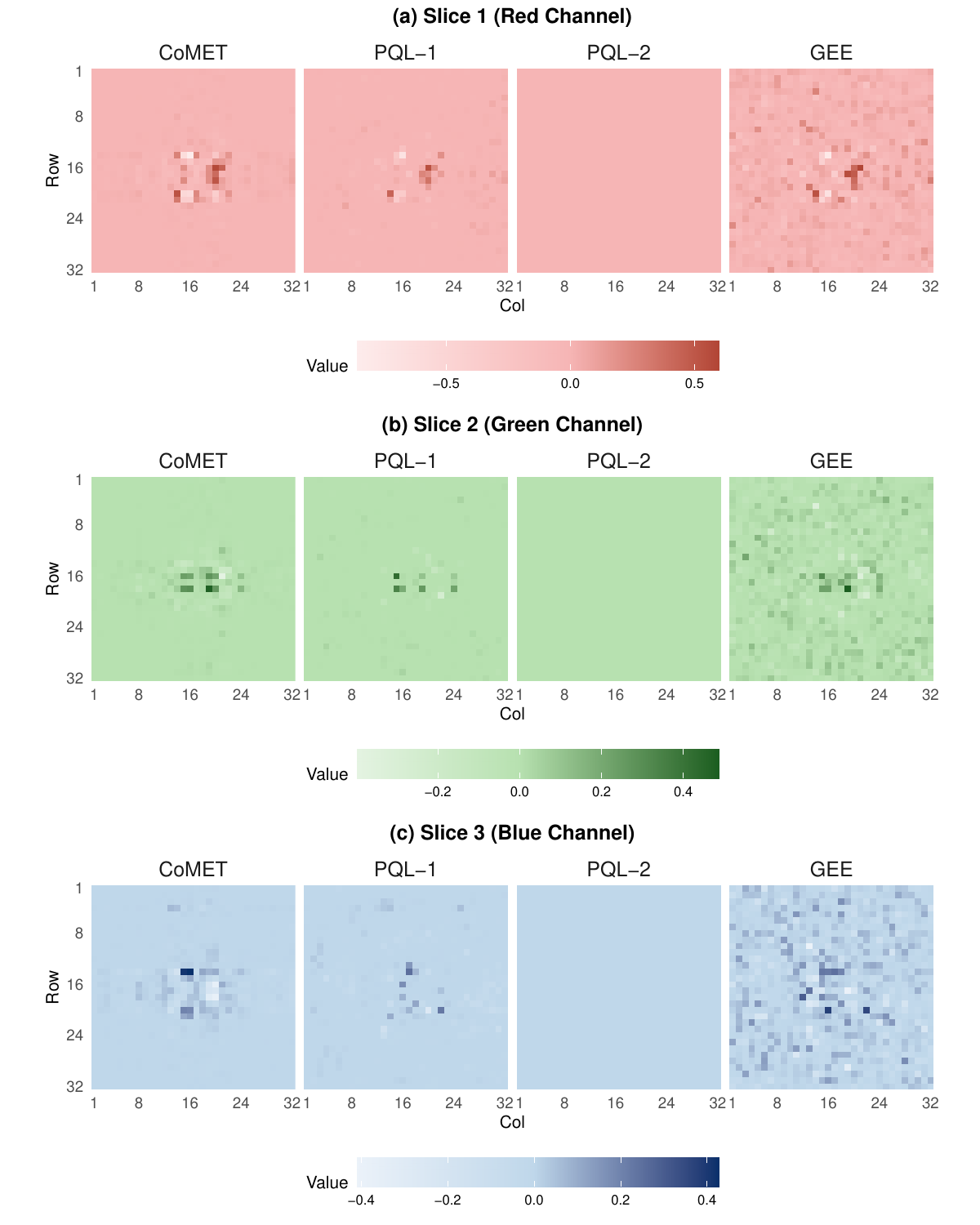}
    \caption{Estimated frontal slices of $32 \times 32 \times 3$ fixed-effect coefficient $\mathcal{B}$, corresponding to each of the 3 color channels, averaged over 25 random training splits of the LFW dataset. CoMET, the compressed mixed-effects model for tensors; PQL-1, the penalized quasi-likelihood approach of \citet{FanLi12}; and PQL-2, the penalized quasi-likelihood approach of \citet{Lietal21}; GEE, the penalized generalized estimating equations approach of \citet{2019_Zhang_etal}.}
    \label{fig:B_estimate_LFW3D}
\end{figure}

\section{Discussion}

The proposed compressed scalar-on-tensor mixed-model allows for several extensions beyond the setup in \eqref{eq:cme1}. For simplicity, we adopted a mode-wise covariance compression strategy in this work. Other tensor-structured projection strategies can be considered to simultaneously reduce the number of tensor modes and mode-specific dimensions \citep{2025_CompBTR_Craiu}. Beyond Gaussian random matrices, alternative constructions of the random projection matrices are possible, including sparse or structured random matrices comprising independent and identically distributed entries with mean zero and finite fourth-order moment \citep{MukDun20}. Additionally, we have used a single realization of the random projection matrices for computational scalability. The prediction risk of the CoMET model can be further improved by averaging over an ensemble of random projections \citep{Sla18}.

The CoMET model's compression strategy can be naturally extended to a semi-parametric mixed-model framework accounting for the temporal dependencies, including those in the DEAM data. Consider the following reformulation of the tensor mixed-effects model in \eqref{eq:Origmodel}
\begin{equation}\label{eq:semi-para}
    y_{ij} = \langle \mathcal{X}_{ij}, \mathcal{B} \rangle + \langle \mathcal{Z}(t_{ij}), \mathcal{A}_{i} \rangle + \epsilon_{ij}, \quad \mathcal{Z}(t) = \sum_{v = 1}^V \bm\phi_{1}^{(v)}(t) \circ \cdots \circ \bm \phi_{D}^{(v)}(t) , \quad \mathcal{A}_i \perp \epsilon_{ij},
\end{equation}
where the random-effects covariate tensors $\mathcal{Z}_{ij} = \mathcal{Z}(t_{ij})$s admit a rank-$V$ CP structure with every $v$-th rank-one component constructed using basis expansions of the observed time points $t_{ij}$s. The function $\bm\phi_{d}^{(v)}(\cdot) \in \mathbb{R}^{q_d}$ denotes a user-specified basis expansion to model the mode-$d$ covariance structure. This construction induces a smooth temporal dependence separable across different modes. As a result, CoMET's covariance compression strategy admits a straightforward extension to model temporally indexed tensor-valued acoustic features in the DEAM data for improved performance in music emotion recognition.

The compression paradigm introduced in this work naturally generalizes to a tensor-on-tensor modeling framework for multimodal applications. For instance, instead of reducing arousal and valence to a scalar response via their first principal component for the DEAM data (see Section \ref{Section:DEAMData}), we can model the two-dimensional emotion to examine how different acoustic features are associated with specific emotion dimensions. We are extending the current framework to accommodate such multivariate responses using a tensor-on-tensor mixed-effects model formulation based on the contracted tensor product \citep{2009_Kolda}.  

\section*{Declaration of the use of generative AI and AI-assisted technologies}
During the preparation of this work the authors used ChatGPT in order to check for grammatical errors. After using this tool the authors reviewed and edited the content as necessary and take full responsibility for the content of the publication.

\section*{Acknowledgement}
Sreya Sarkar was supported by the National Science Foundation grant DMS-1854667 for this work. Kshitij Khare's work in this project was partially supported by the National Science Foundation grant DMS-2506059. Sanvesh Srivastava was partially supported by the National Science Foundation grant DMS-2506058 and the National Institutes of Health grant R01 HD121993. An open-source implementation of CoMET is available in the R package \texttt{BayesCoMET} (\url{https://github.com/SreyaSarkar21/BayesCoMET}) 

\bibliographystyle{apalike}
\bibliography{ref}

\newpage
\begin{center}
    \huge \textbf{Supplementary Material for ``CoMET: A Compressed Bayesian Mixed-Effects Model for High-Dimensional Tensors"}
\end{center}
\appendix
\section{Theoretical Properties of Mean Square Prediction Risk of CoMET}\label{Theory_vecB}
\subsection{Proof of Theorem 3.1}\label{Section:ProofThm}
Consider the setup for establishing the theoretical properties of the compressed mixed model for tensor covariates of order three, without any loss of generality. Let $\mathcal{B}^0, \tau^2_0$ and $\Sigma^0 = \Sigma_{3}^{0} \otimes \Sigma_{2}^{0} \otimes \Sigma_{1}^{0}$ denote the true values of the parameters. The marginalized true data generating model and working covariance compression model for the $i$-th subject are
\begin{align}\label{eq:popln&comp_models}
    y_i &= X_i \bm\beta^0 + \epsilon_{0i}, \quad \epsilon_{0i} \sim \mathcal{N}_{m_i}(0, \tau_0^2 V_{0i}), \quad V_{0i} = Z_i (\Sigma_{3}^{0} \otimes \Sigma_{2}^{0} \otimes \Sigma_{1}^{0}) Z_i^{\T} + I_{m_i}, \nonumber \\
    y_i &= X_i \bm\beta + \epsilon'_{i}, \quad \epsilon'_{i} \sim \mathcal{N}_{m_i}(0, \tau_0^2 C_{i}), \quad C_{i} = Z_i S^{*\T} \Gamma^* R^*R^{* \T} \Gamma^{*\T} S^* Z_i^{\T} + I_{m_i}, \quad i = 1, \dots, n,
\end{align}
respectively, where we define
\begin{align}\label{eq:notations}
    y_i = \begin{pmatrix} y_{i1} \\ \vdots \\ y_{i m_{i}} \end{pmatrix} \in \RR^{m_i}, \quad X_i &= \begin{pmatrix} \text{vec}(\mathcal{X}_{i1})^{\T} \\ \vdots \\ \text{vec}(\mathcal{X}_{i m_{i}})^{\T} \end{pmatrix} \in \RR^{m_i \times p^*}, \quad Z_i = \begin{pmatrix} \text{vec}(\mathcal{Z}_{i1})^{\T} \\ \vdots \\ \text{vec}(\mathcal{Z}_{i m_{i}})^{\T} \end{pmatrix} \in \RR^{m_i \times q^*}, \nonumber \\
    S^* = S_3 \otimes S_2 \otimes S_1 &\in \RR^{k^* \times q^*}, \quad R^* = R_3 \otimes R_2 \otimes R_1 \in \RR^{k^* \times q^*}, \nonumber \\ \Gamma^* = \Gamma_3 \otimes \Gamma_2 \otimes \Gamma_1 &\in \RR^{k^* \times k^*}, \quad \bm\beta = \text{vec}(\mathcal{B}) \in \RR^{p^*}, \ \bm\beta^0 = \text{vec}(\mathcal{B}^0), \nonumber \\
    p^* = p_1 p_2 p_3, \quad & q^* = q_1 q_2 q_3, \quad k^* = k_1 k_2 k_3.
\end{align}

We now set forth the assumptions needed for deriving the theoretical guarantees. Let the error variance $\tau^2$ be known and equals the true value $\tau_0^2$. The support for the Gaussian prior on $\gamma_d$ is restricted to the set $\{\Gamma_d \in \RR^{k_d \times k_d}: \| \Gamma_d\| \leq b_d\}$ where $b_d$ is a universal constant for $d = 1,2,3$, and $\| \cdot \|$ is the operator norm of a matrix. Using the fact that $\|A \otimes B\| = \|A\| \|B\|$ for any two matrices $A$ and $B$, we define the set of matrices $\mathcal{G} = \{\Gamma^* \in \RR^{k^* \times k^*} : \|\Gamma^*\| \leq b^*\}$, where $b^* = b_1 b_2 b_3$. The entries of the tensors $\mathcal{X}_{ij}$ and $\mathcal{Z}_{ij}$ are mutually independent Gaussian random variables with mean zero and variance $\sigma_X^2$ and $\sigma_Z^2$, respectively, where $\sigma_X^2$ is a constant, and $\sigma_Z^2$ satisfies $n\sigma_Z^4 \left(\prod_{d=1}^3q_d^6/k_d^4+ \prod_{d=1}^3 m_{\max}^2/k_d^4 \right) = O(1)$, with $m_{\max} = \max(m_1, \dots, m_n)$.

For simplicity, we assume that $\mathcal{B}$ is full-rank as discussed in the main manuscript, and assign Horseshoe prior directly on the elements of $\text{vec}(\mathcal{B})$ as follows:
\begin{align}\label{hsprior_vecB}
   \mathcal{B}_{j_1 j_2 j_3} \mid \lambda_{j_1 j_2 j_3}^2, \delta^2, \quad \tau^2 & \sim \mathcal{N}(0, \lambda_{j_1 j_2 j_3}^2 \delta^2 \tau^2 ), \quad \lambda_{j_1 j_2 j_3} \sim \mathcal{C}^{+}(0, 1), \quad \delta \sim \mathcal{C}^{+}(0, 1), 
\end{align}
independently for $j_d = 1, \ldots, p_d$ and $d = 1,2,3$. The squared global shrinkage parameter $\delta^2$ controlling the overall sparsity in $\bm\beta$ can depend on $N, p^*, q^*$, and the squared local shrinkage parameter $\lambda_{j_1 j_2 j_3}^2$ for each cell of $\mathcal{B}$ is assumed to be supported on a compact domain $[a, 1/a]$ for a universal constant $a$.

Define $C = \diag(C_1, \ldots, C_n)$, $y \in \RR^N$ with $i$-th row block $y_i$, and $X \in \RR^{N \times p^*}$ with $i$-th row block $X_i$. Then the compressed conditional posterior density of $\bm\beta$ given $\Lambda, \delta^2, C$ is
\begin{align}
  \label{eq:cond_vecBSupp}
  \bm\beta \mid \Lambda, \delta^2, C, y, X  \sim \mathcal{N} \left( \{ X^{\T} C^{-1} X + (\delta^{2} \Lambda)^{-1}\}^{-1} X^{\T} C^{-1} y, \tau_0^{2} \{ X^{\T} C^{-1} X + (\delta^{2} \Lambda)^{-1}\}^{-1}  \right),
\end{align}
where we define $\Lambda = \diag(\lambda_{111}^2, \dots, \lambda_{p_1p_2p_3}^2)$.

The prediction risk of the covariance compressed posterior is defined as
\begin{align}\label{eq:vecB_predriskSupp}
    \frac{1}{N}\EE \|X\bm{\beta}^0 - X \bar{\bm{\beta}}(\phi)\|_2^2 &= \frac{1}{N}\EE_{R,S,X,Z} \EE_y \{ \EE_{\phi \mid y} \|X\bm{\beta}^0 - X \bar{\bm{\beta}}(\phi)\|_2^2 \},
\end{align}
where the conditional posterior mean of $\bm\beta$ in \eqref{eq:cond_vecB_supp} is denoted as $\bar{\bm\beta}(\phi)$ with $\phi = \{\Gamma, \Lambda, \delta^2\}$, $\EE_{R,S, X, Z}$, $\EE_{y}$, and $\EE_{\phi \mid y}$ respectively denote the expectations with respect to the distributions of $(R_1, R_2, R_3, S_1, S_2, S_3, X, Z_1, \dots, Z_n)$, $y$ and the conditional distribution of $\phi$ given $y$. 

We restate Theorem 3.1 from Section~\ref{Section:tensorTheory} of the main manuscript below, followed by a detailed proof.

\begin{theorem}
\label{vecB_TheoremPredRiskSupp}
If the assumptions (A1)-(A3) in Section \ref{Section:tensorTheory} of the main manuscript hold, $p^* = o(N)$, $k^{*2}\log k^* = o(p^*)$, $k^{*2}\log \log N = o(p^*)$ and $\|\bm{\beta}^0 \|^2 = o(N)$, then the posterior predictive risk satisfies
\begin{align*}
   \frac{1}{N}\EE \| X \bm{\beta}^0 - X \bar{\bm{\beta}} (\phi) \|_2^2 &\leq \EE_{X}\left\{ \kappa^4(X)\right\} \Bigg[ \frac{2\| \bm{\beta}^0\|_2^{2}}{N}\left\{O\left(nb^{*8}\sigma_Z^8 \left(\prod_{d=1}^3 \frac{q_d^{12}}{k_d^8} + \prod_{d=1}^3\frac{m_{\max}^4}{k_d^8} \right) + nb^{*8}\sigma_Z^8 \right)\right\} \\
    &\qquad \qquad \qquad \qquad + \frac{4\tau_{0}^{2}}{a^{4}} \Bigg\{ \left(\frac{N-p^*}{N\log N}\right)O\left( nb^{*5}\sigma_Z^6 \left(\prod_{d=1}^3\frac{q_d^9}{k_d^6} + \prod_{d=1}^3 \frac{m_{\max}^3}{k_d^6}\right)\right) \\
        &\qquad \qquad \qquad \qquad + \left(\frac{p^* + 4e^{-p^*/8}}{N}\right)O\left(nb^{*3}\sigma_Z^4\left(\prod_{d=1}^3\frac{q_d^6}{k_d^4}+ \prod_{d=1}^3 \frac{m_{\max}^2}{k_d^4} \right)\right)\Bigg\}\Bigg] \\
    &= o(1),
\end{align*}
where $\kappa(X)$ denotes the condition number of $X$.
\end{theorem}

\begin{proof}
We begin with decomposing the squared loss with respect to the compressed posterior mean $\bar{\bm{\beta}} (\phi)$ as follows:
    \begin{align}\label{eq:vecB_mse}
  \| X \bm{\beta}^0 - X \bar{\bm{\beta}} (\phi) \|_2^2 &=  \big \|  X \bm{\beta}^0 -  X \{  X^{\T} C^{-1} X + (\delta^{2} \Lambda)^{-1}\}^{-1} X^{\T} C^{-1} y \big \|_2^2 \nonumber \\
    &=
    \big \|  X \bm{\beta}^0 -  X \{  X^{\T} C^{-1} X + (\delta^{2} \Lambda)^{-1}\}^{-1} X^{\T} C^{-1} (X \bm{\beta}^0 + \epsilon_0) \big \|_2^2 \nonumber  \\
    &\leq 2 \big \|  \left[ X - X \{  X^{\T} C^{-1} X + (\delta^{2} \Lambda)^{-1}\}^{-1} X^{\T} C^{-1} X \right] \bm{\beta}^0 \big \|_2^2 + \nonumber \\
    &\quad \; 2
    \big \|X \{  X^{\T} C^{-1} X + (\delta^{2} \Lambda)^{-1}\}^{-1} X^{\T} C^{-1} \epsilon_0 \big \|_2^2.
\end{align}
Lemma A.3 in \cite{2025_SarKhSri} implies that \eqref{eq:vecB_mse} reduces to
\begin{align}
  \label{eq:vecB_mse1}
   \| X \bm{\beta}^0 - X \bar{\bm{\beta}} (\phi) \|_2^2 &\leq 2 \| C (X \delta^2 \Lambda X^\T
   + C)^{-1} X \bm{\beta}^0 \|_2^{2} + 2 \| X\delta^2 \Lambda X^\T (X \delta^2 \Lambda
   X^\T + C)^{-1} \epsilon_0 \|_2^2 \nonumber \\
   &\equiv 2 T_1 + 2 T_2.
\end{align}
Lemma \ref{vecB_lemt1} implies that $T_1$ in \eqref{eq:vecB_mse1} satisfies
\begin{equation}\label{vecB_boundEt1}
    2\EE(T_1) \leq 2O\left(nb^{*8}\sigma_Z^8 \left(\prod_{d=1}^3 \frac{q_d^{12}}{k_d^8} + \prod_{d=1}^3\frac{m_{\max}^4}{k_d^8} \right) + nb^{*8}\sigma_Z^8 \right) \EE_{X} \left\{
\kappa^4(X) \right\} \| \bm{\beta}^0\|_2^{2},
\end{equation}
where $\kappa(X) = s_{\max}(X) / s_{\min}(X)$ is the condition number, and $s_{\max}(X)$ and $s_{\min}(X)$ are the maximum and minimum singular values of $X$. Lemma \ref{vecB_lemt2} implies that $T_2$ in \eqref{eq:vecB_mse1} satisfies
\begin{align}\label{vecB_boundEt2}
    2\EE (T_2) &\leq \frac{4\tau_{0}^{2}}{a^{4}} \; \EE_{X}\bigl\{\kappa^4(X)\bigr\} \; \Bigg\{ \left(\frac{N-p^*}{\log N}\right)O\left( nb^{*5}\sigma_Z^6 \left(\prod_{d=1}^3\frac{q_d^9}{k_d^6} + \prod_{d=1}^3 \frac{m_{\max}^3}{k_d^6}\right)\right) \nonumber\\
&\quad + (p^* + 4e^{-p^*/8})O\left(nb^{*3}\sigma_Z^4\left(\prod_{d=1}^3\frac{q_d^6}{k_d^4}+ \prod_{d=1}^3 \frac{m_{\max}^2}{k_d^4} \right)\right)\Bigg\} \nonumber \\
\end{align}

Combining \eqref{vecB_boundEt1} and \eqref{vecB_boundEt2}, we have the prediction risk
Combining \eqref{vecB_boundEt1} and \eqref{vecB_boundEt2}, we have the prediction risk
\begin{align*}
    \frac{1}{N}\EE \| X \bm{\beta}^0 - X \bar{\bm{\beta}} (\phi) \|_2^2 &\leq  \EE_{X}\left\{ \kappa^4(X)\right\} \Bigg[ \frac{2\| \bm{\beta}^0\|_2^{2}}{N}\left\{O\left(nb^{*8}\sigma_Z^8 \left(\prod_{d=1}^3 \frac{q_d^{12}}{k_d^8} + \prod_{d=1}^3\frac{m_{\max}^4}{k_d^8} \right) + nb^{*8}\sigma_Z^8 \right)\right\} \\
    &\qquad \qquad \qquad \qquad + \frac{4\tau_{0}^{2}}{a^{4}} \Bigg\{ \left(\frac{N-p^*}{N\log N}\right)O\left( nb^{*5}\sigma_Z^6 \left(\prod_{d=1}^3\frac{q_d^9}{k_d^6} + \prod_{d=1}^3 \frac{m_{\max}^3}{k_d^6}\right)\right)\\
        &\qquad \qquad \qquad \qquad + \left(\frac{p^* + 4e^{-p^*/8}}{N}\right)O\left(nb^{*3}\sigma_Z^4\left(\prod_{d=1}^3\frac{q_d^6}{k_d^4}+ \prod_{d=1}^3 \frac{m_{\max}^2}{k_d^4} \right)\right)\Bigg\}\Bigg] \\
    &= o(1),
\end{align*}
since under Assumption (A3), $\EE_{X} \left\{ \kappa^4(X) \right\} = 1 + o(1)$; see \cite{KSS:2025} for details. The theorem is proved.
\end{proof}

\subsection{Upper Bound for $T_2$}\label{Section:vecB_ProofT2}

\begin{lemma}\label{vecB_lemt2}
Under the assumptions of Theorem 1 in the main manuscript,
\begin{align*}
\EE (T_2) &\leq \frac{2\tau_{0}^{2}}{a^{4}} \; \EE_{X}\bigl\{\kappa^4(X)\bigr\} \; \Bigg\{ \left(\frac{N-p^*}{\log N}\right)O\left( nb^{*5}\sigma_Z^6 \left(\prod_{d=1}^3\frac{q_d^9}{k_d^6} + \prod_{d=1}^3 \frac{m_{\max}^3}{k_d^6}\right)\right) \\
&\quad + (p^* + 4e^{-p^*/8})O\left(nb^{*3}\sigma_Z^4\left(\prod_{d=1}^3\frac{q_d^6}{k_d^4}+ \prod_{d=1}^3 \frac{m_{\max}^2}{k_d^4} \right)\right)\Bigg\}
\end{align*}

\end{lemma}

\begin{proof}
The term $T_{2} = \| X\delta^2 \Lambda X^\T (X \delta^2 \Lambda
   X^\T + C)^{-1} \epsilon_0 \|_2^2$ depends on $X\delta^2 \Lambda X^\T$ and $X \delta^2 \Lambda
   X^\T + C$. Let the singular value decomposition of $X$ and
spectral decomposition of $X \delta^2 \Lambda X^\T$ be
\begin{align} \label{eq:vecB_T2_1}
  X = U
  \begin{bmatrix}
  D \\
  0_{N-p^* \times p^*}
  \end{bmatrix}
  V^\T, \quad
  X \delta^2 \Lambda X^\T = U
\begin{bmatrix}
  D V^\T \delta^2 \Lambda V D & 0_{p^* \times N - p^*} \\
  0_{N-p^* \times p^*} & 0_{N - p^* \times N - p^*}
  \end{bmatrix}
  U^\T,
\end{align}
where $N > p^*$, $U$ and $V$ are $N \times N$ and $p^* \times p^*$
orthonormal matrices, and $D$ is a $p^* \times p^*$ diagonal matrix of singular
values of $X$. If $\tilde D = D V^\T \delta^2 \Lambda V D$ and $\tilde C = U^{\T} C U$, then the second
decomposition in (\ref{eq:vecB_T2_1}) implies that
\begin{align}\label{eq:vecB_T2_2}
  X \delta^2 \Lambda X^\T
+ C =   X \delta^2 \Lambda X^\T
+ U U^{\T} C U U^{\T} = U
\begin{bmatrix}
  \tilde D + \tilde C_{11} & \tilde C_{12} \\
  \tilde C_{12}^{\T} & \tilde C_{22}
  \end{bmatrix}
  U^\T,
\end{align}
where $\tilde C_{ij}$ is the $ij$-th block of the $2 \times 2$ block matrix
$\tilde C$. Substituting (\ref{eq:vecB_T2_1}) and (\ref{eq:vecB_T2_2}) in $T_{2}$ implies that
\begin{align}
\label{eq:vecB_T2_3}
T_{2} &=
\left \|  U
\begin{bmatrix}
  \tilde D & 0_{p^* \times N - p^*} \\
  0_{N-p^* \times p^*} & 0_{N - p^* \times N - p^*}
  \end{bmatrix}
  U^\T U
\begin{bmatrix}
  \tilde D + \tilde C_{11} & \tilde C_{12} \\
  \tilde C_{12}^{\T} & \tilde C_{22}
  \end{bmatrix}^{{-1}}
  U^\T \epsilon_0 \right  \|_2^{2} \nonumber \\
  &=
\left \|
\begin{bmatrix}
  \tilde D & 0_{p^* \times N - p^*} \\
  0_{N-p^* \times p^*} & 0_{N - p^* \times N - p^*}
  \end{bmatrix}
\begin{bmatrix}
  \tilde D + \tilde C_{11} & \tilde C_{12} \\
  \tilde C_{12}^{\T} & \tilde C_{22}
  \end{bmatrix}^{{-1}}
  \tilde \epsilon \right  \|_2^{2},
\end{align}
where the first $U$ leaves the norm unchanged because $U^{\T} U = I_{N}$ and
$\tilde \epsilon = U^{\T} \epsilon_0$ is distributed as $\Ncal(0, \tau_{0}^{2}
I_{N})$ because $U$ is orthonormal. Furthermore, $C$ is positive definite because $C=\text{diag}(C_1, \dots, C_n)$ with $C_i(\Gamma^*) = Z_i S^{*\T}
\Gamma^* R^*  \left( Z_i S^{*\T} \Gamma^* R^* \right)^{\T} + I_{m_i}$, which implies that $\tilde C = U^{\T}C U$ is also a positive definite matrix due to orthonormality of $U$. Hence, $\tilde C_{22}^{-1}$ exists. If $F = \tilde D + \tilde C_{11} - \tilde
C_{12} \tilde C_{22}^{-1} \tilde C_{12}^{\T}$, then the block matrix inversion
formula \citep{Har97} gives
\begin{align}\label{eq:blockmatinv}
&\begin{bmatrix}
  \tilde D + \tilde C_{11} & \tilde C_{12} \\
  \tilde C_{12}^{\T} & \tilde C_{22}
  \end{bmatrix}^{{-1}} =
\begin{bmatrix}
  F^{-1} &  - F^{{-1}} \tilde C_{12} \tilde C_{22}^{-1} \\
  - \tilde C_{22}^{{-1}} \tilde C_{12}^{\T} F^{-1}  & \tilde C_{22}^{-1} + \tilde
  C_{22}^{-1} C_{12}^{\T} F^{{-1}} \tilde C_{12} \tilde C_{22}^{-1}
\end{bmatrix}.
\end{align}
Substituting this identity in (\ref{eq:vecB_T2_3}) implies that
\begin{align}\label{eq:vecB_T2_4}
\begin{bmatrix}
  \tilde D & 0_{p^* \times N - p^*} \\
  0_{N-p^* \times p^*} & 0_{N - p^* \times N - p^*}
  \end{bmatrix}
&\begin{bmatrix}
  \tilde D + \tilde C_{11} & \tilde C_{12} \\
  \tilde C_{12}^{\T} & \tilde C_{22}
  \end{bmatrix}^{{-1}} =
\begin{bmatrix}
  \tilde D F^{-1} &  - \tilde D  F^{{-1}} \tilde C_{12} \tilde C_{22}^{-1} \\
   0_{N-p^* \times p^*} & 0_{N - p^* \times N - p^*}
\end{bmatrix}.
\end{align}
Using (\ref{eq:vecB_T2_4}), (\ref{eq:vecB_T2_3}) is written as the quadratic form
\begin{align}
\label{eq:vecB_T2_5}
  T_{2} &= \tilde \epsilon^{\T}
  \begin{bmatrix}
   F^{-1} \tilde D &  0_{p^* \times N- p^*} \\
   - \tilde C_{22}^{-1} \tilde C_{12}^{\T} F^{{-1}} \tilde D      & 0_{N - p^* \times N - p^*}
  \end{bmatrix}
  \begin{bmatrix}
  \tilde D F^{-1} &  - \tilde D  F^{{-1}} \tilde C_{12} \tilde C_{22}^{-1} \\
   0_{N-p^* \times p^*} & 0_{N - p^* \times N - p^*}
\end{bmatrix}  \tilde \epsilon \nonumber \\
&= \tilde \epsilon^{\T}
  \begin{bmatrix}
   F^{-1} \tilde D^{2} F^{-1} &  - F^{{-1}} \tilde D^{2}  F^{{-1}} \tilde C_{12} \tilde C_{22}^{-1} \\
   - \tilde C_{22}^{-1} \tilde C_{12}^{\T} F^{{-1}} \tilde D^{2} F^{-1}      &
   \tilde C_{22}^{-1} \tilde C_{12}^{\T} F^{{-1}} \tilde D^{2} F^{-1} \tilde C_{12} \tilde C_{22}^{-1}
  \end{bmatrix}
  \tilde \epsilon,
\end{align}
which is further simplified by partitioning $\tilde \epsilon$ into two blocks as $\tilde \epsilon^{\T} = (\tilde \epsilon_{1}^{\T}, \tilde
\epsilon_{2}^{\T})$ of dimensions $p^* \times 1$ and $(N -
p^*) \times 1$, respectively.  Using Cauchy-Schwartz inequality for quadratic forms in
\eqref{eq:vecB_T2_5} gives
\begin{align}
\label{eq:vecB_T2_6}
  T_{2} \leq  2 \tilde \epsilon_{1}^{\T} F^{-1} \tilde D^{2} F^{-1} \tilde
  \epsilon_{1} + 2 \tilde \epsilon_{2}^{\T} \tilde C_{22}^{-1} \tilde C_{12}^{\T}
  F^{{-1}} \tilde D^{2} F^{-1} \tilde C_{12} \tilde C_{22}^{-1} \tilde \epsilon_{2}.
\end{align}
The first term in (\ref{eq:vecB_T2_6}) satisfies
\begin{align}
\label{eq:vecB_T2_7}
\tilde \epsilon_{1}^{\T} F^{-1} \tilde D^{2} F^{-1} \tilde \epsilon_{1}
\leq  \lambda_{\max}(\tilde D^{2}) \;   \tilde \epsilon_{1}^{\T} F^{-1} F^{-1}
\tilde \epsilon_{1} = \lambda^{2}_{\max}(\tilde D) \;   \tilde \epsilon_{1}^{\T} F^{-2} \tilde \epsilon_{1},
\end{align}
where $\lambda_{\max}(\tilde D)$ is the maximum eigen value of $\tilde D$ and $\lambda_{\max}(\tilde D^{2}) = \lambda^{2}_{\max}(\tilde D)$ because $\tilde D$ is a symmetric matrix. 
The Schur complement of $\tilde C_{22}$ in $\tilde C$ is $\tilde C_{11} - \tilde C_{12}
\tilde C_{22}^{-1} \tilde C_{12}^{\T}$, so $\lambda_{\min}(\tilde C_{11} - \tilde C_{12}
\tilde C_{22}^{-1} \tilde C_{12}^{\T}) \geq 0$, where $\lambda_{\min}(A)$ is the minimum eigen value of $A$. This implies that
\begin{align}
\label{eq:vecB_T2_8}
  F^{{-1}} = (\tilde D + \tilde C_{11} - \tilde C_{12} \tilde C_{22}^{-1} \tilde
C_{12}^{\T})^{{-1}} \preceq \tilde D^{-1} \preceq \lambda_{\min}^{-1}(\tilde D) \; I_{p^*},
  \quad F^{{-2}} \preceq \lambda_{\min}^{-2}(\tilde D) \; I_{p^*},
\end{align}
where $A \preceq B$ for two positive semi-definite matrices $A$ and $B$ implies
that $\lambda_{\min}(B - A) \geq 0$. Using \eqref{eq:vecB_T2_8} in \eqref{eq:vecB_T2_7} gives
\begin{align}
\label{eq:vecB_T2_9}
\tilde \epsilon_{1}^{\T} F^{-1} \tilde D^{2} F^{-1} \tilde \epsilon_{1}
\leq  \frac{\lambda_{\max}^{2}(\tilde D)}{\lambda_{\min}^{2}(\tilde D)} \;   \tilde \epsilon_{1}^{\T} \tilde \epsilon_{1}.
\end{align}

We simplify $\tilde D$ to find upper and lower bounds for
$\lambda_{\max}(\tilde D)$ and $\lambda_{\min}(\tilde D)$. The Assumption (A1) on the support of $\Lambda$ implies that $\lambda_{j_1j_2j_3}^2 \in
[a, 1/a]$ for $j_d = 1, \dots, p_d$ and $ d = 1,2,3$, so 
\begin{align}
\label{eq:vecB_T2_10}
\tilde D &= D V^\T \delta^2 \Lambda V D \preceq
a^{-1} \delta^2  D V^\T  V D  = a^{-1} \delta^2
D^{2}, \nonumber\\
 \tilde D &= D V^\T \delta^2 \Lambda V D \succeq
 a \delta^2  D V^\T  V D = a \delta^2 D^{2}.
\end{align}
Furthermore, (\ref{eq:vecB_T2_1}) implies that $D_{11} = s_{\max}(X)$ and $D_{p^* p^*} =
s_{\min}(X)$, where $s_{\max}(X)$ and $s_{\min}(X)$ are the maximum and minimum singular values of $X$. Using this in \eqref{eq:vecB_T2_10} gives
\begin{align}
\label{eq:vecB_T2_11}
\lambda_{\max}(\tilde D) \leq  a^{-1} \delta^{2} s^{2}_{\max}(X), \quad
\lambda_{\min}(\tilde D) \geq a \delta^{2}  s^{2}_{\min}(X), \quad
\frac{\lambda_{\max}^{2}(\tilde D)}{\lambda_{\min}^{2}(\tilde D)} \leq \frac{1}{a^{4}} \frac{s_{\max}^{4}(X)}{s_{\min}^{4}(X)}.
\end{align}
Substituting the last inequality in \eqref{eq:vecB_T2_9} implies that
\begin{align}
\label{eq:vecB_T2_12}
\tilde \epsilon_{1}^{\T} F^{-1} \tilde D^{2} F^{-1} \tilde \epsilon_{1}
\leq  \frac{1}{a^{4}} \; \kappa^4(X)  \; \tilde \epsilon_{1}^{\T} \tilde \epsilon_{1}, 
\end{align}
where $\kappa(X) = {s_{\max}(X)}/{s_{\min}(X)}$ denotes the condition number of $X$. Finally, we take expectations on both sides in (\ref{eq:vecB_T2_12}). The term $\tilde
\epsilon_{1}^{\T} \tilde \epsilon_{1}$ only depends on the distribution of
$\tilde \epsilon$ and is independent of the distribution of $X$, implying that
\begin{align}
\label{eq:vecB_T2_13}
\EE (\tilde \epsilon_{1}^{\T} F^{-1} \tilde D^{2} F^{-1} \tilde \epsilon_{1})
\leq  \frac{1}{a^{4}} \; \EE_{X} \{\kappa^4(X)\}  \;  \EE_{\tilde \epsilon} (\tilde
\epsilon_{1}^{\T} \tilde \epsilon_{1}) = \frac{\tau_{0}^{2}}{a^{4}} \; \EE_{X} \{\kappa^4(X)\}  \;  p^*,
\end{align}
where $\EE_{\tilde \epsilon} (\tilde \epsilon_{1}^{\T} \tilde \epsilon_{1}) = \tau_{0}^2 p^*$
because $\epsilon_{1} / \tau_{0}$ is distributed as $\Ncal(0, I_{p^*})$.

Consider the second term in (\ref{eq:vecB_T2_6}). The bounds in (\ref{eq:vecB_T2_9}) and
(\ref{eq:vecB_T2_11}) imply that
\begin{align}
\label{eq:vecB_T2_14}
\tilde \epsilon_{2}^{\T} \tilde C_{22}^{-1} \tilde C_{12}^{\T} F^{{-1}} \tilde
D^{2} F^{-1} \tilde C_{12} \tilde C_{22}^{-1} \tilde \epsilon_{2} &\leq
\frac{\lambda_{\max}^{2}(\tilde D)}{\lambda_{\min}^{2}(\tilde D)} \;   \tilde
\epsilon_{2}^{\T} \tilde C_{22}^{-1} \tilde C_{12}^{\T}  \tilde C_{12} \tilde
C_{22}^{-1} \tilde \epsilon_{2} \nonumber \\
&\leq \frac{1}{a^{4}} \kappa^4(X) \;   \tilde
\epsilon_{2}^{\T} \tilde C_{22}^{-1} \tilde C_{12}^{\T}  \tilde C_{12} \tilde
C_{22}^{-1} \tilde \epsilon_{2}.
\end{align}
The quadratic form $\tilde
\epsilon_{2}^{\T} \tilde C_{22}^{-1} \tilde C_{12}^{\T}  \tilde C_{12} \tilde
C_{22}^{-1} \tilde \epsilon_{2}$ is independent of the distribution of $X$, so taking expectation on both sides of (\ref{eq:vecB_T2_14}) gives
\begin{align}
\label{eq:vecB_T2_15}
    \EE (\tilde \epsilon_{2}^{\T} \tilde C_{22}^{-1} \tilde C_{12}^{\T} F^{{-1}} \tilde
D^{2} F^{-1} \tilde C_{12} \tilde C_{22}^{-1} \tilde \epsilon_{2}) &\leq \frac{1}{a^4}\EE_{X}\left\{\kappa^4(X) \right\} \EE(\tilde
\epsilon_{2}^{\T} \tilde C_{22}^{-1} \tilde C_{12}^{\T}  \tilde C_{12} \tilde
C_{22}^{-1} \tilde \epsilon_2).
\end{align}
Let $\EE_{R,S, Z}, \ \EE_{\tilde \epsilon \mid R, S, Z}$ and $\EE_{\Gamma^* \mid y}$ respectively denote the expectations with respect to the distributions of $(R_1, R_2, R_3, S_1, S_2, S_3, Z_1, \dots, Z_n)$, $\tilde \epsilon$ given $(R_1, R_2, R_3, S_1, S_2, S_3, Z_1, \dots, Z_n)$, and $\Gamma^*$ given $y$. Then, the second expectation on the right of \eqref{eq:vecB_T2_15} is
\begin{align}
\label{eq:vecB_T2_16}
\EE(\tilde \epsilon_{2}^{\T} \tilde C_{22}^{-1} \tilde C_{12}^{\T}  \tilde
C_{12} \tilde C_{22}^{-1} \tilde \epsilon_{2}) &=\EE_{R,S,Z} \EE_{\tilde \epsilon \mid R,S,Z} \EE_{\Gamma^* \mid y}
(\tilde \epsilon_{2}^{\T} \tilde C_{22}^{-1} \tilde C_{12}^{\T}  \tilde C_{12}
\tilde C_{22}^{-1} \tilde \epsilon_{2})\nonumber \\
&\leq \EE_{R,S,Z}\EE_{\tilde \epsilon \mid R,S,Z} \{\sup_{{\Gamma^*: \| \Gamma^* \| \leq b^*}} (\tilde
\epsilon_{2}^{\T} \tilde C_{22}^{-1} \tilde C_{12}^{\T}  \tilde C_{12} \tilde
C_{22}^{-1} \tilde \epsilon_{2}) \EE_{\Gamma^* \mid y} (1) \} \nonumber \\
& =  \EE_{R,S,Z}\EE_{\tilde \epsilon \mid R,S,Z} \{\sup_{{\Gamma^*: \| \Gamma^* \| \leq b^*}} (\tilde
\epsilon_{2}^{\T} \tilde C_{22}^{-1} \tilde C_{12}^{\T}  \tilde C_{12} \tilde
C_{22}^{-1} \tilde \epsilon_{2}) \},
\end{align}
where $b^*$ is a universal positive constant defined in Assumption (A3). 

We use a covering argument to bound the last term involving the supremum. Lemma \ref{auxlem1} shows that there is a $1/ \log N$-net, denoted as $\tilde \Gcal$, for the set $\Gcal = \{ \Gamma^* \in \RR^{k^* \times k^*}: \| \Gamma^* \| \leq b^*\}$ such that the cardinality of $\tilde \Gcal$ is $\lceil 2b^*k^*\log N \rceil^{k^{*2}}$. This implies that for any $\Gamma^* \in \Gcal$, there is a  $\tilde \Gamma^* \in \tilde \Gcal$ such that $\| \Gamma^* - \tilde \Gamma^* \| \leq 1 / \log N$. Let  $A(\Gamma^*) = \tilde C_{12}(\Gamma^*) \tilde C_{22} (\Gamma^*)^{-1}$ be the matrix that appears in \eqref{eq:vecB_T2_16}, where the notation emphasizes the dependence on $\Gamma^*$. We show that the norms of terms involving $A(\Gamma^*)$ for any $\Gamma^* \in \Gcal$ are controlled using $A(\tilde \Gamma^*)$ for  $\tilde \Gamma^* \in \tilde \Gcal$. Specifically,
\begin{align}\label{eq:vecB_T2_17}
\tilde \epsilon_{2}^{\T} \left\{ A(\Gamma^*)^{\T} A(\Gamma^*) - A(\tilde \Gamma^*)^{\T} A(\tilde \Gamma^*)\right\} \tilde \epsilon_{2} &\leq \| A(\Gamma^*)^{\T} A(\Gamma^*) - A(\tilde \Gamma^*)^{\T} A(\tilde \Gamma^*) \| \| \tilde \epsilon_{2} \|^2 \nonumber\\
&\overset{(i)}{\leq} 4b^* (2\breve{c} + 3 b^{*2} \breve{c}^2 + b^{*4} \breve{c}^3)\, \| \Gamma^* - \tilde \Gamma^*\| \nonumber\\
&\overset{(ii)}{\leq} 4b^* (2\breve{c} + 3 b^{*2} \breve{c}^2 + b^{*4} \breve{c}^3) /\log N
\end{align} 
where $(i)$ follows from Lemma \ref{auxlem2}, $\breve{c}$ is defined in Lemma \ref{auxlem2},  and $(ii)$ follows because $\tilde \Gcal$ is a $1/\log N$-net of $\Gcal$.

Using  \eqref{eq:vecB_T2_17} in  \eqref{eq:vecB_T2_16} gives 
\begin{align}
\label{eq:vecB_T2_18}
\sup_{\Gamma^*: \| \Gamma^* \| \leq b^*} (\tilde
\epsilon_{2}^{\T} \tilde C_{22}^{-1} \tilde C_{12}^{\T}  \tilde C_{12} \tilde
C_{22}^{-1} \tilde \epsilon_{2}) &=
\sup_{{\Gamma^*: \| \Gamma^* \| \leq b^*}} \tilde \epsilon_{2}^{\T} \Bigl[
 \left\{ A(\Gamma^*)^{\T} A(\Gamma^*) - A(\tilde \Gamma^*)^{\T} A(\tilde \Gamma^*)\right\} \nonumber \\
 &\quad + \left\{ A(\tilde \Gamma^*)^{\T} A(\tilde \Gamma^*)\right\} \Bigr] \tilde \epsilon_{2} \nonumber\\
&\leq \frac{4b^*}{\log N} (2\breve{c} + 3 b^{*2} \breve{c}^2 + b^{*4} \breve{c}^3)  \| \tilde \epsilon_{2} \|^{2} + \underset{\tilde \Gamma^* \in \tilde \Gcal}{\max} \; \tilde \epsilon_{2}^{\T} \left\{ A(\tilde \Gamma^*)^{\T} A(\tilde \Gamma^*)\right\} \tilde \epsilon_{2},\nonumber\\
\EE_{\tilde \epsilon \mid R,S, Z} \left\{ \sup_{{\Gamma^*: \| \Gamma^* \| \leq b^*}} (\tilde
\epsilon_{2}^{\T} \tilde C_{22}^{-1} \tilde C_{12}^{\T}  \tilde C_{12} \tilde
C_{22}^{-1} \tilde \epsilon_{2}) \right\} &\leq   \frac{4b^*}{\log N} (2\breve{c} + 3 b^{*2} \breve{c}^2 + b^{*4} \breve{c}^3)   
 \EE_{\tilde \epsilon \mid R,S, Z} \| \tilde \epsilon_{2} \|^{2} + \nonumber\\
&\qquad \EE_{\tilde \epsilon \mid R, S, Z} \underset{\tilde \Gamma^* \in \tilde \Gcal}{\max} \; \tilde \epsilon_{2}^{\T} \left\{ A(\tilde \Gamma^*)^{\T} A(\tilde \Gamma^*)\right\} \tilde \epsilon_{2} \nonumber\\
&\overset{(i)}{=} \frac{4b^*\tau_0^2(N-p^*)}{\log N} (2\breve{c} + 3 b^{*2} \breve{c}^2 + b^{*4} \breve{c}^3)  
    + \nonumber\\
&\qquad \EE_{\tilde \epsilon \mid R, S, Z} \underset{\tilde \Gamma^* \in \tilde \Gcal}{\max} \; \tilde \epsilon_{2}^{\T} \left\{ A(\tilde \Gamma^*)^{\T} A(\tilde \Gamma^*)\right\} \tilde \epsilon_{2},
\end{align}
where $(i)$ follows because $ \epsilon_{2} / \tau_{0}$ is distributed as $\Ncal(0, I_{N-p^*})$. 
For every $\tilde \Gamma^* \in \tilde \Gcal$, the rank of $A(\tilde \Gamma^*)  = \tilde C_{12} (\tilde \Gamma^*) \tilde C_{22}(\tilde \Gamma^*)^{-1}$ is $\min\{k^*, p^*\}$ and \eqref{eq:bound_A_Gamma} implies that $\| A(\tilde \Gamma^*)\|$ is bounded above by $c_*$, where $c_* = 1 + b^{*2} \check c$ is the upper bound on $\|C\|$. This in turn implies that $\|A(\tilde \Gamma^*)\|^{2}_{2}$ is stochastically dominated by $c_*^{2} \tau_{0}^{2}\| \breve{\epsilon} \|^{2}$, where $\breve{\epsilon}$ is distributed as $\Ncal(0, I_{\min(k^*, p^*)})$ and $\| \breve{\epsilon} \|^{2}$ is distributed as $\chi^{2}_{\min(k^*, p^*)}$. Using this fact, we provide an upper bound for the expectation in \eqref{eq:vecB_T2_18} as follows
\begin{align}
\label{eq:vecB_T2_19}
\EE_{\tilde \epsilon \mid R, S, Z} \underset{\tilde \Gamma^* \in \tilde \Gcal}{\max} \; \tilde \epsilon_{2}^{\T} A(\tilde \Gamma^*)^{\T} A(\tilde \Gamma^*) \tilde \epsilon_{2} & = \int_{0}^{\infty} \PP \left \{ \underset{\tilde \Gamma^* \in \tilde \Gcal}{\max} \; \| A(\tilde \Gamma^*) \tilde \epsilon_{2} \|^{2}_{2} > t \right\} dt \nonumber \\
&= 2c_*^{2} \tau_{0}^{2} p^*  +
\int_{2 c_*^{2} \tau_{0}^{2} p^*}^{\infty} \PP \left \{ \underset{\tilde \Gamma^* \in \tilde \Gcal}{\max} \;\| A(\tilde \Gamma^*) \tilde \epsilon_{2} \|^{2}_{2} > t \right\} dt, \nonumber \\
\int_{2 c_*^{2} \tau_{0}^{2} p^*}^{\infty} \PP \left \{ \underset{\tilde \Gamma^* \in \tilde \Gcal}{\max} \;\| A(\tilde \Gamma^*) \tilde \epsilon_{2} \|^{2}_{2} > t \right\} dt
 &\overset{(i)}{\leq}
\sum_{\tilde \Gamma^* \in \tilde \Gcal} \int_{2 c_*^{2} \tau_{0}^{2} p^*}^{\infty} \PP \left \{ \| A(\tilde \Gamma^*) \tilde \epsilon_{2} \|^{2}_{2} > t \right\} dt \nonumber \\
&\overset{(ii)}{\leq}
\sum_{\tilde \Gamma^* \in \tilde \Gcal}
\int_{2 c_*^{2} \tau_{0}^{2} p^*}^{\infty} \PP \left \{ c_*^{2} \tau_{0}^{2} \chi^{2}_{\min(k^*, p^*)} > t \right\} dt \nonumber \\
&\overset{(iii)}{\leq}
\sum_{\tilde \Gamma^* \in \tilde \Gcal}
\int_{2 c_*^{2} \tau_{0}^{2} p^*}^{\infty} \PP \left( \chi^{2}_{p^*} > t / c_*^{2} \tau_{0}^{2}  \right) dt \nonumber\\
&\overset{(iv)}{=}
c_*^{2} \tau_{0}^{2}  \lceil 2 b^* k^* \log N \rceil^{k^{*2}}
\int_{p^*}^{\infty} \PP \left( \chi^{2}_{p^*} - p^* > u  \right) du,
\end{align}
where $(i)$ follows from the union bound, $(ii)$ follows because  $c_*^{2} \tau_{0}^{2}\| \breve{\epsilon} \|^{2}$ stochastically dominates $\| A(\tilde \Gamma^*) \tilde \epsilon_{2} \|^{2}_{2}$,  $(iii)$ follows because $ \chi^{2}_{p^*}$ stochastically dominates $\chi^{2}_{\min(k^*, p^*)}$, and $(iv)$ follows from the substitution $u = t/c_*^{2} \tau_{0}^{2} - p^*$. Definition 2.7 in \citet{Wainwright_2019} implies that if $f$ follows $\chi^{2}_{p}$ distribution, then its moment generating function is
\begin{align*}
\EE \{e^{\lambda (f - p)}\} = \frac{e^{- \lambda p}}{(1 - 2 \lambda)^{p/2}} \leq e^{4 p \lambda^{2} / 2}, \quad |\lambda | < 1/4,
\end{align*}
and it is subexponential with parameters $(\nu, \alpha) = (2 \sqrt{p}, 4)$. Proposition 2.9  in \citet{Wainwright_2019} gives
\begin{align}
\label{eq:vecB_T2_20}
\int_{p}^{\infty} \PP \left( \chi^{2}_{p} - p > u  \right) du \leq \int_{p}^{\infty} e^{- \frac{u}{8}} du = 8 e^{-p/8}.
\end{align}
Using (\ref{eq:vecB_T2_20}) in (\ref{eq:vecB_T2_19}) gives
\begin{align}
\label{eq:vecB_T2_21}
\EE_{\tilde \epsilon \mid R, S, Z} \underset{\tilde \Gamma^* \in \tilde \Gcal}{\max} \; \tilde \epsilon_{2}^{\T} A(\tilde \Gamma^*)^{\T} A(\tilde \Gamma^*) \tilde \epsilon_{2}
&\leq 2 c_*^{2} \tau_{0}^{2} p^*  + 8 c_*^{2} \tau_{0}^{2}  \lceil 2 b^* k^* \log N \rceil^{k^{*2}}  e^{-p^*/8} \nonumber \\
&= 2 c_*^{2} \tau_{0}^{2} p^*  + 8 c_*^{2} \tau_{0}^{2} e^{k^{*2} \log \lceil 2 b^* k^* \log N \rceil -p^*/8 } \nonumber\\
&\leq 2 c_*^{2} \tau_{0}^{2} p^*  + 8 c_*^{2} \tau_{0}^{2} e^{k^{*2} \log  \{(2 b^* + 1) k^* \log N\}  -p^*/8 } \nonumber\\
&= 2 c_*^{2} \tau_{0}^{2} p^*  + 8 c_*^{2} \tau_{0}^{2} e^{k^{*2} \log  (2 b^* + 1) +  k^{*2} \log k^* + k^{*2} \log \log N  -p^*/8 } \nonumber\\
&\overset{(i)}{=} 2 c_*^{2} \tau_{0}^{2} p^*  + 8 c_*^{2} \tau_{0}^{2}e^{-p^*/8 \{1 + o(1) \}}, 
\end{align}
where $(i)$ follows if  $k^{*2} \log k^* = o(p^*)$ and $k^{*2} \log \log N = o(p^*)$.
We combine \eqref{eq:vecB_T2_16}, \eqref{eq:vecB_T2_18}, \eqref{eq:vecB_T2_21} and use them in \eqref{eq:vecB_T2_15} to obtain
\begin{align}
\label{eq:vecB_T2_22}
&\EE (\tilde \epsilon_{2}^{\T} \tilde C_{22}^{-1} \tilde C_{12}^{\T} F^{{-1}} \tilde
D^{2} F^{-1} \tilde C_{12} \tilde C_{22}^{-1} \tilde \epsilon_{2}) \nonumber\\
&\leq \frac{\tau_0^2}{a^4}\EE_{X}\left\{ \kappa^4(X) \right\} \, \EE_{R, S, Z} \left[\frac{4b^*(N-p^*)}{\log N} (2\breve c + 3 b^{*2} \breve c^2 + b^{*4} \breve c^3)  + 2 c_*^{2} p^*  + 8c_*^{2} e^{-p^*/8 }\right] \nonumber\\
&= \frac{2\tau_0^2}{a^4}\EE_{X}\left\{ \kappa^4(X) \right\} \, \EE_{R, S, Z} \left[\frac{2b^*(N-p^*)}{\log N} (2\breve c + 3 b^{*2} \breve c^2 + b^{*4} \breve c^3)  +  (1 + b^{*4} \breve c^2 + 2 b^{*2} \breve c) (p^*  + 4 e^{-p^*/8}) \right] \nonumber\\
&= \frac{2\tau_0^2}{a^4}\EE_{X}\left\{ \kappa^4(X) \right\} \, \left\{ (p^*  + 4 e^{-p^*/8}) + a_{1} \EE_{R, S, Z}(\breve c) + a_2 \EE_{R, S, Z} (\breve c^2) + a_3 \EE_{R, S, Z} (\breve c^3) \right\}, \nonumber\\
a_1 &= \frac{4b^*(N-p^*)}{\log N} + 2 b^{*2} (p^*  + 4 e^{-p^*/8}), \quad a_2 = \frac{6b^{*3}(N-p^*)}{\log N} + b^{*4}(p^*  + 4 e^{-p^*/8}), \quad 
a_3 =  \frac{2b^{*5}(N-p^*)}{\log N}. 
\end{align}
Finally, taking expectation on both sides of \eqref{eq:vecB_T2_6}, and substituting \eqref{eq:vecB_T2_13}, \eqref{eq:vecB_T2_22}, 
and using $r=1,2,3$ in Lemma \ref{auxlem3}
gives
\begin{align*}
\EE (T_2)
&\leq \frac{2\tau_{0}^{2}}{a^{4}} \; \EE_{X}\bigl\{\kappa^4(X)\bigr\} \; \Bigl[3p^* + 8e^{-p^*/8} \\
&\quad + 2^8  b^*\sigma_Z^2\Bigl\{\frac{2(N-p^*)}{\log N} + b^* (p^*  + 4 e^{-p^*/8})\Bigr\} \prod_{d=1}^3\frac{1}{k_d^{2}}\left(q_d + k_d + 2\right)^2\,\{n(q^*+2)+N\} \\
&\quad + 2^{19}b^{*3}\sigma_Z^4\, \Bigl\{\frac{6(N-p^*)}{\log N}  + b(p^* + 4e^{-p^*/8}) \Bigr\}\prod_{d=1}^3\frac{1}{k_d^{4}}\left(q_d^2 + k_d^2 + 2\right)^2\,\Bigl\{n(q^{*2}+2)+\sum_{i=1}^n m_i^2\Bigr\} \\
&\quad + 2^{32}b^{*5}\sigma_Z^6\,\Bigl( \frac{N-p^*}{\log N} \Bigr)\prod_{d=1}^3\frac{1}{k_d^{6}}\left(q_d^3 + k_d^3 + 3\right)^2\,\Bigl\{n(q^{*3}+3)+\sum_{i=1}^n m_i^3\Bigr\}\Bigr] \nonumber \\
&\overset{(i)}{\leq} \frac{2\tau_{0}^{2}}{a^{4}} \; \EE_{X}\bigl\{\kappa^4(X)\bigr\} \; \Bigl[3p^* + 8e^{-p^*/8} \\
&\quad + 2^8  b^*\sigma_Z^2\Bigl\{\frac{2(N-p^*)}{\log N} + b^* (p^*  + 4 e^{-p^*/8})\Bigr\} \prod_{d=1}^3 \frac{1}{k_d^2}(4q_d^2+4k_d^2+8)\,\{n(q^*+2)+N\} \\
&\quad + 2^{19}b^{*3}\sigma_Z^4\, \Bigl\{\frac{6(N-p^*)}{\log N} + b^*(p^* + 4 e^{-p^*/8}) \Bigr\}\prod_{d=1}^3 \frac{1}{k_d^4}(4 q_d^4+4 k_d^4+8)\,\Bigl\{n(q^{*2}+2)+n m_{\max}^2\Bigr\} \\
&\quad + 2^{32}b^{*5}\sigma_Z^6\,\Bigl( \frac{N-p^*}{\log N} \Bigr)\prod_{d=1}^3 \frac{1}{k_d^6}(4 q_d^6+4 k_d^6+18)\,\Bigl\{n(q^{*3}+3)+n m_{\max}^3\Bigr\}\Bigr], \\
&= \frac{2\tau_{0}^{2}}{a^{4}} \; \EE_{X}\bigl\{\kappa^4(X)\bigr\} \; \Bigg\{ \left(\frac{N-p^*}{\log N}\right)O\left( nb^{*5}\sigma_Z^6 \left(\prod_{d=1}^3\frac{q_d^9}{k_d^6} + \prod_{d=1}^3 \frac{m_{\max}^3}{k_d^6}\right)\right) \\
&\quad + (p^* + 4e^{-p^*/8})O\left(nb^{*3}\sigma_Z^4\left(\prod_{d=1}^3\frac{q_d^6}{k_d^4}+ \prod_{d=1}^3 \frac{m_{\max}^2}{k_d^4} \right)\right)\Bigg\}
\end{align*}
where $(i)$ follows from the inequality $(u+v)^2 \leq 2(u^2 + v^2)$. The lemma is proved.

\end{proof}

\subsection{Upper Bound for $T_1$}\label{Section:vecB_ProofT1}
\begin{lemma}\label{vecB_lemt1}
Under the assumptions of Theorem 1 in the main manuscript,
\begin{align*}
\EE(T_1) \leq  O\left(nb^{*8}\sigma_Z^8 \left(\prod_{d=1}^3 \frac{q_d^{12}}{k_d^8} + \prod_{d=1}^3\frac{m_{\max}^4}{k_d^8} \right) + nb^{*8}\sigma_Z^8 \right) \EE_{X} \left\{
\kappa^4(X) \right\} \| \bm{\beta}^0\|_2^{2}
\end{align*}

\end{lemma}
\begin{proof}
The term $T_{1}$ satisfies
\begin{align}
\label{eq:vecB_T1_1}
T_1 \leq \| C \|^2 \, \| (X \delta^2 \Lambda X^\T + C)^{-1} X  \bm{\beta}^{0}\|_2^2 \leq c_*^2 \, \| (X \delta^2 \Lambda X^\T + C)^{-1} X  \bm{\beta}^0\|_2^2 ,
\end{align}
where $c_* = 1 + b^{*2} \breve{c}$ is the upper bound of $\|C\|$ in \eqref{normC-bd}. Using \eqref{eq:vecB_T2_1}, \eqref{eq:vecB_T2_2} and \eqref{eq:blockmatinv}, we get
\begin{align}
\label{eq:vecB_T1_2}
 (X \delta^2 \Lambda X^\T + C)^{-1} X  \bm{\beta}^0 &= U
\begin{bmatrix}
  F^{-1} &  - F^{{-1}} \tilde C_{12} \tilde C_{22}^{-1} \\
  - \tilde C_{22}^{{-1}} \tilde C_{12}^{\T} F^{-1}  & \tilde C_{22}^{-1} + \tilde
  C_{22}^{-1} C_{12}^{\T} F^{{-1}} \tilde C_{12} \tilde C_{22}^{-1}
\end{bmatrix}
U^{\T} U
  \begin{bmatrix}
  D \\
  0_{N-p \times p}
  \end{bmatrix}
  V^\T \bm{\beta}^0 \nonumber\\
  &= U
\begin{bmatrix}
  F^{-1} D V^{\T} \bm{\beta}^0\\
  - \tilde C_{22}^{{-1}} \tilde C_{12}^{\T} F^{-1}  D V^{\T} \bm{\beta}^0
\end{bmatrix}.
\end{align}
This further implies that the second term on the right-hand side of \eqref{eq:vecB_T1_1} satisfies
\begin{align}
\label{eq:vecB_T1_3}
 \| (X \delta^2 \Lambda X^\T + C)^{-1} X  \bm{\beta}^0 \|_2^{2} &= \|  F^{-1} D V^{\T}
 \bm{\beta}^0 \|_2^{2} + \| \tilde C_{22}^{{-1}} \tilde C_{12}^{\T} F^{-1}  D V^{\T}
 \bm{\beta}^0\|_2^{2} \nonumber\\
 &\leq \bm{\beta}^{0\T} V D F^{-2} D
V^{\T} \bm{\beta}^0 + \| \tilde C_{22}^{{-1}} \tilde C_{12}^{\T}\|^2 \|F^{-1}  D V^{\T}
 \bm{\beta}^0\|_2^{2} \nonumber\\
&= \bm{\beta}^{0\T} V D F^{-2} D
V^{\T} \bm{\beta}^0 + \|A(\Gamma)^{\T}\|^2\bm{\beta}^{0\T} V D F^{-2} D
V^{\T} \bm{\beta}^0 \nonumber\\
&\leq \{1 + (1 + b^{*2} \breve{c})^2\} \bm{\beta}^{0\T} V D F^{-2} D V^{\T} \bm{\beta}^0,
\end{align}
which follows from \eqref{eq:bound_A_Gamma} in Lemma \ref{auxlem2}.

Following the same arguments to derive \eqref{eq:vecB_T2_13}, we obtain that
\begin{align}
\label{eq:vecB_T1_4}
\bm{\beta}^{0\T} V D F^{-2} D V^{\T} \bm{\beta}^0 &\leq
\frac{1}{\lambda_{\min}^{2}(\tilde D)} \bm{\beta}^{0\T} V D^{2} V^{\T} \bm{\beta}^0 \nonumber \\ &\leq \frac{\lambda_{\max}^{2}(
  D)}{\lambda_{\min}^{2}(\tilde D)} \, \bm{\beta}^{0\T} V V^{\T} \bm{\beta}^0 \nonumber \\
  &=  \frac{\lambda_{\max}^{2}(
  D)}{\lambda_{\min}^{2}(\tilde D)} \,\| \bm{\beta}^0\|_2^{2} \nonumber \\
  & = \kappa^4(X) \| \bm{\beta}^0\|_2^{2}.
\end{align}
Substituting \eqref{eq:vecB_T1_4} and \eqref{eq:vecB_T1_3} in \eqref{eq:vecB_T1_1} gives
\begin{align}\label{vecB_bound_T1}
    T_1 &\leq (1 + b^{*2} \breve{c})^2 \{1 + (1 + b^{*2} \breve{c})^2\}\kappa^4(X) \| \bm{\beta}^0\|_2^{2} \nonumber \\
    &= \{(1 + b^{*2} \breve{c})^2 + (1 + b^{*2} \breve{c})^4 \} \kappa^4(X) \| \bm{\beta}^0\|_2^{2} \nonumber \\
    &\leq \{2(1 + b^{*4} \breve{c}^2) + 2^3(1 + b^{*8} \breve{c}^4)\}\kappa^4(X) \| \bm{\beta}^0\|_2^{2},
\end{align}
which follows from the inequality $(u+v)^r \leq 2^{r-1}(u^r + v^r)$. Taking expectations on both sides of \eqref{vecB_bound_T1} implies that
\begin{align*}
\EE(T_{1}) &\leq \EE_{R,S,Z}\left\{2(1 + b^{*4} \breve{c}^2) + 2^3(1 + b^{*8} \breve{c}^4)\right\} \EE_{X} \left\{
\kappa^4(X) \right\} \| \bm{\beta}^0\|_2^{2} \nonumber \\
&\overset{(i)}{\leq} \Bigl[10 + 2^{19}b^{*4} \sigma_Z^4 \prod_{d=1}^3 \frac{1}{k_d^4}(q_d^2 + k_d^2 + 2)^2\Bigl\{n(q^{*2}+2)+\sum_{i=1}^n m_i^2\Bigr\} \nonumber \\
&\quad + 2^{45}b^8\sigma_Z^8 \prod_{d=1}^3 \frac{1}{k_d^8}(q_d^4 + k_d^4 + 6)^2\Bigl\{n(q^{*4}+6)+\sum_{i=1}^n m_i^4\Bigr\}\Bigr]\EE_{X} \left\{
\kappa^4(X) \right\} \| \bm{\beta}^0\|_2^{2} \nonumber \\
&\overset{(ii)}{\leq} \Bigl[10 + 2^{19}b^{*4} \sigma_Z^4 \prod_{d=1}^3 \frac{1}{k_d^4}(4q_d^4 + 4k_d^4 + 8)\Bigl\{n(q^{*2}+2)+\sum_{i=1}^n m_i^2\Bigr\} \nonumber \\
&\quad + 2^{45}b^{*8}\sigma_Z^8 \prod_{d=1}^3 \frac{1}{k_d^8}(4q_d^8 + 4k_d^8 + 72)\Bigl\{n(q^{*4}+6)+\sum_{i=1}^n m_i^4\Bigr\}\Bigr]\EE_{X} \left\{
\kappa^4(X) \right\} \| \bm{\beta}^0\|_2^{2} \\
&\leq \Bigl[10 + 2^{19}b^{*4} \sigma_Z^4\prod_{d=1}^3 \frac{1}{k_d^4}(4 q_d^4 + 4 k_d^4 + 8)\Bigl\{n(q^{*2}+2)+n m_{\max}^2\Bigr\} \nonumber \\
&\quad + 2^{45}b^{*8}\sigma_Z^8\prod_{d=1}^3 \frac{1}{k_d^8}(4 q_d^8 + 4 k_d^8 + 72)\Bigl\{n(q^{*4}+6)+n m_{\max}^4\Bigr\}\Bigr]\EE_{X} \left\{
\kappa^4(X) \right\} \| \bm{\beta}^0\|_2^{2} \\
&= O\left(nb^{*8}\sigma_Z^8 \left(\prod_{d=1}^3 \frac{q_d^{12}}{k_d^8} + \prod_{d=1}^3\frac{m_{\max}^4}{k_d^8} \right) + nb^{*8}\sigma_Z^8 \right) \EE_{X} \left\{
\kappa^4(X) \right\} \| \bm{\beta}^0\|_2^{2} ,
\end{align*}
where $(i)$ holds by using $r = 2$ and $r = 4$ in Lemma \ref{auxlem3} and $(ii)$ follows from the inequality $(u+v)^2 \leq 2(u^2 + v^2)$. The lemma is proved.
\end{proof}

\subsection{Auxiliary Lemmas}
\begin{lemma}\label{auxlem1}
Consider the set of matrices $\Gcal_{d} = \{\Gamma_d \in \RR^{k_d \times k_d}: \| \Gamma_d \| \leq b_d \}$ for all mode $d = 1, 2, 3$. Define, $\Gcal = \{\Gamma^* \in \RR^{k^* \times k^*}: \| \Gamma^* \| \leq b^* \}$ with $b^* = b_1 b_2 b_3$. Then, there is a $1/(\log N)$-cover for the metric space $(\Gcal, \| \cdot \|)$ with cardinality $\lceil 2 b^* k^* \log N \rceil^{k^{*2}} $, where $k^* = k_1 k_2 k_3$.
\end{lemma}

\begin{proof}

  Consider the approximation of $[-b^*, b^*]$ by a uniform grid with $\lceil 2b^* k^* \log
  N \rceil$ partitions. Every partition in the grid has length less than or equal to $1 / (k^* \log N)$. If the center of the $l$-th partition is $g_l$, then we can form at most $\lceil 2b^*k^* \log N \rceil^{k^{*2}}$ $k^* \times k^*$  matrices using $\{g_l: l = 1, \ldots, \lceil 2b^* k^* \log
  N \rceil\}$. Denote the set formed using these matrices as $\tilde \Gcal = \{\tilde \Gamma^*_l \in \RR^{k^* \times k^*}: l = 1, \ldots, \lceil 2b^* k^* \log
  N \rceil^{k^{*2}} \}$.

  We prove the lemma by showing that $\tilde \Gcal$ is a $1/(\log N)$-net for the metric space $(\Gcal, \| \cdot \|)$. For any $\Gamma^* \in \Gcal$ and $i, j = 1, \ldots, k^*$,  $ \Gamma^*_{ij} \in [-b^*, b^*]$ because $\| \Gamma^* \| \leq b^*$. The construction of $\tilde \Gcal$ implies that there exists a $\tilde \Gamma^* \in \tilde \Gcal$ such that  $\mid \Gamma^*_{ij} -  \tilde \Gamma^*_{{ij}} \mid \leq 1 / (k^* \log N)$ for $i,j = 1, \ldots, k^*$. Furthermore, summing this inequality over rows and columns gives
  \begin{align}
  \label{eq:auxlem1_1}
  \sum_{j=1}^{k^*} \mid \Gamma^*_{ij} - \tilde \Gamma^*_{{ij}} \mid \leq \frac{1}{\log
  N}, \quad
  \sum_{i=1}^{k^*} \mid \Gamma^*_{ij} - \tilde \Gamma^*_{{ij}} \mid \leq \frac{1}{\log
  N}.
  \end{align}
  Using \eqref{eq:auxlem1_1}, in the definitions of $\| \cdot \|_{\infty}$, $\| \cdot \|_1$, and $\| \cdot \|$ norms for matrices implies that
  \begin{align}
  \| \Gamma^* - \tilde  \Gamma^* \|_{\infty} &= \underset{1 \leq i \leq k^*}{\max}
  \sum_{j=1}^{k^*} \mid \Gamma^*_{ij} - \tilde \Gamma^*_{{ij}} \mid \leq \frac{1}{\log
  N}, \quad
  \| \Gamma^* - \tilde \Gamma^* \|_{1} = \underset{1 \leq j \leq k^*}{\max}
  \sum_{i=1}^{k} \mid \Gamma^*_{ij} - \tilde \Gamma^*_{{ij}} \mid \leq \frac{1}{\log
  N}, \nonumber\\
  \|\Gamma^* - \tilde \Gamma^* \| &\leq \left( \| \Gamma^* - \tilde \Gamma^* \|_{1} \| \Gamma^* - \tilde
  \Gamma^* \|_{\infty}\right)^{1/2} \leq \frac{1}{\log N},
  \end{align}
  where the last inequality proves that $\tilde \Gcal$ is a $1 / (\log N)$-cover for the metric space $(\Gcal, \| \cdot \|)$. The lemma is proved.
\end{proof}

\begin{lemma} \label{auxlem2}
Let $\Gcal = \{ \Gamma^* \in \RR^{k^* \times k^*}: \| \Gamma^* \| \leq b^*\}$ be the set used in in Lemma \ref{auxlem1} with $\tilde \Gcal$ as its $1/ \log N$-net and $A(\Gamma^*) = \tilde C_{12}(\Gamma^*) \tilde C_{22} (\Gamma^*)^{-1}$ be the matrix that appears in \eqref{eq:vecB_T2_16}, where the notation emphasizes the dependence on $\Gamma^*$. Then, for any $\Gamma^* \in \Gcal$ and $\tilde \Gamma^* \in \tilde \Gcal$,
 \begin{align}\label{eq:auxlem2_1}
     \| A(\Gamma^*)^{\T} A(\Gamma^*) - A(\tilde \Gamma^*)^{\T} A(\tilde \Gamma^*) \| 
  &\leq 4b^* (2\breve c +3  b^{*2} \breve c^2 + b^{*4} \breve c^3)\, \| \Gamma^* - \tilde \Gamma^*\|, \nonumber \\
  \breve c &= \prod_{d=1}^3 \left(\|S_d\|^2 \|R_d\|^2\right) \max(\|Z_1\|^2, \ldots, \| Z_n\|^2).
 \end{align}   

\end{lemma}

\begin{proof}
Recall that 
$C_i(\Gamma^*) = Z_i S^{*\T}
\Gamma^* R^* R^{*\T} \Gamma^{*\T} S^* Z_i^{\T} + I_{m_i}$. For any $\Gamma^* \in \Gcal$, 
\begin{align}\label{eq:auxlem2_2}
1 \overset{(i)}{\leq} \lambda_{\min}\{ C_i(\Gamma^*) \} &\leq \lambda_{\max}\{ C_i(\Gamma^*) \} \leq 1 + \|Z_i S^{*\T} \Gamma^* R^* \|^2 \leq  1 + b^{*2} \| S^* \|^2 \| R^*\|^2 \|Z_i\|^2, \nonumber\\
1 \overset{(ii)}{\leq} \lambda_{\min}\{C(\Gamma^*)\} &\leq \lambda_{\max}\{C(\Gamma^*)\} \overset{(iii)}{\leq} 1 + b^{*2} \| S^* \|^2 \| R^*\|^2 \max(\|Z_1\|^2, \ldots, \| Z_n\|^2), 
\end{align}
where $(i)$ follows because $C_i(\Gamma^*) = Z_i S^{*\T}
\Gamma^* R^*  \left( Z_i S^{*\T} \Gamma^* R^* \right)^{\T} + I_{m_i}$,  $(ii)$ follows because  $\lambda_{\min}\{C(\Gamma^*)\} = \min_i \lambda_{\min}\{ C_i(\Gamma^*) \}$, and $(iii)$ follows because $\lambda_{\max}\{C(\Gamma^*)\} = \max_i \lambda_{\max}\{ C_i(\Gamma^*) \}$. Finally, using the fact that $\|A \otimes B\| = \|A\| \|B\|$ for any two matrices $A$ and $B$, for any $\Gamma^* \in \mathcal{G}$, we can provide an upper bound for the operator norm of the working covariance matrix as
\begin{align}\label{normC-bd}
\|C(\Gamma^*)\| &= \lambda_{\max}\{C(\Gamma^*)\} \nonumber \\
&\leq 1 + b^{*2} \prod_{d=1}^3 \left(\|S_d\|^2 \|R_d\|^2\right) \max(\|Z_1\|^2, \ldots, \| Z_n\|^2) \nonumber \\
&\equiv 1 + b^{*2} \breve{c}.
\end{align}

For any $\Gamma^* \in \Gcal$ and $\tilde \Gamma^* \in \tilde \Gcal$,
\begin{align}\label{eq:auxlem2_3}
   \| C_i(\Gamma^*) - C_i(\tilde \Gamma^*) \| &= \| Z_i S^{*\T}
\Gamma^* R^* R^{*\T} \Gamma^{*\T} S^* Z_i^{\T} - Z_i S^{*\T}
\tilde \Gamma^* R^* R^{*\T} \tilde \Gamma^{*\T} S^* Z_i^{\T} \| \nonumber\\
&= \| Z_i S^{*\T}
(\Gamma^* - \tilde \Gamma^*) R^* R^{*\T} \Gamma^{*\T} S^* Z_i^{\T} + 
Z_i S^{*\T} \tilde \Gamma^* R^* R^{*\T}  (\Gamma^{*\T} - \tilde \Gamma^{*\T}) S^* Z_i^{\T} 
  \| \nonumber\\  
&\leq \| Z_i S^{*\T} \| \|R^* R^{*\T} \Gamma^{*\T} S^* Z_i^{\T}\| \|\Gamma^* - \tilde \Gamma^* \| + \|Z_i S^{*\T} \tilde \Gamma^* R^* R^{*\T}\| \| S^* Z_i^{\T} \| \| \Gamma^* - \tilde \Gamma^*\| \nonumber\\
&\leq 2 b^* \, \| Z_i S^{*\T} \| \, \| R^* R^{*\T} \| \, \| S^* Z_i^{\T} \| \, \| \Gamma^* - \tilde \Gamma^*\| \nonumber \\
&= 2 b^* \, \| Z_i \|^2 \| S^* \|^2 \, \| R^* \|^2  \, \| \Gamma^* - \tilde \Gamma^*\| \nonumber \\
&= 2 b^* \, \| Z_i \|^2 \prod_{d=1}^3 \left(\|S_d\|^2 \|R_d\|^2\right)\| \Gamma^* - \tilde \Gamma^*\|,
\end{align}
which implies
\begin{align}\label{eq:auxlem2_4}
    \|C(\Gamma^*) - C(\tilde \Gamma^*)\| &= \max_i \|C_i(\Gamma^*) - C_i(\tilde \Gamma^*)\| \nonumber \\
    &= 2 b^* \prod_{d=1}^3 \left(\|S_d\|^2 \|R_d\|^2\right) \max(\|Z_1\|^2, \dots, \|Z_n\|^2) \| \Gamma^* - \tilde \Gamma^*\| \nonumber \\
    &\equiv 2 b^* \breve{c} \| \Gamma^* - \tilde \Gamma^*\|.
\end{align}
Consider the term
\begin{align}\label{eq:auxlem2_5}
  \| A(\Gamma^*)^{\T} A(\Gamma^*) - A(\tilde \Gamma^*)^{\T} A(\tilde \Gamma^*) \|  &=  \| A(\Gamma^*)^{\T} A(\Gamma^*) - A(\tilde \Gamma^*)^{\T} A(\Gamma^*) + A(\tilde \Gamma^*)^{\T} A(\Gamma^*) - A(\tilde \Gamma^*)^{\T} A(\tilde \Gamma^*) \| \| \nonumber\\
&=  \| \{ A(\Gamma^*)^{\T}  - A(\tilde \Gamma^*)^{\T} \} A(\Gamma^*) + A(\tilde \Gamma^*)^{\T} \{ A(\Gamma^*) -  A(\tilde \Gamma^*)\}\|  \nonumber\\
&\leq \{ \| A(\Gamma^*) \| + \| A(\tilde \Gamma^*) \| \} \, \|A(\Gamma^*)  - A(\tilde \Gamma^*)\|.
\end{align}
For any $\Gamma^* \in \Gcal$, 
\begin{align}\label{eq:bound_A_Gamma}
    \| A(\Gamma^*)\| = \|\tilde C_{12}(\Gamma^*) \tilde C_{22}(\Gamma^*)^{-1}\| &\leq \|\tilde C_{12}(\Gamma^*)\| \|\tilde C_{22}(\Gamma^*)^{-1}\| = \|\tilde C_{12}(\Gamma^*)\| \frac{1}{\lambda_{\min} \{ \tilde C_{22}(\Gamma^*) \} } \nonumber \\
    &\overset{(i)}{\leq}  \frac{\|\tilde C(\Gamma^*)\|}{\lambda_{\min}\{\tilde C_{22}(\Gamma^*)\}} \overset{(ii)}{\leq}  \frac{\| C(\Gamma^*)\|}{\lambda_{\min}\{ C(\Gamma^*)\}} \nonumber\\
    &\overset{(iii)}{\leq} \| C(\Gamma^*)\| \overset{(iv)}{\leq}  1 + b^{*2} \breve c,
\end{align}
where $(i)$ follows from Cauchy's interlacing theorem, $(ii)$ follows from Cauchy's interlacing theorem and the fact that  $\|\tilde C\| = \|U^{\T}CU\| = \|C\|$ because $U$ is an orthonormal matrix, and $(iii)$ and $(iv)$ follow from \eqref{eq:auxlem2_2} and \eqref{normC-bd} respectively.
Consider the last term in \eqref{eq:auxlem2_5},
\begin{align}\label{eq:auxlem2_6}
    \|A(\Gamma^*)  - A(\tilde \Gamma^*)\| &= \| A(\Gamma^*) - \tilde C_{12}(\tilde \Gamma^*) \tilde C_{22} (\Gamma^*)^{-1} + \tilde C_{12}(\tilde \Gamma^*) \tilde C_{22} (\Gamma^*)^{-1} - A(\tilde \Gamma^*)\| \nonumber \\
    &= \| \{ \tilde C_{12}(\Gamma^*) - \tilde C_{12}(\tilde \Gamma^*) \}\tilde C_{22} (\Gamma^*)^{-1} + \tilde C_{12}(\tilde \Gamma^*) \{\tilde C_{22} (\Gamma^*)^{-1} -\tilde C_{22} (\tilde \Gamma^*)^{-1} \}\| \nonumber \\
    &\leq \| \tilde C_{22} (\Gamma^*)^{-1}\| \|\tilde C_{12}(\Gamma^*) - \tilde C_{12}(\tilde \Gamma^*)\| + \|\tilde C_{12}(\tilde \Gamma^*) \| \|\tilde C_{22} (\Gamma^*)^{-1} -\tilde C_{22} (\tilde \Gamma^*)^{-1} \| \nonumber \\
    &= \| \tilde C_{22} (\Gamma^*)^{-1}\| \| \tilde C_{12}(\Gamma^*) - \tilde C_{12}(\tilde \Gamma^*)\| + \|\tilde C_{12}(\tilde \Gamma^*)\| \|\tilde C_{22} (\Gamma^*)^{-1} \{\tilde C_{22} (\tilde \Gamma^*) - \tilde C_{22} (\Gamma^*)\} \tilde C_{22} (\tilde \Gamma^*)^{-1}\| \nonumber \\
    &\overset{(i)}{\leq} \| \tilde C_{22} (\Gamma^*)^{-1}\| \|C(\Gamma^*) - C(\tilde \Gamma^*)\| + \|\tilde C_{12}(\tilde \Gamma^*)\| \|\tilde C_{22} (\Gamma^*)^{-1}\| \|\tilde C_{22} (\tilde \Gamma^*)^{-1}\|  \| C (\tilde \Gamma^*) -  C (\Gamma^*)\| \nonumber \\
     &=\lambda_{\min}^{-1} \{\tilde C_{22} (\Gamma^*)\} \, ( 1 + \|\tilde C_{12}(\tilde \Gamma^*)\| \|\tilde C_{22} (\tilde \Gamma^*)^{-1}\| )  \, \|C(\Gamma^*) - C(\tilde \Gamma^*)\|  \nonumber \\
   &\overset{(ii)}{\leq}  \lambda_{\min}^{-1} \{C(\Gamma^*)\} \, ( 1 + \| C(\Gamma^*)\| )  \, \|C(\Gamma^*) - C(\tilde \Gamma^*)\| \nonumber\\
   &\overset{(iii)}{\leq} (2 + b^{*2} \breve{c}) \|C(\Gamma^*) - C(\tilde \Gamma^*)\| \nonumber\\
   &\overset{(iv)}\leq  (2 + b^{*2} \breve{c}) \, 2 b^* \breve{c} \, \|\Gamma^* - \tilde \Gamma^*\| \nonumber\\
   &= (4b^* \breve{c} + 2b^{*3} \breve{c}^2) \, \| \Gamma^* - \tilde \Gamma^*\|,
\end{align}
where $(i)$ holds because $\tilde C_{12}(\Gamma^*) - \tilde C_{12}(\tilde \Gamma^*)$ and $\tilde C_{22}(\Gamma^*) - \tilde C_{22}(\tilde \Gamma^*)$ are $p \times (N-p)$ and $(N-p) \times (N-p)$ submatrices of $\tilde C(\Gamma^*) - \tilde C(\tilde \Gamma^*)$, $(ii)$ and $(iii)$ follow from the arguments used in \eqref{eq:bound_A_Gamma}, and $(iv)$ follows from \eqref{eq:auxlem2_4}. Combining \eqref{eq:auxlem2_5}, \eqref{eq:bound_A_Gamma} and
\eqref{eq:auxlem2_6} implies that 
\begin{align*}
  \| A(\Gamma^*)^{\T} A(\Gamma^*) - A(\tilde \Gamma^*)^{\T} A(\tilde \Gamma^*) \| &\leq 2 (1 + b^{*2} \breve{c}) (4 b^* \breve{c} + 2 b^{*3} \breve{c}^2)\, \| \Gamma^* - \tilde \Gamma^*\| \nonumber\\
  &= 
  4b^* (2\breve{c} +3  b^{*2} \breve{c}^2 + b^{*4} \breve{c}^3)\, \| \Gamma^* - \tilde \Gamma^*\| .
\end{align*}
The lemma is proved.
\end{proof}

\begin{lemma} \label{auxlem3}   
Let $R^*$, $S^*$, $Z_1, \ldots, Z_n$ be the matrices defined in \eqref{eq:notations} used to define the compressed model \eqref{eq:popln&comp_models}. Then, for any $r \geq 1$,
\begin{align*}
    \EE_{R,S,Z}(\breve{c}^{r}) \leq 2^{12r - 6}\sigma_Z^{2r}\prod_{d=1}^3 \left[ \frac{1}{k_d^{2r}}\left\{q_d^{r} + k_d^{r} + 2^{2-r}r\Gamma(r) \right\}^2\right] \, \left\{ nq^{*r} + n2^{2-r}r\Gamma(r) + \sum_{i=1}^n m_i^{r} \right\},
\end{align*}
where we define $\breve{c} = \prod_{d=1}^3 \left(\|S_d\|^2 \|R_d\|^2\right) \max(\|Z_1\|^2, \ldots, \| Z_n\|^2)$.
\end{lemma}

\begin{proof}

Consider any $r \geq 1$. Then, using the independence of the distributions of $R_d$'s and $S_d$'s,
\begin{align*}
    \EE_{R,S,Z}(\breve{c}^r) &= \EE_{R,S,Z}\left\{\prod_{d=1}^3 \left(\|S_d\|^{2r} \|R_d\|^{2r}\right) \max(\|Z_1\|^{2r}, \ldots, \| Z_n\|^{2r})\right\} \nonumber \\
    &= \prod_{d=1}^3 \EE_{S_d}\|S_d\|^{2r} \EE_{R_d}\|R_d\|^{2r} \EE_Z \max(\|Z_1\|^{2r}, \ldots, \| Z_n\|^{2r}) \nonumber \\
    &\overset{(i)}{\leq} \prod_{d=1}^3 \left[\frac{1}{k_d^{2r/2}} 2^{4r - 2} \{q_d^{2r/2} + k_d^{2r/2} + 2^{1-2r/2}(2r)\Gamma(2r/2) \}\right]^2 \nonumber \\
    &\quad \sigma_Z^{2r}2^{4r - 2} \left[ n(q^*)^{2r/2} + n2^{1-2r/2}(2r)\Gamma(2r/2) + \sum_{i=1}^n m_i^{2r/2} \right] \nonumber \\
    &= 2^{12r - 6}\sigma_Z^{2r}\prod_{d=1}^3 \left[ \frac{1}{k_d^{2r}}\left\{q_d^{r} + k_d^{r} + 2^{2-r}r\Gamma(r) \right\}^2\right] \, \left\{ nq^{*r} + n2^{2-r}r\Gamma(r) + \sum_{i=1}^n m_i^{r} \right\}
\end{align*}
where $(i)$ follows from Lemma A.6 in \cite{2025_SarKhSri}. The lemma is proved.
\end{proof}

\section{Additional Numerical Results}\label{Section:supp_numresults}

In this section, we report the empirical findings corresponding to the larger covariance compression dimensions ($k = 6, 9$) considered in the simulation study. The experimental setup remains same as that described in Section~\ref{Section:Sim_Setup} of the main manuscript. The additional simulation results allows us to assess the robustness of CoMET to the choice of covariance compression dimensions. 

Similar to the findings for $k = 3$ in the main manuscript, across different choices of fixed-effect CP-rank ($K$) and cluster size ($m$), the CoMET model exhibits superior estimation and prediction accuracy at both $k = 6$ and $k = 9$, compared to its regularized competitors  (Figures~\ref{fig:rmsermspe_equicorr_k6} and~\ref{fig:rmsermspe_equicorr_k9}). Similarly, CoMET's performance in fixed-effects inference for $k = 6$ and $k = 9$ remain consistent with that for $k = 3$. For moderate choices of the fixed-effect coefficient's rank ($K = 4, 6, 8$), CoMET consistently yields narrower credible intervals than the PQL methods, along with attaining near-nominal coverage (Figures~\ref{fig:covgwidthCI_equicorr_k6} and~\ref{fig:covgwidthCI_equicorr_k9}). Additionally, the CoMET model more accurately identifies the true sparsity patterns in the fixed-effect coefficient tensor $\mathcal{B}$ compared to the penalized methods, while closely matching the performance of the oracle benchmark (Figure~\ref{fig:B_estimate_sim_k3K4}). These findings further strengthen the effectiveness of our proposed model in fixed-effects estimation, prediction, and uncertainty quantification.

\begin{figure}
    \centering
    \includegraphics[width=6.5in,height=6in]{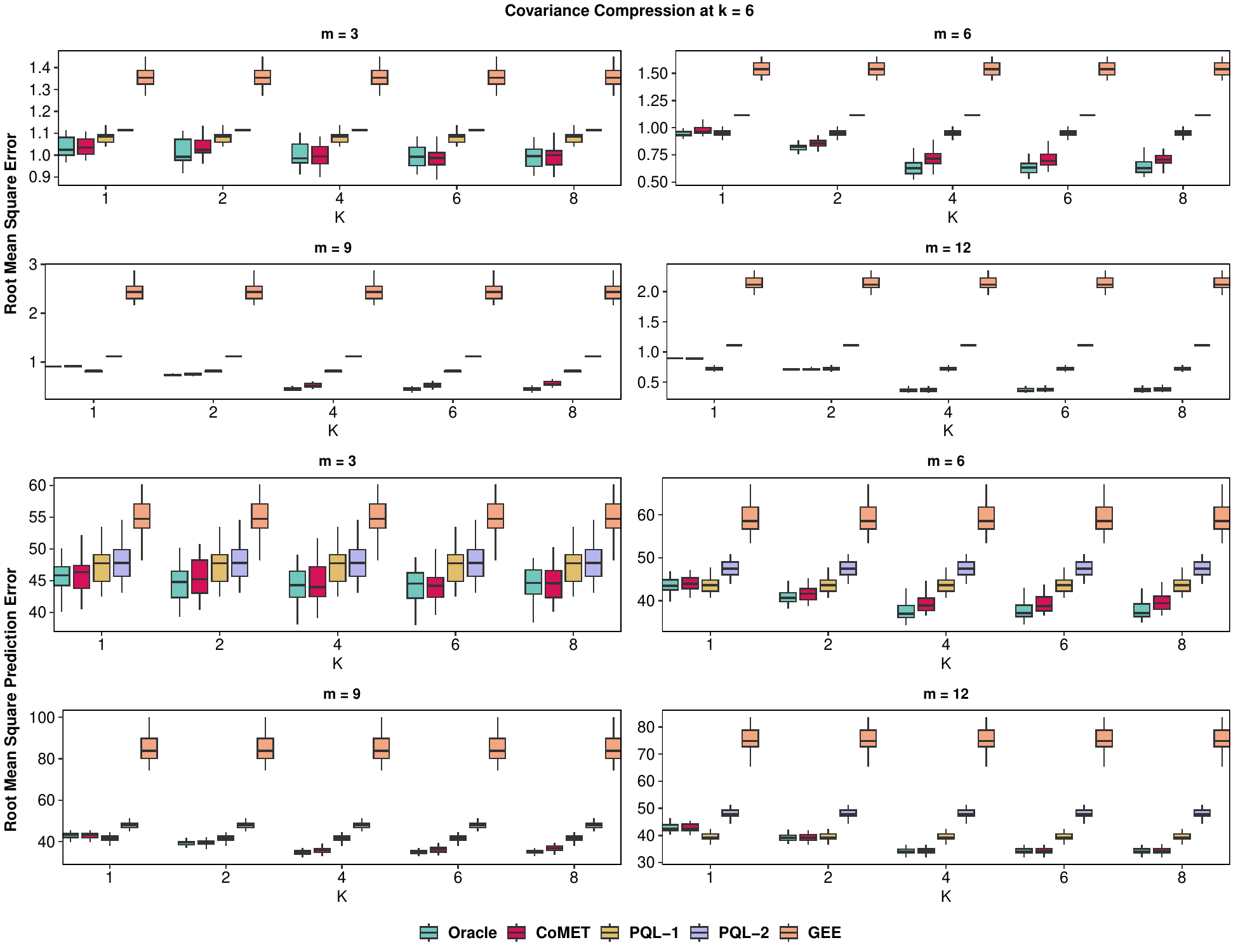}
    \caption{CoMET substantially outperforms existing penalized quasi-likelihood methods in estimation of $\mathcal{B}$ (RMSE) and out-of-sample prediction (RMSPE) across various choices of fixed-effect rank $K$ and cluster sizes $m \in \{3, 6, 9, 12\}$. Results are presented for compressed covariance dimension $k = 6$, summarized over 25 replications. CoMET, compressed mixed-effects tensor model; Oracle, the oracle benchmark; PQL-1, the penalized quasi-likelihood approach of \citet{FanLi12}; PQL-2, the penalized quasi-likelihood approach of \citet{Lietal21}; GEE, the penalized generalized estimating equations method of \citet{2019_Zhang_etal}.}
    \label{fig:rmsermspe_equicorr_k6}
\end{figure}

\begin{figure}
    \centering
    \includegraphics[width=6.5in,height=6in]{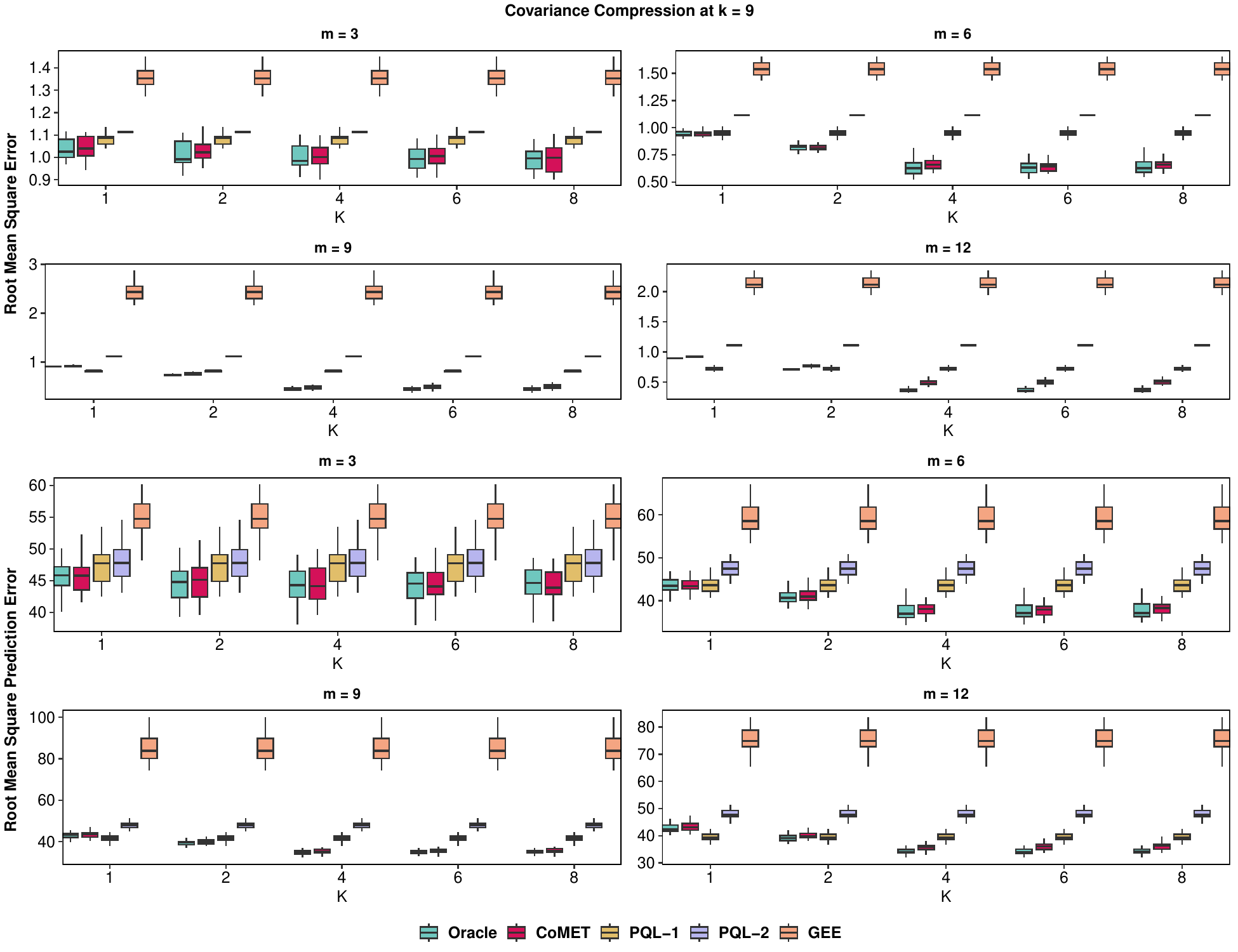}
    \caption{CoMET substantially outperforms existing penalized quasi-likelihood methods in estimation of $\mathcal{B}$ (RMSE) and out-of-sample prediction (RMSPE) across various choices of fixed-effect rank $K$ and cluster sizes $m \in \{3, 6, 9, 12\}$. Results are presented for compressed covariance dimension $k = 9$, summarized over 25 replications. CoMET, compressed mixed-effects tensor model; Oracle, the oracle benchmark; PQL-1, the penalized quasi-likelihood approach of \citet{FanLi12}; PQL-2, the penalized quasi-likelihood approach of \citet{Lietal21}; GEE, the penalized generalized estimating equations method of \citet{2019_Zhang_etal}.}
    \label{fig:rmsermspe_equicorr_k9}
\end{figure}

\begin{figure}[htbp]
    \centering
    \includegraphics[width=6.5in,height=6in]{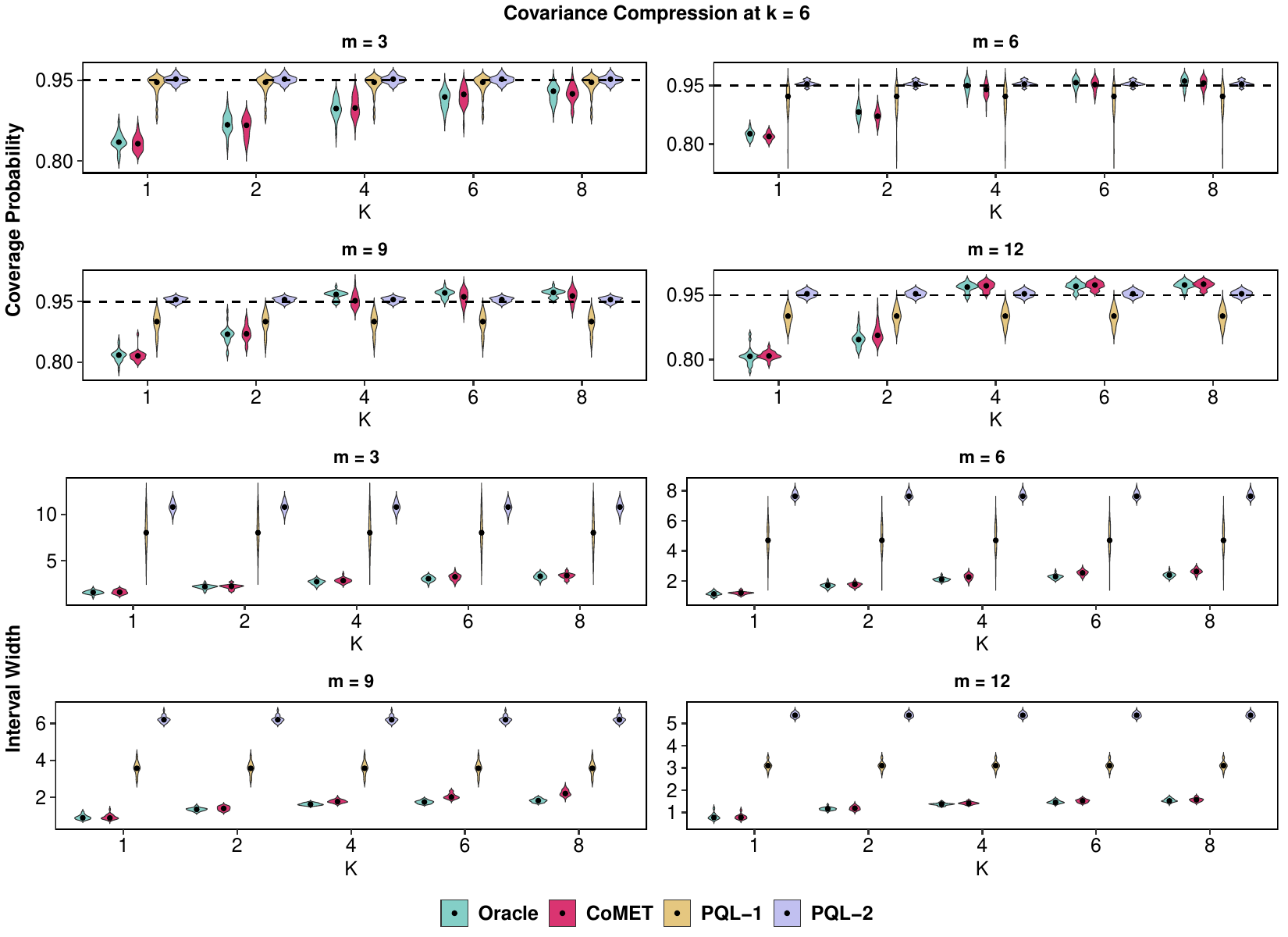}
    \caption{Fixed-effects inference comparisons. The CoMET model produces narrower credible intervals for $\mathcal{B}$ entries than the PQL methods when $k = 6$ and $K \in \{4, 6, 8\}$, along with achieving near-nominal coverage across all cluster sizes $m \in \{3, 6, 9, 12\}$. All metrics are summarized over 25 replications. CoMET, the compressed mixed-effects tensor model; oracle, the oracle variant of CoMET; PQL-1, the penalized quasi-likelihood approach of \citet{FanLi12}; PQL-2, the penalized quasi-likelihood approach of \citet{Lietal21}. GEE, the penalized generalized estimating equations approach of \citet{2019_Zhang_etal}, is excluded because of absence of debiasing technique for confidence intervals construction.}
    \label{fig:covgwidthCI_equicorr_k6}
\end{figure}

\begin{figure}[htbp]
    \centering
    \includegraphics[width=6.5in,height=6in]{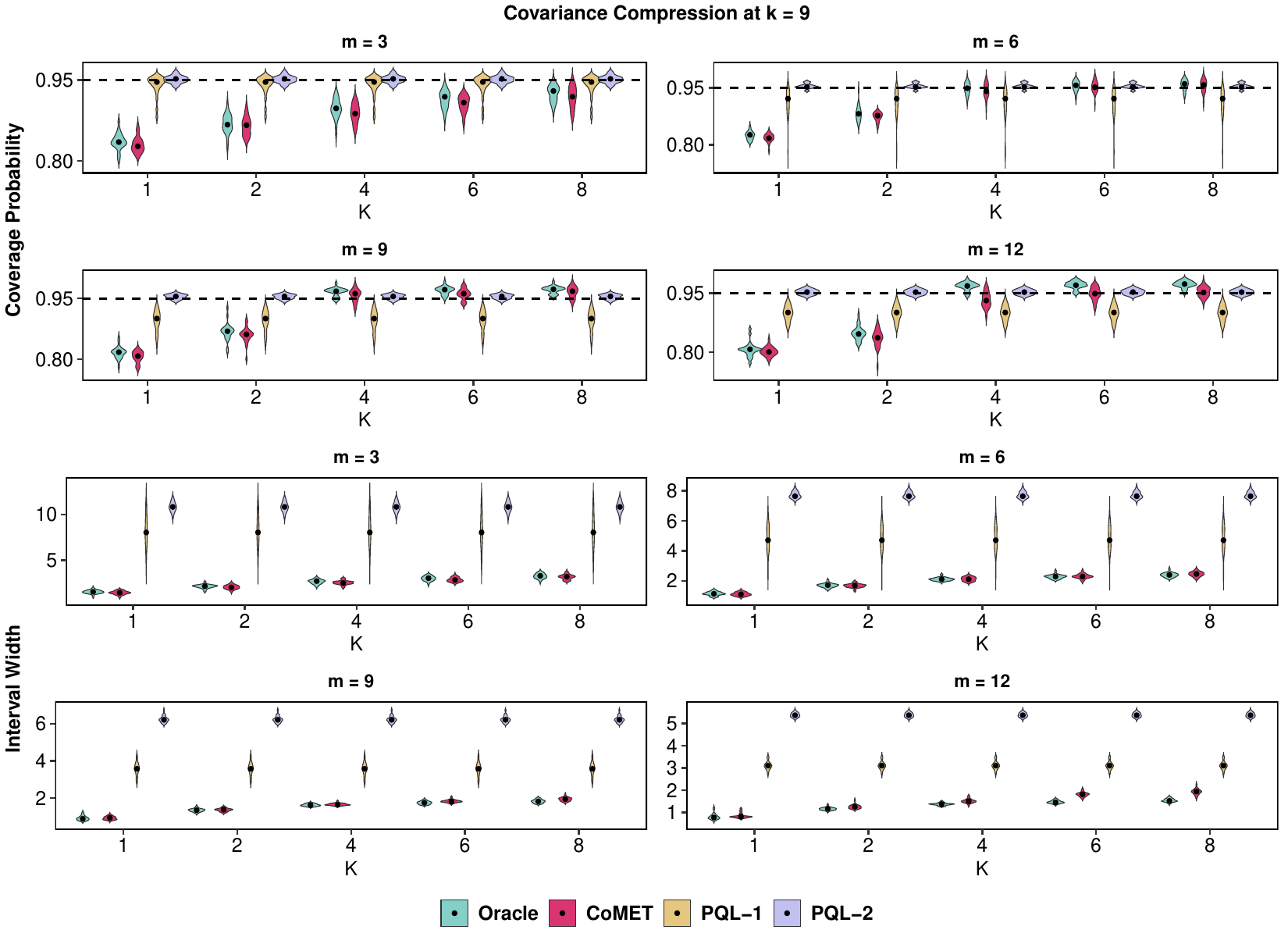}
    \caption{Fixed-effects inference comparisons. The CoMET model produces narrower credible intervals for $\mathcal{B}$ entries than the PQL methods when $k = 9$ and $K \in \{4, 6, 8\}$, along with achieving near-nominal coverage across all cluster sizes $m \in \{3, 6, 9, 12\}$. All metrics are summarized over 25 replications. CoMET, the compressed mixed-effects tensor model; oracle, the oracle variant of CoMET; PQL-1, the penalized quasi-likelihood approach of \citet{FanLi12}; PQL-2, the penalized quasi-likelihood approach of \citet{Lietal21}. GEE, the penalized generalized estimating equations approach of \citet{2019_Zhang_etal}, is excluded because of absence of debiasing technique for confidence intervals construction.}
    \label{fig:covgwidthCI_equicorr_k9}
\end{figure}

\begin{figure}[h!]
    \centering
    \includegraphics[width=0.95\linewidth,height=6.5in]{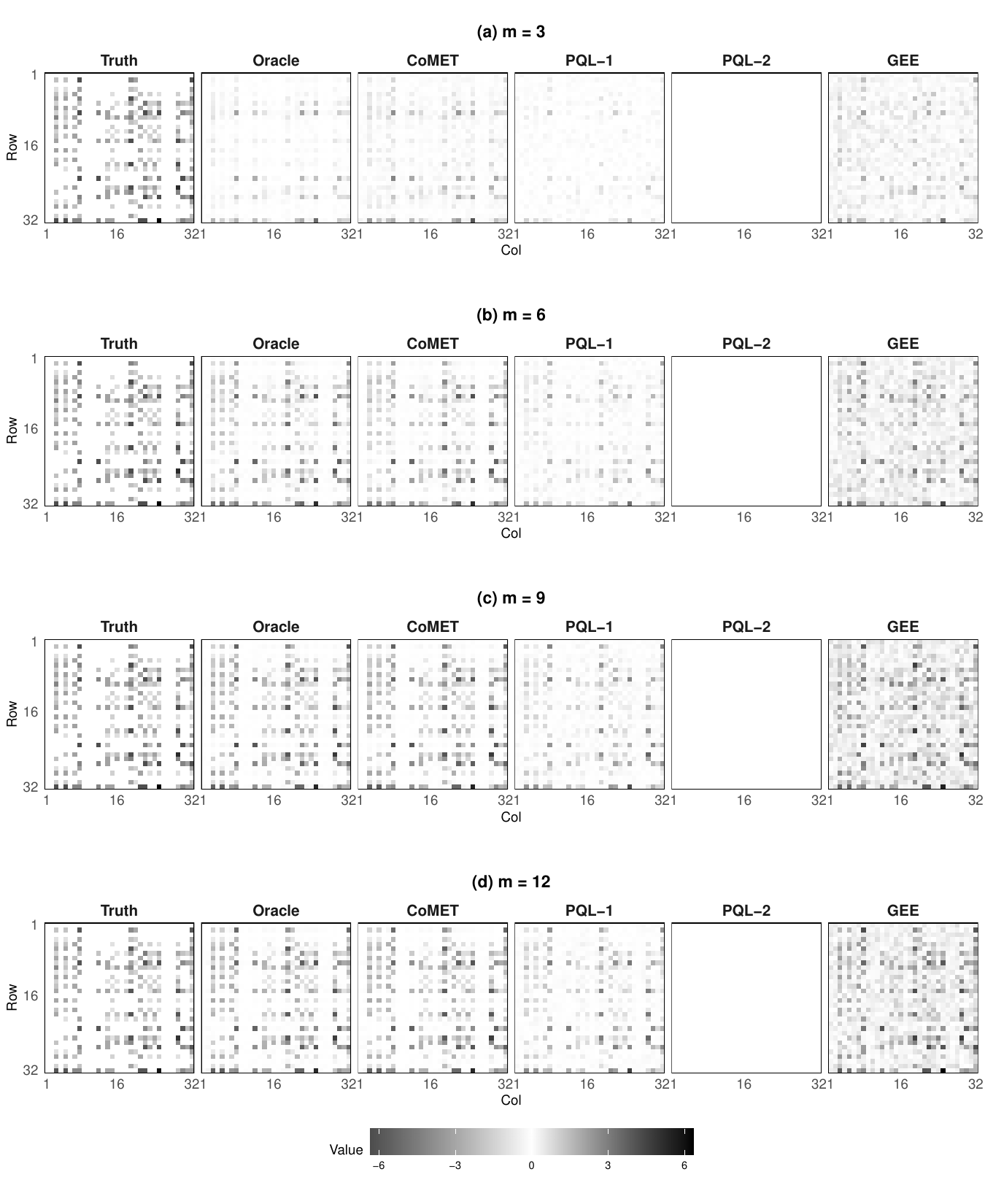}
    \caption{Estimated $32 \times 32$ fixed-effect coefficient matrix $\mathcal{B}$, averaged over 25  simulated training datasets. CoMET accurately recovers the true sparsity patterns in $\mathcal{B}$, closely mimicing the oracle. Results for oracle and CoMET are presented at $K = 4$ and $k = 3$. Truth, the data generating $\mathcal{B}$; Oracle, the oracle benchmark with known random-effects covariance structure; CoMET, the compressed mixed-effects model for tensors; PQL-1, the penalized quasi-likelihood approach of \citet{FanLi12}; and PQL-2, the penalized quasi-likelihood approach of \citet{Lietal21}; GEE, the penalized generalized estimating equation approach of \citet{2019_Zhang_etal} with equicorrelation working correlation specification.}
    \label{fig:B_estimate_sim_k3K4}
\end{figure}

\section{Posterior Computation}\label{Section:FullCondDerivations}

This section presents the detailed derivation of the posterior distributions outlined in Algorithm~\ref{algoCoMET:1} of the main manuscript. We first consider the posterior sampling of the compressed covariance parameters. For every mode $d = 1, 2, 3$, let $\check{Z}_{\gamma_d}$ be the design matrix consisting of the compressed covariance parameters for the remaining modes, and $\check{y}$ be the residual vector obtained by stacking the difference between the observed response and the fitted values computed using $\mathcal{B}$; see the model in \eqref{eq:regmodel_gammad}. Given the observed data, the imputed compressed random-effects $\tilde d_i = \operatorname{vec}(\tilde{\mathcal{D}_i})$s, $(\mathcal{B}, \tau^2)$, and $\gamma_{d'} (d' \neq d)$s, the Gaussian prior on $\gamma_d$ along with its Gaussian likelihood yields the full conditional density as follows:
\begin{align}
    &f(\gamma_d \mid y, \tilde{d_1}, \dots, \tilde{d_n}, \gamma_1, \dots, \gamma_{d-1}, \gamma_{d+1}, \dots, \gamma_3, \mathcal{B}, \tau^2) \nonumber \\
    &\propto f(\check y \mid \gamma_1, \gamma_2, \gamma_3, \mathcal{B}, \tau^2) f(\gamma_d) \nonumber \\
    &\propto (\tau^2)^{-\frac{N}{2}} \exp \left\{ -\frac{1}{2\tau^2} (\check y - \check Z_{\gamma_d} \gamma_d)^{\T} (\check y - \check Z_{\gamma_d} \gamma_d) \right\} (\sigma_{d}^2)^{-\frac{k_d^2}{2}} \exp \left( -\frac{1}{2\sigma_{d}^2} \gamma_d^{\T} \gamma_d\right) \nonumber \\
    &\propto (\sigma_{d}^2)^{-\frac{k_d^2}{2}} \exp \left[ -\frac{1}{2} \left\{ \gamma_d^{\T} \left( \frac{1}{\tau^2}\check Z_{\gamma_d}^{\T}\check Z_{\gamma_d} + \frac{1}{\sigma_{d}^2} I_{k_d^2}\right)\gamma - \frac{2}{\tau^2}\check y^{\T} \check Z_{\gamma_d} \gamma_d \right\}\right] \nonumber \\
    &\propto \exp \left\{ -\frac{1}{2}(\gamma_d - \mu_{\gamma_d})^{\T} \Sigma_{\gamma_d}^{-1}(\gamma_d - \mu_{\gamma_d})\right\}, \nonumber \\
    & \equiv \mathcal{N}(\mu_{\gamma_d}, \Sigma_{\gamma_d}),
\end{align}
where $\mu_{\gamma_d} = \tau^{-2} \Sigma_{\gamma_d}\check Z_{\gamma_d}^{\T}\check y$, and $\Sigma_{\gamma_d} = \left( \tau^{-2} \check Z_{\gamma_d}^{\T}\check Z_{\gamma_d} + \sigma_{d}^{-2} I_{k_d^2} \right)^{-1}$ respectively denote the mean and covariance parameters of the Gaussian density.

We now derive the full conditional distributions of the factor matrices ($B_1, B_2, B_3$) of the three-dimensional fixed-effect coefficient tensor $\mathcal{B}$, and the idiosyncratic error variance $\tau^2$, based on the likelihood marginalized over the imputed compressed random-effects $\tilde{\mathcal{D}}_i$s; see model \eqref{eq:modelBd_cycle2} in the main manuscript. Using the scale mixture representation of the half-Cauchy distribution assigned a priori to the global and local shrinkage parameters in the margin-structured Horseshoe prior in \eqref{eq:hsprior_2} in Section 2.2 of the main manuscript, we obtain the parameter expanded prior specification as
\begin{align}
    \beta^{(g)}_{dj} \mid \lambda_{gdj}^2, \delta_{g}^2, \tau^2 &\overset{\text{ind}}{\sim} \mathcal{N}(0, \lambda_{gdj}^2 \delta_{g}^2 \tau^2 ),\nonumber \\
    \lambda_{gdj}^2 \mid \nu_{gdj} &\overset{\text{ind}}{\sim} \mathcal{IG}(1/2, 1/\nu_{gdj}), \nonumber \\ 
    \delta_{g}^2 \mid \xi_{g} &\overset{\text{ind}}{\sim} \mathcal{IG}(1/2, 1/\xi_{g}), \nonumber \\
    \tau^2 &\sim \mathcal{IG}(a_0, b_0), \nonumber \\
    \nu_{gdj}, \xi_g &\overset{\text{iid}}{\sim} \mathcal{IG}(1/2, 1), g = 1, \dots, K, \ d = 1, 2, 3, \ j = 1, \dots, p_d,
\end{align}
where $\mathcal{IG}(a, b)$ denotes the inverse-gamma distribution with shape parameter $a$ and scale parameter $b$.

The conditional posterior density of the vectorized mode-$d$ factor matrix $\tilde{\bm \beta}_d^{\T} = (\bm \beta_{1d}^{\T}, \dots, \bm \beta_{Kd}^{\T})$ is then given by
\begin{align}\label{eq:cond_post_vecBd}
f(\tilde{\bm \beta}_d &\mid \tilde{\bm \beta}_1, \dots, \tilde{\bm\beta}_{d-1}, \tilde{\bm\beta}_{d+1}, \dots, \tilde{\bm \beta}_3, \tau^2, \gamma_1, \gamma_2, \gamma_3,\lambda^2_{111}, \dots, \lambda^2_{K3p_3},\nonumber \\ &\delta^2_{1}, \dots, \delta^2_{K}, \nu^2_{111}, \dots, \nu^2_{K3p_3}, \xi_{1}, \dots, \xi_{K}) \nonumber \\
&\propto (\tau^2)^{-\frac{N}{2}} \exp \left\{-\frac{1}{2\tau^2}(y^* - X^*_{B_d}\tilde{\bm \beta}_{d})^{\T}(y^* - X^*_{B_d}\tilde{\bm \beta}_{d})\right\} \prod_{g=1}^K (\tau^2)^{-\frac{p_d}{2}} \exp\left\{-\frac{1}{2\tau^2}\bm \beta_{Kd}^{\T}(\delta_K^2 \Lambda_{Kd})^{-1}\bm \beta_{Kd} \right\} \nonumber \\
&\propto (\tau^2)^{-\frac{N + Kp_d}{2}} \exp \left\{-\frac{1}{2\tau^2}(y^* - X^*_{B_d}\tilde{\bm \beta}_{d})^{\T}(y^* - X^*_{B_d}\tilde{\bm \beta}_{d})\right\} \nonumber \\ &\exp\left[-\frac{1}{2\tau^2}\tilde{\bm \beta}_{d}^{\T}\{\text{diag}(\delta_1^2 \Lambda_{1d}, \dots, \delta_K^2 \Lambda_{Kd})\}^{-1}\tilde{\bm \beta}_{d} \right] \nonumber \\
& \equiv \mathcal{N}(\mu_{B_d}, \tau^2\Sigma_{B_d}),
\end{align}
where $\Sigma_{B_d} = \left[ X_{B_d}^{*\T}X_{B_d}^{*} + \{\text{diag}(\delta_{1}^{2} \Lambda_{1d}, \dots, \delta_{K}^{2} \Lambda_{Kd})\}^{-1} \right]^{-1}$ is a ${Kp_d \times Kp_d}$ positive definite matrix and $\Lambda_{gd} = \text{diag}(\lambda^{2}_{gd1}, \dots, \lambda^{2}_{gdp_d})$ is a ${p_d \times p_d}$ diagonal matrix for $g = 1, \dots, K$ and $d = 1, 2, 3$.

The conditional posterior density of the squared local shrinkage parameter $\lambda^2_{gdj}$ is
\begin{align}\label{eq:cond_post_lambdagdj}
    f(\lambda_{gdj}^2 &\mid \bm{\beta}, \tau^2, \gamma_1, \gamma_2, \gamma_3, \delta_1, \dots, \delta_K, \nonumber \\ &\nu^2_{111}, \dots, \nu^2_{K3p_3}, \xi_{1}, \dots, \xi_{K}) \nonumber\\
    &\propto \frac{1}{(\mid \delta_g^2 \Lambda_{gd} \mid)^{1/2}}\exp \left\{ -\frac{1}{2\tau^2}\bm \beta_{gd}^{\T}(\delta_{g}^2 \Lambda_{gd})^{-1}\bm \beta_{gd}\right\} (\lambda_{gdj}^2)^{-\frac{1}{2}-1}\exp \left( -\frac{1/\nu_{gdj}}{\lambda_{gdj}^2}\right) \nonumber \\
    &\propto \left(\prod_{j=1}^{p_d}\lambda_{gdj}^2\right)^{-\frac{1}{2}} \exp \left\{-\frac{1}{2\tau^2 \delta_g^2}\sum_{j=1}^{p_d} \frac{\beta_{dj}^{2(g)}}{\lambda_{gdj}^2} \right\} (\lambda_{gdj}^2)^{-\frac{1}{2}-1}\exp \left( -\frac{1/\nu_{gdj}}{\lambda_{gdj}^2}\right) \nonumber \\
    &\propto (\lambda_{gdj}^2)^{-1 -1} \exp \left\{ -\frac{1}{\lambda_{gdj}^2} \left(\frac{1}{\nu_{gdj}} + \frac{\beta_{dj}^{2(g)}}{2\tau^2\delta_g^2} \right)\right\} \nonumber \\
    &\equiv \mathcal{IG}\left(1, \frac{1}{\nu_{gdj}} + \frac{\beta_{dj}^{2(g)}}{2\tau^2\delta_g^2}\right), \ g = 1, \dots, K, \ d = 1, 2, 3, \ j = 1, \dots, p_d.
\end{align}

The conditional posterior density of the squared global shrinkage parameter $\delta_g^2$ corresponding to the $g$-th rank-1 component in CP decomposition of $\mathcal{B}$ is
\begin{align}\label{eq:cond_post_deltag}
    f(\delta_g^2 &\mid \bm \beta, \tau^2, \gamma_1, \gamma_2, \gamma_3, \lambda^2_{111}, \dots, \lambda^2_{K3p_3},\nonumber\\ &\nu^2_{111}, \dots, \nu^2_{K3p_3}, \xi_{1}, \dots, \xi_{K}) \nonumber\\
    &\propto \prod_{d=1}^3 \frac{1}{(\mid \delta_g^2 \Lambda_{gd} \mid)^{1/2}}\exp \left\{ -\frac{1}{2\tau^2}\bm \beta_{gd}^{\T}(\delta_g^2 \Lambda_{gd})^{-1}\bm \beta_{gd}\right\} (\delta_g^2)^{-\frac{1}{2}-1}\exp \left( -\frac{1/\xi_{g}}{\delta_g^2}\right) \nonumber \\
    &\propto \prod_{d=1}^3(\delta_g^2)^{-\frac{p_d}{2}} \exp \left\{-\frac{1}{2\tau^2 \delta_g^2}\sum_{j=1}^{p_d} \frac{\beta_{dj}^{2(g)}}{\lambda_{gdj}^2} \right\} (\delta_g^2)^{-\frac{1}{2}-1}\exp \left( -\frac{1/\xi_{g}}{\delta_g^2}\right) \nonumber \\
    &\propto (\delta_g^2)^{-\frac{1 + \sum_{d=1}^3 p_d}{2} - 1} \exp \left\{-\frac{1}{\delta_g^2} \left( \frac{1}{\xi_g} + \frac{1}{2\tau^2}\sum_{d=1}^3 \sum_{j = 1}^{p_d}\frac{\beta_{dj}^{2(g)}}{\lambda_{gdj}^2}\right)\right\} \nonumber \\
    & \equiv \mathcal{IG}\left(\frac{1 + \sum_{d=1}^{3}p_d}{2}, \frac{1}{\xi_g} + \frac{1}{2\tau^2}\sum_{d=1}^3 \sum_{j = 1}^{p_d}\frac{\beta_{dj}^{2(g)}}{\lambda_{gdj}^2}\right), \ g = 1, \dots, K.
\end{align}

The conditional posterior density of the error variance $\tau^2$ is
\begin{align}
    f(\tau^2 &\mid \bm\beta, \gamma_1, \gamma_2, \gamma_3, \lambda^2_{111}, \dots, \lambda^2_{K3p_3}, \nonumber \\
    & \delta_1, \dots, \delta_K, \nu^2_{111}, \dots, \nu^2_{K3p_3}, \xi_{1}, \dots, \xi_{K}) \nonumber\\
    & \propto (\tau^2)^{-\frac{N}{2}} \exp \left\{-\frac{1}{2\tau^2}(y^* - X^*\bm \beta)^{\T}(y^* - X^*\bm \beta)\right\} \nonumber \\
    &\prod_{g=1}^K \prod_{d=1}^3 \frac{1}{(\mid \tau^2 \delta_g^2 \Lambda_{gd} \mid)^{1/2}}\exp \left\{ -\frac{1}{2\tau^2}\bm \beta_{gd}^{\T}(\delta_g^2 \Lambda_{gd})^{-1}\bm \beta_{gd}\right\} (\tau^2)^{-a_0 - 1} \exp{\left(-\frac{b_0}{\tau^2} \right)} \nonumber \\
    & \propto (\tau^2)^{-a_0 - (N + K \sum_{d=1}^3 p_d)/2 - 1}\exp \left[-\frac{1}{\tau^2}\left\{(y^* - X^*\bm \beta)^{\T}(y^* - X^*\bm \beta)/2 + \sum_{g=1}^K \sum_{d=1}^3 \frac{\beta_{dj}^{2(g)}}{2\delta_g^2 \lambda_{gdj}^2}\right\} + b_0\right] \nonumber \\
    & \equiv \mathcal{IG}\left( a_0 + \frac{N + K \sum_{d=1}^3 p_d}{2}, b_0 + \frac{1}{2}\left\{(y^* - X^*\bm \beta)^{\T}(y^* - X^*\bm \beta) + \sum_{g=1}^K \sum_{d=1}^3 \frac{\beta_{dj}^{2(g)}}{\delta_g^2 \lambda_{gdj}^2}\right\} \right)
\end{align}

Finally, we outline the full conditional densities of the auxiliary variables $\xi_g$ and $\nu_{gdj} (g = 1, \dots, K; d = 1, 2, 3; j = 1, \dots, p_d)$ as follows:
\begin{align}
    f(\xi_g &\mid \bm{\beta}, \tau^2, \gamma_1, \gamma_2, \gamma_3, \delta_1, \dots, \delta_K, \lambda^2_{111}, \dots, \lambda^2_{K3p_3},\nonumber\\ &\nu^2_{111}, \dots, \nu^2_{K3p_3}, \xi_1, \dots, \xi_{g-1}, \xi_{g+1}, \dots, \xi_K) \nonumber \\
    &\propto (\xi_g)^{-\frac{1}{2}}(\delta_g^2)^{-\frac{1}{2} - 1} \exp\left(-\frac{1/\xi_g}{\delta_g^2} \right) (\xi_g)^{-\frac{1}{2} - 1}\exp\left( -\frac{1}{\xi_g}\right)\nonumber \\
    &\propto (\xi_g)^{-1 -1} \exp \left\{ -\frac{1}{\xi_g} \left( 1 + \frac{1}{\delta_g^2}\right)\right\} \nonumber \\
    &\equiv \mathcal{IG}\left(1, 1 + \frac{1}{\delta_g^2} \right),
\end{align}
and
\begin{align}
    f(\nu_{gdj} &\mid \bm{\beta}, \tau^2, \gamma_1, \gamma_2, \gamma_3, \nu_{g'd'j'}, \nonumber \\
    &\lambda^2_{111}, \dots, \lambda^2_{K3p_3}, \delta_1, \dots, \delta_K, \xi_1, \dots, \xi_K) \nonumber \\
    &\propto (\nu_{gdj})^{-\frac{1}{2}}(\lambda_{gdj}^2)^{-\frac{1}{2} - 1} \exp\left(-\frac{1/\nu_{gdj}}{\lambda_{gdj}^2} \right) (\nu_{gdj})^{-\frac{1}{2} - 1}\exp\left( -\frac{1}{\nu_{gdj}}\right)\nonumber \\
    &\propto (\nu_{gdj})^{-1 -1} \exp \left\{ -\frac{1}{\nu_{gdj}} \left( 1 + \frac{1}{\lambda_{gdj}^2}\right)\right\} \nonumber \\
    &\equiv \mathcal{IG}\left(1, 1 + \frac{1}{\lambda_{gdj}^2} \right), g' \neq g, d' \neq d, j' \neq j.
\end{align}

\end{document}